%% file: s_csl.tex
\documentclass[a4paper,UKenglish,cleveref, autoref, thm-restate]{lipics-v2021}
\newif\iflong
\longtrue  % long version
\newcommand{\longonly}[1]{\iflong #1\fi}

\iflong
\hideLIPIcs  %uncomment to remove references to LIPIcs series (logo, DOI, ...), e.g. when preparing a pre-final version to be uploaded to arXiv or another public repository
\fi

\title{ For Generalised Algebraic Theories, Two Sorts Are Enough} %TODO Please add

\titlerunning{For GATs, Two Sorts Are Enough} %TODO optional, please use if title is longer than one line

\author{Samy Avrillon}{ENS Lyon, France}{samy.avrillon@ens-lyon.fr}{}{}

\author{Ambrus Kaposi}{Eötvös Loránd University, Hungary}{akaposi@inf.elte.hu}{0000-0001-9897-8936}{}

\author{Ambroise Lafont}{LIX, CNRS, Inria, École polytechnique, Institut Polytechnique de Paris, France}{lafont@lix.polytechnique.fr}{0000-0002-9299-641X}{}

\author{Niyousha Najmaei}{LIX, CNRS, Inria, École polytechnique, Institut Polytechnique de Paris, France}{najmaei@lix.polytechnique.fr}{0009-0001-2892-5537}{}

\author{Johann Rosain}{Nantes Université, École Centrale Nantes, CNRS, INRIA, LS2N, UMR 6004, Nantes, France}{rosain@lix.polytechnique.fr}
{0000-0003-1719-2654}{}

\authorrunning{S. Avrillon, A. Kaposi, A. Lafont, N. Najmaei, J. Rosain} %TODO mandatory. First: Use abbreviated first/middle names. Second (only in severe cases): Use first author plus 'et al.'

\Copyright{Samy Avrillon, Ambrus Kaposi, Ambroise Lafont, Niyousha Najmaei and Johann Rosain} %TODO mandatory, please use full first names. LIPIcs license is "CC-BY";  http://creativecommons.org/licenses/by/3.0/

 %TODO mandatory; please add comma-separated list of keywords

\category{} %optional, e.g. invited paper
\iflong
\else
\relatedversiondetails[cite=avrillon:hal-05485795]{Full Version}{https://hal.science/hal-05485795v2} %linktext and cite are optional
\fi

\funding{
 Funded in part by government funding managed by the French National Research Agency under the France 2030 programme, reference ANR-22-EXES-0013.
 Funded in part by the European Union (ERC, HOTT, 101170308). Views and opinions expressed are however those of the authors only and do not necessarily reflect those of the European Union
or the European Research Council. Neither the European Union nor the granting authority can be held responsible for them.
}%optional, to capture a funding statement, which applies to all authors. Please enter author specific funding statements as fifth argument of the \author macro.

\iflong
\nolinenumbers %uncomment to disable line numbering
\fi

\EventEditors{John Q. Open and Joan R. Access}
\EventNoEds{2}
\EventLongTitle{42nd Conference on Very Important Topics (CVIT 2016)}
\EventShortTitle{CVIT 2016}
\EventAcronym{CVIT}
\EventYear{2016}
\EventDate{December 24--27, 2016}
\EventLocation{Little Whinging, United Kingdom}
\EventLogo{}
\SeriesVolume{42}
\ArticleNo{23}
\usepackage [xcolor,hyperref,notion,electronic]{knowledge}
\knowledgestyle*{intro notion}{color=black,emphasize}
\knowledgestyle*{notion}{color}
\input{knowledge.txt}
\usepackage{color}
\usepackage{tabularx,colortbl}
\usepackage{amsfonts,amsmath,amsthm}
\usepackage{enumitem}
\usepackage{mdframed}
\usepackage{mathpartir}
\usepackage{bbm}
\usepackage{calc}
\usepackage{tikz-cd}
\tikzset{%
symbol/.style={%
draw=none,
every to/.append style={%
edge node={node [sloped, allow upside down, auto=false]{$#1$}}}
}
}
\tikzset{pullback/.style={minimum size=1.2ex,path picture={
\draw[opacity=1,black,-,#1] (-0.5ex,-0.5ex) -- (0.5ex,-0.5ex) -- (0.5ex,0.5ex);%
}}}
\usepackage{mathtools}
\usepackage{hyperref}
\usepackage{float}
\usepackage{proof}
\usepackage{todonotes}
\usepackage{quiver}
\usepackage{yade}

\usepackage{tcolorbox}
\usepackage[graphicx]{realboxes}
\usepackage{adjustbox}
\usepackage{rotating}
\usepackage{stmaryrd}
\usepackage{xparse}
\usepackage{thmtools}
\usepackage{thm-restate}
\usepackage{xspace}
\usepackage{ebutf8}

\newcommand{\cat}[1]{\ensuremath{\mathbf{#1}}}
\newcommand{\Set}{\cat{Set}}

\newcommand{\Fam}{\ensuremath{\cat{Fam}}}

\newcommand{\ra}{\rightarrow}

\newcommand{\id }{\ensuremath{\mathsf{Id}}}

\newcommand{\op}{\ensuremath{\mathsf{op}}}

\newcommand{\iso}{\cong}

\newcommand{\Hom}{\mathrm{Hom}}

\newcommand{\XX}{\underline{X}}
\newcommand{\YY}{\underline{Y}}
\newcommand{\pp}{\underline{p}}
\newcommand{\Pu}{\underline{P}}
\newcommand{\familyGAT}{\text{ family GAT}}

\knowledgenewrobustcmd{\PtdSet}{\cmdkl{𝐏𝐭𝐝𝐒𝐞𝐭}}
\knowledgenewrobustcmd{\FamG}{\cmdkl{\Fam}_{\mathcal{G}}}
\knowledgenewrobustcmd{\FamS}{\cmdkl{\Fam}}
\knowledgenewrobustcmd{\iGAT}{\cmdkl{e}}
\knowledgenewrobustcmd{\Lex}[1]{\cmdkl{\mathsf{Lex}}(#1)}
\knowledgenewrobustcmd{\TmSet}[1]{\cmdkl{\mathsf{Tm}}(#1)}
\knowledgenewrobustcmd{\TmSetF}[2]{\cmdkl{\mathsf{Tm}}_{#1}(#2)}

\knowledgenewrobustcmd{\TmFam}[1]{\cmdkl{\mathsf{TmFam}}(#1)}
\knowledgenewrobustcmd{\TmFamF}[2]{\cmdkl{\mathcal{F}}_{#1}(#2)}
\knowledgenewrobustcmd{\TmElFam}[1]{\cmdkl{\mathcal{E}}({#1})}
\knowledgenewrobustcmd{\TWFam}[1]{\cmdkl{\mathsf{TFam}}(#1)}
\knowledgenewrobustcmd{\UGAT}[1]{\cmdkl{\U_{#1}}}
\knowledgenewrobustcmd{\UGATw}[1]{\cmdkl{\U_{#1}}}
\knowledgenewrobustcmd{\ElMorTw}[1]{\cmdkl{\El\,#1}}
\knowledgenewrobustcmd{\ElMor}[2]{\cmdkl{\El_{#1}\,#2}}
\newcommand{\TmU}[1]{\TmSet{\UGAT{#1}}}
\newcommand{\TmEl}[2]{\TmSet{\ElMor{#1}{#2}}}
\knowledgenewrobustcmd{\Lexp}{\cmdkl{\cat{Lexp}}}
\knowledgenewrobustcmd{\TW}{\cmdkl{T}}
\knowledgenewrobustcmd{\TWb}{\cmdkl{ \mathbb{T}}} 
\knowledgenewrobustcmd{\coref}[1]{\cmdkl{R}_{#1}}
\knowledgenewrobustcmd\FinGat{\cmdkl{\mathcal{G}}}
\knowledgenewrobustcmd\GatCat{\cmdkl{\cat{Gat}}}
\knowledgenewrobustcmd\pGAT{\cmdkl{p_\mathcal{G}}}
\knowledgenewrobustcmd\XGAT{\cmdkl{X_\mathcal{G}}}
\knowledgenewrobustcmd\YGAT{\cmdkl{Y_\mathcal{G}}}
\knowledgenewrobustcmd\CAT{\cmdkl{\cat{CAT}}}
\knowledgenewrobustcmd\model[1]{\cmdkl{⟦}{#1}\cmdkl{⟧}}
\knowledgenewrobustcmd\FamCat[1]{\cmdkl{\text{Fam}}(#1)}
\knowledgenewrobustcmd\GatFam[1]{{#1}_{\cmdkl{\text{Fam}}}}
\knowledgenewrobustcmd\FamFunctor[1]{\cmdkl{\mathcal{F}_{#1}}}
\knowledgenewrobustcmd\MFunctor[1]{#1\cmdkl{/PX}}
\knowledgenewrobustcmd\modelTerm[2]{\cmdkl{[}{#1}\cmdkl{]}({#2})}
\knowledgenewrobustcmd\yo{\cmdkl{𝐲}}
\knowledgenewrobustcmd\M[1]{\cmdkl{M}_{#1}}
\knowledgenewrobustcmd\Z[1]{\cmdkl{Z}_{#1}}
\knowledgenewrobustcmd\tw[1]{\cmdkl{τ}_{#1}}
\knowledgenewrobustcmd\bang{\cmdkl{\mathop{!}}}

\KnowledgeNewDocumentCommand{\colaxSlice}{ o o }{%
  \IfNoValueTF{#1}
    {{\ensuremath{\cmdkl{\cat{CAT}\sslash\Fam}}}}
	{\IfNoValueTF{#2}
	{\ensuremath{#1\sslash\Fam}}
	{\ensuremath{#1\sslash #2}}}
}

\KnowledgeNewDocumentCommand{\colaxSliceCart}{ o o }{%
  \IfNoValueTF{#1}
    {\ensuremath{\cmdkl{\CAT\sslash_{\mathsf{c}}\Fam}}}
    {
	\IfNoValueTF{#2}	
	{\ensuremath{#1\sslash_{\mathsf{c}}\Fam}}
	{\ensuremath{#1\sslash_{\mathsf{c}}#2}}}
}

\newcommand{\cartCommaFam}[1]{\ensuremath{#1
  {\slash_{\mathsf{c}}\Fam}}}

\newcommand{\dom}{\text{dom}}
\newcommand{\cod}{\text{cod}}

\newcommand{\CC}{\ensuremath{\mathbb{C}}}
\newcommand{\DD}{\ensuremath{\mathbb{D}}}

\newcommand{\PP}{\ensuremath{\mathbb{P}}}

\newcommand{\Con}{\ensuremath{\mathsf{Con}}}

\usepackage{ifthen}
\NewDocumentCommand\pullbackFunctor{
  m % Name of the morphism
}{\ensuremath{#1^{*}}}
\NewDocumentCommand\postcompFunctor{
  m % Name of the morphism
}{\ensuremath{#1_{!}}}
\NewDocumentCommand\forgetSliceFunctor{
  m % Name of the object
}{\ensuremath{#1_{!}}}
\NewDocumentCommand\fiberProductFunctor{
  m % Codomain of the arrow
  m % Domain of the arrow
}{\ensuremath{(- \times_{#1} #2)}}
\definecolor{dkgreen}{rgb}{0,0.6,0}
\definecolor{ltblue}{rgb}{0,0.4,0.4}
\definecolor{dkviolet}{rgb}{0.3,0,0.5}
\definecolor{dkblue}{RGB}{35, 38, 176}

\NewCommandCopy{\oldGamma}{\Gamma}\renewcommand{\Gamma}{{\mathit{\oldGamma}}}
\NewCommandCopy{\oldDelta}{\Delta}\renewcommand{\Delta}{{\mathit{\oldDelta}}}
\NewCommandCopy{\oldTheta}{\Theta}\renewcommand{\Theta}{{\mathit{\oldTheta}}}
\newcommand{\GVar}[1]{\mathbf{#1}}
\newcommand{\GatSet}{\mathsf{Set}}

\newcommand{\Ty}{\mathsf{Ty}}

\newcommand{\blank}{\mathord{\hspace{1pt}\text{--}\hspace{1pt}}} %from the book

\newcommand{\U}{\mathsf{U}}
\newcommand{\El}{\mathsf{El}}

\newcommand{\p}{\mathsf{p}}

\newcommand{\appendixOrNot}[2]{#1}

\definecolor{RocqOrange}{RGB}{194,99,0}
\NewDocumentCommand{\proofLink}{
  m % Name of the file
  o % Target definition
  m % Prefix (inside the link)
  m % Prefix (displayed)
}{{\color{RocqOrange}\href{./html/#3.#1.html\IfValueT{#2}{\##2}}{#4.#1}}}
\NewDocumentCommand{\proofLinkCT}{
  m % Name of the file
  o % Target definition
}{\proofLink{#1}[#2]{CategoryTheory}{CT}}
\NewDocumentCommand{\proofLinkTS}{
  m % Name of the file
  o % Target definition
}{\proofLink{#1}[#2]{TwoSortification}{TS}}

\theoremstyle{definition}
\newtheorem{notation}[theorem]{Notation}
\crefname{notation}{Notation}{Notations}

\begin{document}

\maketitle

\begin{abstract}

  Generalised algebraic theories (GATs) allow multiple sorts indexed
  over each other. Categories have two sorts,
  objects and morphisms, and the latter are double-indexed over the
  former. Martin-Löf type theory has four sorts: contexts,
  substitutions, types and terms. For example, types are indexed over
  contexts, and terms are indexed over both contexts and types. In
  this paper we show that any GAT can be reduced to a GAT involving only two sorts:
  a first one without dependencies, and another one indexed over the former.

  We show that a model of the reduced GAT is equivalently
  given by a pair consisting of a model of the original GAT, which we call 
  the core model,
   and a "cartesian" morphism 
  from the family of sets induced by this core model.
 
  In a type-theoretic
  metatheory, the initial model of a GAT is called a quotient
  inductive-inductive type (QIIT). The above characterisation 
  justifies soundness of the reduction, in the sense that 
  the QIIT of the original GAT can be simulated from the 
  QIIT of the reduced GAT, by taking its core model.

  A particular case of our reduction
  is the well-known method of reducing mutual inductive types to a
  single indexed family.  
  Other instances of this reduction can be found in the literature to 
  simulate sophisticated QIITs
  in Cubical Agda.
  Indeed, the reduced GAT is simpler than the original GAT in the following
  aspects: it does not have sort equalities; it does not have
  interleaved sorts and operations; if the original GAT did not have
  interleaved sorts and operations, then the reduced GAT won't have
  operations interleaved between different sorts. 
  
  %  without resorting to forward
  % declarations. 
  
  Our approach is semantic in that it does not
  rely on a syntactic description of GATs, but instead, on Uemura's
  bi-initial characterisation of the category of (finite) GATs
  in the 2-category of finitely complete categories with a
  chosen exponentiable morphism.
\end{abstract}
\maketitle

\section{Introduction}

\emph{Generalised algebraic theories} (GATs), introduced by Cartmell~\cite{DBLP:journals/apal/Cartmell86}, extend 
the notion of algebraic theories by allowing the specification of multi-sorted algebraic structures where sorts can be indexed over each other. For
example, the GAT of transitive graphs has two sorts, vertices and edges, 
and the sort of edges is double-indexed over the sort of vertices.
Syntactically, it consists of two sort declarations  $𝐕 :\GatSet,
𝐄:V\ra V\ra\GatSet$, and an operation $𝐓 : ∏ v₁, v₂, v₃ : V. E\, v₁\, v₂ → E\, v₂\, v₃ → E\, v₁\, v₃$.
% , and is then completed with 
% various operations (composition, identity morphisms), and equations (associativity and neutrality). 
The category of models of this GAT is precisely the category of transitive graphs.
%, \blank\circ\blank,\mathsf{id}$
Another example is the GAT of
Martin-Löf type theory, which has four sorts: contexts,
substitutions, types and terms, with various indexing between them.
This GAT involves the sort declarations $\GVar{Con}\colon\GatSet,
\GVar{Sub}\colon\mathsf{Con}\ra \mathsf{Con}\ra\GatSet,
\GVar{Ty}\colon\mathsf{Con}\ra \GatSet,
\GVar{Tm}\colon\Pi \Gamma\colon\mathsf{Con}.\,\Pi A\colon\mathsf{Ty}\,\Gamma.\,\GatSet$.

The present work shows that any GAT can be reduced to a \emph{family GAT},
i.e. a GAT which, essentially, starts with $\GVar{U}:\GatSet, \GVar{El}:\U\ra\GatSet$ and does not involve any other sort declaration.
We call this reduction \emph{two-sortification}. The idea of the translation is simple: 
in the original GAT, replace $\GatSet$ with $\U$, and insert $\El$ in front of any used sort;
  finally prepend the result with $\GVar{U}:\GatSet,\GVar{El}:\U\ra\GatSet$.
For example, applying the translation to the GAT of transitive graphs yields $\GVar{U}:\GatSet,\GVar{El}:\U\ra\GatSet,
\GVar{V}:\U,\GVar{E}:\El\,V\ra \El\,V\ra \U, 
𝐓 : ∏ v₁, v₂, v₃ : \El\,V. \El(E\, v₁\, v₂) → \El(E\, v₂\, v₃) → \El(E\, v₁\, v₃)
$.
% In fact, the main challenge is not to define the translation formally, but to relate
% the models of the original and reduced GATs.
\subsection{Motivation}

Different classes of inductive types correspond to different classes
of algebraic theories. For example, simple inductive types are initial
models (algebras) of algebraic theories without equations; mutual
inductive types are initial models of equation-free multi-sorted
algebraic theories; quotient-inductive-inductive types (QIITs) are
initial models of GATs; inductive-inductive types (IITs) are initial
models of equation-free GATs. Understanding the relationship between
different classes of algebraic theories allows us to reduce between
different classes of inductive types. Reducing between different
classes of inductive types is important for implementation and
semantics.

\subparagraph{Implementation.}
Having a generic reduction from one class of inductive types to
another one allows us to program with simpler inductive types. For example, Rocq does not
support IITs, but they can be recovered via their reduction to indexed
inductive types \cite{DBLP:conf/types/KaposiKL19}. Cubical Agda supports equality constructors,
but not general sort equalities. 
% Using two-sortification we can
% reduce a QIIT to one which is supported by (Cubical) Agda,
% without resorting to \emph{forward declarations}~\cite{agda-manual-mutual}.

We give below a few concrete examples of use of two-sortification in this vein.
\begin{itemize}
\item Interleaved constructors: Martin-Löf type theory has sorts
  $\GVar{Con} : \GatSet, \GVar{Sub} : \Con\ra\Con\ra\GatSet,
  \GVar{Ty}:\Con\ra\GatSet, \GVar{Tm}:(\Gamma:\Con)\ra\Ty\,\Gamma\ra\GatSet$. 
  % In
  % Cubical Agda, all constructors of each sort have to be declared at
  % once, it is not directly possible to declare some constructors of
  % $\GVar{Sub}$, then some constructors of $\GVar{Ty}$, then some more
  % constructors of $\GVar{Sub}$ (this can be worked around with mutual
  % definitions and forward declarations). However 
  % Martin-Löf type theory
  It is not possible to define the full GAT without alternating
  constructors of different sorts. For instance,
  the functor law
  $A[\gamma\circ\delta] = A[\gamma][\delta]$ is an equality
  constructor of $\GVar{Ty}$, and refers to the composition
  constructor\footnote{It could make sense to think of composition as an operation defined
  recursively rather than as a constructor.
   This kind of \emph{inductive-recursive} definition is not directly allowed in GATs 
  or QIITs.} $\circ$ of $\GVar{Sub}$, so it must be declared
  afterwards. However, we cannot declare all the constructors of $\GVar{Sub}$ before those of $\GVar{Ty}$
  because the substitution extension constructor
  $\blank,\blank :
  (\gamma:\GVar{Sub}(\Delta,\Gamma))\rightarrow\GVar{Tm}(\Delta,A[\gamma])\rightarrow\GVar{Sub}(\Delta,\Gamma.A)$
  refers to the substitution constructor $\blank[\blank]$ in
  $\GVar{Ty}$. In the two-sortified version of the QIIT all these
  constructors are members of $\El\,X$ for some $X$, thus the
  interleaving disappears. This reduction has been used in practice to
  implement the syntax of Martin-Löf type theory in Cubical Agda
  \cite{DBLP:journals/corr/abs-2509-14988}.
\item Interleaving sorts and constructors: in the theory of System
  F$_\omega$, types are indexed by kinds and terms are indexed by
  types at kind ${*}$; thus, the constructor ${*}$ has to be declared
  before the sort of terms. The
  two-sortified version of System F$_\omega$ has only two sorts
  $\GVar{U}$ and $\GVar{El}$. As a result, two-sortification
  transforms interleaved sorts and constructors into interleaved
  constructors of different sorts. Above, we explained how
  two-sortification removes interleaved constructors; hence, by
  applying two-sortification once more, we can also remove interleaved
  constructors of different sorts.
\item Sort equations: a Russell universe in Martin-Löf type theory can
  be declared\footnote{Additional equations should ensure that
  $\GVar{Univ}$ is compatible with substitution. Moreover, the type theory specified in this way actually suffers from Russell's paradox. This can be resolved by
  introducing a full hierarchy of universes~\cite[Section 4.3]{DBLP:journals/corr/abs-1902-08848}.} as a type $\GVar{Univ} : ∏ Γ : \GVar{Con}.
  \GVar{Ty}\, Γ$ together with a sort equation
  $\GVar{Tm}\,Γ\,(\GVar{Univ}\, Γ) =_{\GatSet} \GVar{Ty}\,
  Γ$.
  %  Although Cubical Agda allows equality constructors, it does not
  % directly allow \emph{sort} equality constructors. 
  % Again, two-sortification
  % removes this problem,
  After two-sortification, the sort equation becomes the equality constructor
  $\GVar{Tm}\,Γ\,(\GVar{Univ}\, Γ) =_{\U} \GVar{Ty}\, Γ$. This
  technique has been used, for example, by Altenkirch and
  Scoccola~\cite{DBLP:conf/lics/AltenkirchS20} to define the integers
  as a pointed type with a sort equation.
\end{itemize}

\subparagraph{Semantics.}
Reduction to simpler classes of inductive types is also important for
semantics: if we want to show that a model of type theory supports
some class of inductive types, it is enough to show that it supports
the class that it can be reduced to. For this purpose, indexed
inductive types were reduced to non-indexed ones
\cite{10.1007/978-3-540-24849-1_14}, nested to non-nested ones
\cite{10.1007/978-3-540-27836-8_8}, small inductive-recursive to indexed
\cite{DBLP:conf/tlca/HancockMGMA13} ones, QITs were reduced to QW
types \cite{DBLP:journals/lmcs/FiorePS22}. Some of these reductions
work even in an intensional setting, e.g.\ indexed W-types can be
reduced to non-indexed ones \cite{IW}, and recently Hugunin showed
that all inductive types can be reduced to W-types even without
function extensionality
\cite{DBLP:conf/types/Hugunin20}. Two-sortification continues this
tradition: it reduces the wide class of QIITs to two-sorted ones
without interleaving or sort equations. Actually, two-sortification
has already been used by Sestini \cite{sestini} who reduced IITs to
indexed inductive types: he only considered two-sorted IITs for
simplicity. This simplification is justified by the current paper.

% In addition, in the semantics of QIITs (and thus GATs), infinitary
% operators or equational assumptions are incompatible with sort
% equalities; we can have one or the other \cite{kovacs_thesis}. In the
% reduced GAT, however, we have no sort equalities, but we allow
% infinitary operators and equational assumptions.

\subsection{What is a GAT?}
A GAT can be thought of as a context in a logical framework or domain-specific type theory:
this has been the guiding intuition in the examples given so far. Substitutions between such contexts
are the morphisms.
It has occasionally been argued~\cite{DBLP:journals/pacmpl/KaposiKA19,kovacs_thesis} that GATs are inconvenient to work with directly because 
of their low-level definition as raw syntax and typing judgements~\cite{DBLP:journals/apal/Cartmell86}. 
We circumvent this issue by leveraging a simple universal property of the category of finite GATs due to Uemura~\cite{UEMURA2022106991}.
Concretely, the category of finite GATs, equipped with the projection 
 $\pGAT ∶ (\GVar{U}:\GatSet, \GVar{u} : \U) → (\GVar{U}: \GatSet)$, is \emph{bi-initial} among 
finitely-complete categories equipped with an exponentiable morphism.
Exponentiability means that the pullback functor along this morphism has a right adjoint,
called 
the pushforward along that morphism.
Bi-initiality means that given any finitely complete category $\CC$ with a chosen 
exponentiable morphism $p'$, there exists a unique (up to unique isomorphism) 
finite-limit-preserving functor $F$
from the category of GATs to $\CC$ 
such that $F(\pGAT) ≅ p'$ and pushforwards along $\pGAT$ are mapped (up to isomorphism) to pushforwards along $p'$.
% The connection with a type-theoretic perspective on GATs is as follows.
%  ConcThe exponentiable morphism $\pGAT$ mentioned in the universal property
% is the projection $(\GVar{U}:\GatSet, \GVar{u} : \U) → (\GVar{U}: \GatSet)$.

 Type-theoretic perspectives on (variants of) GATs can be found in the literature~\cite{DBLP:journals/pacmpl/KaposiKA19,DBLP:conf/lics/KovacsK20,DBLP:journals/mscs/BezemCDE21}.
 These provide (stricter) initiality properties and, like Uemura's characterisation, avoid dealing with untyped syntax.
 To construct a functor from the category of GATs to some category $\CC$,
which the functor is expected to map the GAT morphism $(\GVar{U}:\GatSet, \GVar{u} : \U) → (\GVar{U}: \GatSet)$ to. 
%, and check that $p'$ is exponentiable. 
% This determines a functor 
% from GATs to $\CC$ mapping $p$ to $p'$ (up to isomorphism).
% In more traditional terms, initiality can be seen as a form of recursion principle.
% This is the main recipe that we use to define two-sortification.

As a preliminary example, we can easily reconstruct the \emph{model functor} to (large) categories that computes the category of models of a given GAT:
since we expect the exponentiable GAT morphism $(\GVar{U}:\GatSet, \GVar{u} : \U) → (\GVar{U}: \GatSet)$ to be mapped 
to the forgetful functor from the category of pointed sets to the category of sets, we simply 
choose $p'$ to be this (exponentiable) functor.
We will see that two-sortification can be defined in a similar way, by equipping 
the category of GATs sliced over the GAT of families with a suitable exponentiable morphism.

\subsection{Contributions}
\label{sec:contributions}
% \subsubsection{A Syntactic and Semantic Account of Two-sortification}
We give a formal account of two-sortification as an endofunctor on the category of (finite) GATs. Each theory $Γ$ is mapped to a family theory $\TW Γ$ equipped with a morphism to the GAT of families $\U, \El$.
We show that the mapping $Γ ↦ (\TW Γ →(\U, \El))$ actually induces a 
%fully faithful 
functor to the slice category over the GAT of families (\cref{def:two-sortification}).
%, as well as a GAT morphism  
%from $\reintro*\TW Γ$ to the theory $\U ∶ \GatSet, \El ∶ \U → \GatSet$, intuitively forgetting everything from $\reintro* \TW Γ$ but the two sorts.
% The mapping $Γ ↦ (\TW Γ →(\U, \El))$ actually induces a fully faithful functor 
% to the slice category over the theory $\U, \El$.

We justify two-sortification from the point of view of initial algebra semantics as follows. 
For each theory $Γ$, we define a GAT morphism from $\TW Γ$ to $Γ$, which we call the
\emph{core projection GAT morphism} (\cref{def:core-projection-morphism}).
This morphism induces a mapping from models of the reduced GAT to models of the original GAT, and 
we show this mapping preserves the initiality property (\cref{cor:core-projection-morphism-functor}).
In a type theory, a QIIT can be characterised as an initial model of a GAT, and thus the above result implies that the QIIT of the original GAT can be simulated from the QIIT of the reduced GAT by applying this mapping.

We attack the proof of this initiality preservation by first clarifying the exact relationship between models of the original GAT and models of the reduced GAT.
More precisely, we define, for each model $M$ of any GAT, a family 
$El_M : U_M → \Set$ of sets (\cref{cor:family-functor-gat}).
%  indexed over the underlying set $U_M$ of the sort $\GVar{U}$, and that a morphism between models induces a morphism between the corresponding families.
% For each theory $Γ$, 
% we construct a functor from its category of models to the category of families of sets.
% We then show that of the category of models of the reduced GAT is, up to isomorphism, 
% the category 
% that we now sketch.
% First, let us define the \emph{family functor} of a GAT as the composition of the coreflective left adjoint and the projection mapping a model of the reduced GAT to its underlying (set-indexed) family.
% Let $El_M ∶ U_M → 𝐒𝐞𝐭$ denote the image of a model $M$ of the original GAT by the family functor. 
Then, we show in \cref{thm:characterisation-models-two-sortification} that the category of models of the reduced GAT is (up to isomorphism) the category whose objects 
are tuples of a model $M$ of the original GAT, a family $(U',El')$, and a 
    function $f ∶ U_M → U'$ such that $El'(f(u)) = El_M(u)$ for each $u ∈ U_M$.
  %   ;
  %   \item A morphism between two objects consists of a morphism between the underlying models and a morphism between the underlying families compatible with the underlying functions.
  % \end{itemize}
  \begin{example}
  \label{ex:transitive-graphs-characterisation-reduced}
 A model of the GAT of transitive graphs $G=(V,E,T)$ induces a family $(U_G, El_G)$ defined by
$U_G = \{ * \} + V  \times  V$, $El_G(*) = V$, and $El_G(a,b) = E(a,b)$. 

According to the above description, a model of the reduced GAT  
 is equivalently given by a transitive graph $G=(V,E,T)$, a family $(U',El')$, and a 
    function $f ∶ U_G → U'$ such that $El'(f(u)) = El_G(u)$ for each $u ∈ U_G$.
  Giving such a function $f$ amounts to picking an element $V'$ in $U'$, and providing 
an operation $E':V × V → U'$. In this way, the equation reformulates as $El'(V') = V$ and 
$El'(E'(a,b))=E(a,b)$. This means  that $V$ and $E$ are redundant since they are uniquely determined from $(U',El',V',E')$. 
Finally, the transitive operation $T$ has exactly the right type to 
make $(U',El',V',E',T)$ a model of the translated GAT. This induces a bijection between 
objects of this category and the models of the translated GAT, that actually extends to an isomorphism of categories.
\end{example}

  Note that, according to this description, a model of the reduced GAT explicitly includes a model $M$ of the original GAT, which we call its \emph{core model}, as a component.
 We call the functor mapping a model to its core model the \emph{core projection functor}. It is easy to show that it preserves initiality (\cref{prop:core_projection_preserves_initial}).

 The justification of two-sortification is completed by proving that the mapping induced by the core projection GAT morphism coincides with the core projection functor (\cref{cor:core-projection-morphism-functor}), which implies that it also preserves initiality, as desired.
 
  Although we focus on finite GATs for simplicity, 
  we explain how to extend Uemura's characterisation
    and our results to cover infinite GATs in 
    \iflong Appendix~\ref{sec:infinite-gats}\else ~\cite[Appendix~D]{avrillon:hal-05485795}\fi.

\subsection{Related Work}
\subsubsection{Formal definition of GATs}
As hinted above, we heavily rely on Uemura's bi-initial characterisation of GATs~\cite{UEMURA2022106991}. Garner~\cite{GARNER20151885} provides an alternative categorical definition of GATs, as algebras for a certain monad on a category of presheaves.

We also already mentioned related work with a type-theoretic perspective on GATs, studying variants of Cartmell's original notion.
Kovács's thesis~\cite{kovacs_thesis} develops the theory of finitary QIITs with sort equations (Chapter 4, see also \cite{DBLP:journals/pacmpl/KaposiKA19}), and
the theory of infinitary QIITs without sort equations (Chapter 5, see also Kov\'{a}cs and Kaposi~\cite{DBLP:conf/lics/KovacsK20}). Infinitary QIITs allow operations with infinitely
many arguments, a feature absent from GATs. 
%These works feature signatures that can be infinite but in some restricted way, e.g., the infinite GAT of semi-simplicial types above mentioned is out of scope.
Finally, let us mention that we follow their convention regarding the direction of the morphisms of GATs, which differs from Cartmell (in particular, a morphism of GATs induces a functor between the categories of models in the same direction).

\subsubsection{Initial Models of GATs}
In the spirit of initial algebra semantics~\cite{InitialSemantics},
an initial model of a GAT can be thought of as the categorical axiomatisation of the corresponding inductive type with its full dependent recursion principle
(see, e.g., \cite[Section 7.3]{DBLP:journals/pacmpl/KaposiKA19}).

The above cited works about signatures for QIITs provide proofs of existence of initial models for their respective notions of GATs.
Cartmell's original paper on GATs~\cite{DBLP:journals/apal/Cartmell86}
states (without proof) that the image of any GAT morphism by the model functor has a left adjoint, thus providing an initial model of any GAT by considering GAT morphisms to the empty theory. 
This is also  a consequence of 
   Cartmell's translation from GATs to essentially algebraic theories
  (although the translation of sort equalities is not much detailed).

  Our work fits in the line of research consisting in constructing 
  initial models by reducing some notion of GATs to a simpler one. 
  In this respect, 
  two-sortification can be understood as an extension of the reduction of mutual inductive types to indexed inductive types~\cite{DBLP:conf/fscd/KaposiR20}.
As explained above, it eliminates 
some features of GATs such as sort equations. Moreover, we show that the initial model of the original GAT can be recovered from the initial model of the reduced one.
In fact, Sestini~\cite{sestini} explicitly conjectured 
the validity of two-sortification for his variant of GATs
to justify focusing on the two-sorted ones, before 
reducing them further to indexed-inductive types, in a type-metatheoretic setting.
% Nonetheless, our proof of initiality is direct and does not rely on the reduction to two-sorted GATs.

 \subsection{Synopsis}
 In \cref{sec:prelim}, we recall some preliminaries. Importantly, we review Uemura's bi-initial characterisation of finite GATs.
 In \cref{sec:syntactic-translation}, we define two-sortification as a bi-initial functor from finite GATs to GATs sliced over the theory $\Fam$ of families.
 In \cref{sec:semantic-translation}, we show our main result:
 a characterisation of the models of the translated GAT. This allows us to justify this translation from the point of view of initial algebra semantics in \cref{sec:core-proj-functor}, by studying the projection 
 induced by this characterisation from models of the translated GAT to models of the original GAT.

%  from which we legitimate the coreflection between the two categories of models.
%  In \cref{sec:connecting-syntactic-semantic}, we make the connection between the syntactic and semantic translations, that is, we justify that 
%  \cref{sec:semantic-translation} is indeed constructing the 
%  category of models of the two-sortification defined in \cref{sec:syntactic-translation}.
% In \cref{sec:initial-models-gats}, we prove that each finite GAT has an initial model. The arguments developed there are used in 
% \cref{sec:two-sortification-fully-faithful}
% to show that the two-sortification functor is fully faithful.
% Finally, in \cref{sec:infinite-gats}, we discuss how to extend our results to infinite GATs,
% by extending the bi-initiality property of the category of finite GATs.
%  Finally, in \cref{sec:applications}, we discuss some applications of our results.

For the reader's convenience, we provide hyperlinks from occurrences of a notion to its definition.

% A few definitions and results of this paper have been formalised in \Rocq. We provide this
% formalisation as supplementary material~\cite{supplementary-material}. A detailed account
% of what has been formalised can be found in the
% \href{https://anonymous.4open.science/r/FormalizedTwoSortification/README.md}{\texttt{README}}
% file.
% , and give links (assuming the PDF file is in the same directory
% as the \verb|html| directory of the supplementary material) of formalised definitions and theorems,
%  \eg ~{\color{RocqOrange}\href{./html/index.html}{TwoSortification}}.

 \section{Preliminaries}
 \label{sec:prelim}
 This section recalls various definitions and results that we use throughout the paper, 
 most importantly the bi-initial characterisation of GATs in \cref{sec:generalised-algebraic-theories}.

 \subsection{Exponentiable morphisms}
 \label{sec:exponentiable-morphisms}
 In this subsection, we recall the definition of exponentiable morphisms in a category with finite limits, and various basic properties and definitions related to them. 
%  This part can be skipped at first reading; we refer to it later when necessary.
%  \appendixOrNot{The corresponding proofs for this subsection are available in \cref{sec:proof-lexp-property}.}
%  {}
\begin{definition}[Slice category]
  \label{def:slice}
  \AP The \intro{slice category} $ℂ/X$ of a category $ℂ$ over an object $X$ of $ℂ$ 
  is the category whose objects are morphisms $f : Y \to X$ in $ℂ$, and whose morphisms from $f : Y \to X$ to $g : Z \to X$ are morphisms $h : Y \to Z$ in $ℂ$ such that $g \circ h = f$.
\end{definition}
\begin{definition}[Exponentiable morphism]%[{\cite[Definition~2.4]{UEMURA2022106991}}]
  \label{def:exp-mor}
  \AP A \intro{Lex-category} is a category with finite limits. \AP A morphism
  $f : Y \to X$ in a \kl{Lex-category} $ℂ$ is said to be {\intro{exponentiable}} if
  the pullback functor $f^*:ℂ/X → ℂ/Y$ along $f$ has a right adjoint.  We call this right adjoint the \emph{pushforward}
  along $f$ and denote it with $f_*$.
\end{definition}
 
\begin{proposition}[{\cite[Corollary 1.2]{NIEFIELD1982147}}]
  \label{prop:exp-mor-P-def}
   A morphism $f : Y \to X$ is \kl{exponentiable} if and only if the composite
   \begin{tikzcd}[ampersand replacement=\&]
      {\CC/X} \& {\CC/Y} \& \CC
      \arrow["{f^*}", from=1-1, to=1-2]
      \arrow["\dom", from=1-2, to=1-3]
    \end{tikzcd}
   has a right adjoint, where $\dom ∶ ℂ/Y → ℂ$ denotes the functor mapping a morphism to $Y$ to its domain.
    % \[\begin{tikzcd}[ampersand replacement=\&]
    %   {\CC/X} \& {\CC/Y} \& \CC
    %   \arrow["{f^*}", from=1-1, to=1-2]
    %   \arrow["\dom", from=1-2, to=1-3]
    % \end{tikzcd}\]
\end{proposition}

\begin{notation}
  We denote the right adjoint of the composite in \cref{prop:exp-mor-P-def} by $P_f : \CC \to \CC/X$. 
  When it is clear from the context, we conflate $P_f$ with $\dom \circ P_f : \CC \to \CC$, and sometimes write $P$ for $P_f$.
\end{notation}

\begin{definition}
  \label{def:preserves-pushforwards}
  \AP Let $\CC$ and $\CC'$ be categories with finite limits, with $\CC$ equipped with an \kl{exponentiable} morphism $f : Y \to X$.
Let $F : \CC \to \CC'$ be a functor preserving finite limits\footnote{Preservation of pullbacks is actually enough for this definition to make sense.} and such that $Ff$ is \kl{exponentiable}.
  We say that $F$ \emph{\intro{preserves pushforwards} along $f$} if the canonical natural transformation below right, as the mate~\cite[§2.2]{revieweltstwocats} of the isomorphism below left, is also an isomorphism.
\begin{mathpar}
\begin{tikzcd}[ampersand replacement=\&]
	{\CC/Y} \& {\CC/X} \\
	{\CC'/FY} \& {\CC'/FX}
	\arrow["{f^*}"', to=1-1, from=1-2]
	\arrow["{F}"', from=1-1, to=2-1]
	\arrow["{F}", from=1-2, to=2-2]
	\arrow["\iso", between={0.4}{0.7}, Rightarrow, from=2-2, to=1-1]
	\arrow["{(Ff)^*}", to=2-1, from=2-2]
\end{tikzcd}
\and
\begin{tikzcd}[ampersand replacement=\&]
	{\CC/Y} \& {\CC/X} \\
	{\CC'/FY} \& {\CC'/FX}
	\arrow["{f_*}", from=1-1, to=1-2]
	\arrow["{F}"', from=1-1, to=2-1]
	\arrow["{F}", from=1-2, to=2-2]
	\arrow["\iso", between={0.4}{0.7}, Rightarrow, from=1-2, to=2-1]
	\arrow["{(Ff)_*}"', from=2-1, to=2-2]
\end{tikzcd}
\end{mathpar}
\end{definition}

\begin{lemma}[{\cite[Proposition~2.12]{UEMURA2022106991}}]
  \label{lem:preservation_pushforward}\ In the context of
  \cref{def:preserves-pushforwards}, the functor $F : \CC \to \CC'$ \kl{preserves
    pushforwards} along $f$ if and only if it commutes with the associated right adjoints
  $P_f$ and $P_{Ff}$ (up to isomorphism), i.e. the canonical natural transformation from
  $F ∘ P_f$ to $P_{Ff} ∘ F$ is an isomorphism.
  % \[\begin{tikzcd}[ampersand replacement=\&]
	% \CC \& {\CC/X} \\
	% {\CC'} \& {\CC'/FX}
	% \arrow["P_p", from=1-1, to=1-2]
	% \arrow["F"', from=1-1, to=2-1]
	% \arrow["\iso", between={0.3}{0.7}, Rightarrow, from=1-2, to=2-1]
	% \arrow["{F}", from=1-2, to=2-2]
	% \arrow["{P_{Fp}}"', from=2-1, to=2-2]
  % \end{tikzcd}\]
\end{lemma}

% We conclude this subsection with properties of \kl{exponentiable} morphisms that we will use later on. 

% \begin{lemma} %%% XXX: Link
%     \label{lemma:exp_in_slice}
%     Let $\CC$ be a category with finite limits.
%     If $f : Y \to X$ is \kl{exponentiable} in $\CC$, then it is also \kl{exponentiable} as a morphism in $\CC/Z$ of the form $g \circ f \to g$, for each $g : X \to Z$ in $\CC$.
% \end{lemma}

% \begin{lemma} %%% XXX: Link
%     \label{lem:exp-stable-pb}
%     In a category with finite limits, \kl{exponentiable} arrows are
%     closed under identities, composition and pullbacks.
% \end{lemma}

 \subsection{Generalised Algebraic Theories}
 \label{sec:generalised-algebraic-theories}
 In this section, we define GATs and introduce
 the universal properties of the category of finite GATs that we use throughout the paper.
 \begin{definition}
  The category $\AP\intro*\FinGat$ of \emph{finite generalised algebraic theories} is the category 
  of contexts and substitutions of the type theory generated by 
  a universe $\GatSet$ of types, extensional equality types, 
  and dependent products over types in $\GatSet$.
 \end{definition}
%  Here, extensional equality types means that if two terms are provably equal then 
% they are judgementally equal. 
% We will not detail this definition too much, since we mostly rely on the bi-initial characterisation 
% of $\FinGat$, but let us say a few words. 
Equality types provide, for each pair of terms $(t,u)$ of the same type $A$, a new type $t = u$ with the usual expected rules.
Dependent product over types in $\GatSet$ means that if $Γ⊢A : \GatSet$ and $B$ is a type in the context $Γ,x:A$, we can form the product type $Π_{x:A} B$ in context $Γ$ and we also have lambda-abstraction, application, as well as the usual $\beta$ and $\eta$ rules. 
A substitution from $Γ$ to $x_1 ∶ A_1, \ldots, x_n ∶ A_n$ is a list of terms $t_1, \ldots, t_n$ such that $Γ ⊢ t_i : A_i[x_j ↦ t_j ]_{1 ≤ j < i}$.

 We will not detail further this definition, as our main technical tool is the following bi-initial characterisation of GATs.

 \begin{theorem}[{\cite[Theorem~4.1]{UEMURA2022106991}}]
  \label{thm:uemura}
  \AP Let $\intro*\pGAT∶ \intro*\YGAT \to \intro*\XGAT$ be the projection morphism $(A:\GatSet, a:A) \to (A:\GatSet)$ in $\FinGat$.
  We sometimes drop the index $\mathcal{G}$ when no confusion arises.

  The category $\FinGat$, equipped with $\pGAT$, is bi-initial in the 2-category
  $\AP\intro*\Lexp$ defined as follows:
  \begin{itemize}
    \item an object is a $\Lexp$-category $ℂ$, that is, a \kl{Lex-category} $ℂ$ 
    equipped with an \kl{exponentiable} morphism $p ∶ Y → X$ in $ℂ$;
    \item a morphism, which we call a $\Lexp$-functor, from $(ℂ, p)$ to $(ℂ', p')$ is a finite-limit-preserving functor $F∶ ℂ \to ℂ'$ equipped with an isomorphism 
    $F(p) ≅ p'$  in $ℂ'^{→}$, \kl{preserving pushforwards} along $p$ in the sense of  \cref{def:preserves-pushforwards};
    \item a 2-cell between two morphisms $F$ and $F'$ is a natural transformation $α∶F → F'$ 
    such that $α_X$ and $α_Y$ 
    are the isomorphisms $F(X) ≅ X' ≅ F'(X)$ and $F(Y) ≅ Y' ≅ F'(Y)$.
  \end{itemize}
  
\end{theorem}

Bi-initiality of $\FinGat$ as stated by \cref{thm:uemura} means that given any object $(ℂ, p)$ in $\Lexp$, there exists a $\Lexp$-functor $\FinGat → ℂ$, and moreover, any pair of such functors are naturally isomorphic, for a unique isomorphism.
\begin{definition}
  We call any $\Lexp$-functor $\FinGat → (ℂ,p)$ a \emph{bi-initial ($\Lexp$-)functor}, and sometimes even call it \emph{the} bi-initial $\Lexp$ functor to $ℂ$ although it is only unique up to isomorphism.
\end{definition}

% Informally, this hints at an alternative inductive definition of finite GATs: instead of constructing them by iterative context extension by some type, we can build any of them out of the polynomial functor $Q$\todo{agree on the notation P vs Q} and finite limits.

We explain later in this section what the pushforward along the \kl{exponentiable} morphism $\pGAT$ is. For now, 
let us note that we can always strictify a bi-initial functor.
\begin{proposition}
  \label{prop:strictification}
   We say that a $\Lexp$-functor $F ∶(ℂ,Y\xrightarrow{p}X) → (ℂ',Y'\xrightarrow{p'}X')$ is \AP\intro{strict} if the isomorphism $F(p) ≅ p'$ is an identity.
   Given any $\Lexp$-category $ℂ$ with \kl{exponentiable} arrow $p ∶ Y → X$ such that $X ≠ Y$, 
   any $\Lexp$-functor from $ℂ$ is isomorphic to a strict one.
   In particular, there exists a strict morphism 
   from $\FinGat$ to any $\Lexp$-category, which we call \emph{the} strict bi-initial $\Lexp$-functor to $ℂ$,
   although it is only unique up to isomorphism.
\end{proposition}
\begin{proof}
  The functor from the walking arrow category $→$ to $ℂ$ selecting $p$ is injective on objects.
  Therefore, the restriction functor $(ℂ')^{ℂ} → (ℂ')^{→}$ is an isofibration~\cite[Theorem 2]{pseudopbk}. Therefore, we can lift the isomorphism $F(p) ≅ p'$ to an isomorphism $F ≅ F'$ for some $F'$ satisfying $F'(p) = p'$.
\end{proof}
As a first example of application of bi-initiality, we recover the functor that computes the category of models 
of a given GAT, by considering an adequate $\Lexp$-structure on categories, in which \emph{family fibrations} play an essential role.
\begin{definition}[{Family fibrations~\cite[Definition 1.2.1]{DBLP:books/daglib/0023251}}]
  Given a category $𝔻$, the \emph{family fibration}
  \AP$\intro*\FamCat{𝔻} →𝐒𝐞𝐭$ is the first projection functor from the category of (set-indexed) families of $𝔻$: an object $(U,El)$ of $\FamCat{𝔻}$ is a pair of a set $U$ and a family $El ∶ U → 𝔻$ of objects of $𝔻$ indexed by $U$.
  A morphism $(f,g)$ from $(U,El)$ to $(U',El')$ consists of a function $f ∶ U → U'$ and a family of $𝔻$-morphisms $(g_u  ∶ El(u) → El'(f(u)))_{u ∈ U}$.
    % \[
  % % YADE DIAGRAM morphism-famD.yade
  % % GENERATED LATEX
  % \input{diagrams/morphism-famD.tex}
  % % END OF GENERATED LATEX
  % \]
\end{definition}

\begin{proposition}[{The $\Lexp$-category of categories~\cite[Corollary 5.14]{Weber2007}}]
  \label{prop:CAT}
  \AP Consider the category $\intro*\CAT$ of locally small categories, equipped with 
  the (\kl{exponentiable}) projection \AP$p_{\CAT} ∶ \intro*\PtdSet → 𝐒𝐞𝐭$ from pointed sets to sets. The pullback along 
  that functor maps a functor $F : 𝔻 → 𝐒𝐞𝐭$ to its category of elements $∫F$.
  The right adjoint $P_{p_{\CAT}}$ maps a category $𝔻$ to the \emph{family fibration} $\FamCat{𝔻} →𝐒𝐞𝐭$.
\end{proposition}

\begin{definition}
  \AP By \cref{prop:strictification,prop:CAT}, bi-initiality of $\FinGat$ induces a \kl{strict} $\Lexp$-functor from $\FinGat$ to $\CAT$
  which we call the \intro{model functor}. We denote the image of a theory $Γ$ (resp. GAT morphism $σ$) by $\intro*\model{Γ}$ (resp. $\model{ σ }$). 
  We say that a category $ℂ$ \intro{is specified} by a theory $Γ$
  if $ℂ$ is isomorphic to $\model{Γ}$.
\end{definition}
% \begin{proof}
%   The same argument as in the proof of \cref{prop:strictification}, and the same argument
% \end{proof}
% \begin{proof}
%   Uniqueness follows from \cref{thm:uemura}.
%   For existence, consider an initial morphism $F$ from $\FinGat$ to $ℂ$ given by \cref{thm:uemura}.
%   We define a new functor $F' ∶ \FinGat → ℂ$ that coincides with $F$ on objects except 
%   that $F'(A ∶ \GatSet)$ is defined to be $X$ and $F'(A ∶ \GatSet, a : A)$ is defined to be $Y$.
%   We adapt the action of $F'$ on morphisms accordingly, so that it is naturally isomorphic to $F$.
%   By construction, $F'(p)$ is equal to $p'$.
%   Finally, it is straightforward to check that $F'$ preserves pushforwards along $p$ because $F$ does.
%   Therefore, $F'$ is a strict morphism from $\FinGat$ to $ℂ$ in $\Lexp$.
% \end{proof}

We are now in better position to understand what the right adjoint $P_{\pGAT}$ is, as formally described in 
\cite[Proposition 3.34]{UEMURA2022106991}. The guiding intuition is that ${ P Γ }$ \kl{specifies} $\FamCat{\model{Γ}}$, since the \kl{model functor} \kl{preserves pushforwards} along $\pGAT$. 
We provide a few examples in \cref{tab:pushforward-gats}.
%  We
% Based on this intuition, we now give some examples of what $P$ does on simple theories. 
%  let us 
% What it actually does is similar to family fibration 
% Now, let us provide some details about the right adjoint $P ∶ \FinGat → \FinGat/(A:\GatSet)$ to the pullback functor along $(A:\GatSet, a:A) \to (A:\GatSet)$. Intuitively, it maps a theory $Γ$ to the theory $(A:\GatSet, γ : A → Γ)$ of families valued in $Γ$, indexed by some $A ∶ \GatSet$.
% Of course, the type $A → Γ$ of $γ$ is ill-defined as $Γ$ is not a type but a context, but that should 
% give some intuition. 
\begin{table}
  
  \caption{Examples of GATs and their images by $P$\label{tab:pushforward-gats}}
  % In the table , the GATs on the right specify the category of families of the GATs on the left, i.e., and so they are the images under $P$ of the GATs on the left.
  \[
  \begin{array}{c|c}
    Γ & P Γ \text{ specifying families of $\model{Γ}$} \\
    \hline
    \GVar{B}:\GatSet & \GVar{A}:\GatSet, \GVar{B}:A\ra\GatSet \\
    \hline
    \GVar{B} : \GatSet, \GVar{b} : B & 
    \begin{array}{c}
    \GVar{A} : \GatSet,\\  
    \GVar{B} : A → \GatSet, \GVar{b} : ∏ a : A. B\, a
    \end{array}
      \\
    \hline
    % \begin{array}{c}
      \text{(Transitive graphs)} & \GVar{A}:\GatSet,  \\
      \GVar{V} :\GatSet, & \GVar{V} : A → \GatSet, \\
      \GVar{E}: V\ra V\ra\GatSet, & \GVar{E} : ∏ a : A. V\, a → V\, a → \GatSet, \\
      \GVar{T} : ∏ v₁, v₂, v₃ : V. & \GVar{T} : ∏ a:A. ∏ v₁, v₂, v₃ : V\, a. \\
       E\, v₁\, v₂ → E\, v₂\, v₃ → E\, v₁\, v₃ &  E\, a\, v₁\, v₂ → E\, a\, v₂\, v₃ → E\, a\, v₁\, v₃
%       \GVar{V} :\GatSet, 
%       \\
%       \GVar{E}: V\ra V\ra\GatSet, \\
%       \GVar{T} : ∏ v₁, v₂ ,v₃ : V, 
%        E\, v₁\, v₂ → E\, v₂
%     \end{array}
% &
% \begin{array}{c}
%   \GVar{A} : \GatSet, 
%   \\ \GVar{V} : A → \GatSet \\ 
%   \GVar{E} : ∏ a : A. V\, a → V\, a → \GatSet, \\
%   \GVar{T} : ∏ a:A. ∏ v₁,v₂ ,v₃ : V\, a,  \\
%    E\, a\, v₁\, v₂ → E\, a\, v₂\, v₃ → E\, a\, v₁\, v₃
% \end{array}
  \end{array}
  \]
\end{table}
% To see why this is the case, let us characterise
% . First note that $\FinGat/(A:\GatSet)$ is isomorphic to the category of GATs equipped with a term of type $\GatSet$, with substitutions between them preserving that term. In that respect, the pullback functor $\FinGat/(\GVar{A}/\GatSet) → \FinGat$ extends such a GAT with a new variable whose type is given by the term. 
% Now, the right adjoint $P$ should map a GAT $Δ$ to a GAT $P Δ$ equipped with a term 
% $Δ ⊢ K_Δ ∶ \GatSet$ such that given any $Γ ⊢ B : \GatSet$, morphisms from $Γ,b:B$ to $Δ$
% are in one-to-one correspondence with morphisms  $σ ∶Γ→ P Δ$ such that $K_Δ[σ] = B$.
% The homset bijection given by the adjunction states that a morphism from 
% $Γ,b:B $ to $Δ$ should be equivalently given by a morphism from $Γ$ to $P Δ$ over $(A:\GatSet)$.

% A morphism from a theory $Γ$ to $P Δ$ over $(A:\GatSet)$ must be equivalently given by a morphism from the pullback of $Γ$ along $(A:\GatSet, a:A) \to (A:\GatSet)$ to $Δ$.

% The following proposition is helpful in gaining intuition about \cref{thm:uemura}.
% The general definition goes as follows.
% \begin{lemma}[{\cite[Proposition 3.34]{UEMURA2022106991}}]
%   The right adjoint $P ∶ \FinGat → \FinGat/(A:\GatSet)$ to the pullback functor along 
%   $(A:\GatSet, a:A) \to (A:\GatSet)$ maps a theory $\GVar{x_1} ∶ B_1, …, \GVar{x_n} ∶ B_n$ to 
%   $\GVar{A} ∶ \GatSet, \GVar{y_1} ∶ A → B_1,  \GVar{\YY} ∶ ∏a :A,B_2[x_1 ↦ y_1\,a]…, \GVar{y_n} ∶ ∏a:A,  B_n[x_i ↦y_i\,a]$,
%   with its canonical projection to $(\GVar{A}:\GatSet)$.
% \end{lemma}

We end this section with an important result for constructing or showing 
equality of natural transformations between functors from $\FinGat$, used throughout the paper.

\begin{restatable}{theorem}{universalNaturalTransformation}
  \label{thm:universal-property-natural-transformation}
  Let $\CC$ be a $\Lexp$-category with chosen exponentiable morphism $p_\CC : Y_\CC \to X_\CC$. Let $F: \FinGat \to \CC$ be a functor preserving pullbacks along\footnote{In practice, we will only apply this result when $F$ preserves any pullback.} $\pGAT$,
  and $G : \FinGat \to \CC$ be a $\Lexp$-functor.
  Given a pullback square in $\CC$ as below left,
  there exists a unique natural transformation $\alpha_{(x,y)}$ between $F $ and $G$ as below right
  such that the given pullback is the naturality square for $\alpha_{(x,y)}$ at $\pGAT$.
  \begin{mathpar}
    \begin{tikzcd}[ampersand replacement=\&]
	{F\YGAT} \& {G \YGAT} \\
	{F\XGAT} \& {G \XGAT},
	\arrow["y",from=1-1, to=1-2]
	\arrow["F\pGAT"',from=1-1, to=2-1]
	\arrow["\lrcorner"{anchor=center, pos=0.125}, draw=none, from=1-1, to=2-2]
	\arrow["G \pGAT",from=1-2, to=2-2]
	\arrow["x"',from=2-1, to=2-2]
  \end{tikzcd}
  \and
  % YADE DIAGRAM nat-universal.yade
  % GENERATED LATEX
  \input{diagrams/nat-universal.tex}
  % END OF GENERATED LATEX
\end{mathpar}
  
\end{restatable}
\begin{proof}
 A pullback square as above left induces a $\Lexp$-structure on the comma category $F/\CC$, see \cref{sec:proof-comma-exp-structure}.
  The natural transformation $\alpha_{(x,y)}$ is then given 
  by the initial $\Lexp$-functor to $F/\CC$, exploiting 
  the universal property of the comma category~\cite[Proposition 1.6.3]{BorceuxI}.
\end{proof}
% \begin{corollary}
%   Let $F ∶ ℂ → 𝔻$ be a functor preserving pullback along $p_ℂ$ as in 
%   \cref{thm:universal-property-natural-transformation}, let $I_ℂ ∶ \FinGat → ℂ$ and $I_𝔻 ∶ \FinGat → 𝔻$ 
%   be \kl{strict} initial $\Lexp$-functors.
%   Given two natural transformations $α$ and $β$ between $F \circ I_ℂ$ and $I_𝔻$,
%   if $α_{\GVar{A}:\GatSet} = β_{\GVar{A}:\GatSet}$ and $α_{\GVar{A}:\GatSet, \GVar{a}:A} = β_{\GVar{A}:\GatSet, \GVar{a}:A}$,
%   then $α = β$.
% \end{corollary}
% Note that in \cref{thm:universal-property-natural-transformation}, we do not require $F$ to be a $\Lexp$-morphism.
% In fact, an even weaker requirement suffices: to have $F$ preserving pullbacks along $\p_\CC$ instead of all pullbacks.
% In our practical uses of the theorem however, $F$ often preserves all pullbacks.

 \section{Syntactic Translation}
  \label{sec:syntactic-translation}
  In this section, we define the two-sortification functor $\TW$ as a (bi-initial) $\Lexp$-functor mapping a GAT to its translation
  equipped with its projection to the GAT of families.

  The construction is given in \cref{sec:two-sortification-initiality}.
  In \cref{sec:two-sortification-family-gats},
  we show that the translated GAT is indeed a \emph{family GAT}, in particular that its only sorts are (essentially) $\GVar{U}:\GatSet$ and $\GVar{El}:\U\ra\GatSet$.
  % Finally, in \cref{sec:gat-morphism-two-sortification-original}, for each theory $Γ$, we construct a GAT morphism $\dom(\TW Γ) → Γ$, 
  % which we call the \emph{core projection morphism} of $Γ$: we will show later that its image 
  % by the model functor preserves the initial model.
  \subsection{Two-sortification by Bi-Initiality}
  \label{sec:two-sortification-initiality}
  The goal of this section is to define the two-sortification functor.
  \begin{notation}
  \label{not:FamGat}
  We denote the theory $P \XGAT ≅ (\U:\GatSet, \El:\U\ra\GatSet)$ of families by \AP $\intro*\FamG$.
\end{notation}
  We define two-sortification as a $\Lexp$-functor from $\FinGat$ to
  the \kl{slice category} $\FinGat/\FamG$. To that end, 
  by \cref{thm:uemura}, we merely need to equip $\FinGat/\FamG$ with a suitable $\Lexp$-structure.
  We start with a general fact regarding limits in slice categories.
\begin{proposition}
  $\FinGat/\FamG$ is a \kl{Lex-category}: the identity morphism
  on $\FamG$ is terminal in $\FinGat/\FamG$, and pullbacks are computed as in $\FinGat$.
\end{proposition}

 We now need to choose an \kl{exponentiable} morphism in the \kl{slice category} $\FinGat / \FamG$.
 The choice is dictated by what we expect the two-sortification of 
  $\pGAT∶ (\GVar{A}:\GatSet, \GVar{a}:A) \to (\GVar{A}:\GatSet)$ to be:
  the morphism $(\GVar{U}:\GatSet, \GVar{El}:\U\ra\GatSet, \GVar{A} : \U, \GVar{a} : \El\,A ) → (\GVar{U}:\GatSet,\GVar{El}:\U\ra \GatSet,\GVar{A}:\U)$.
  Let us introduce some notation to simplify the reading of such GATs.
  \begin{notation}
We denote a generalised algebraic theory $(\GVar{U}:\GatSet, \GVar{El}:\U\ra\GatSet, Γ)$ over $(\GVar{U}:\GatSet, \GVar{El}:\U\ra\GatSet)$ as
\AP $\intro*\GatFam{Γ}$, leaving the projection to $\FamG$ implicit.
For example, $\FamG$ is denoted by $\GatFam{()}$, and the theory $(\GVar{U}:\GatSet, \GVar{El}:\U\ra\GatSet, \GVar{A} : \U )$ is denoted by $\GatFam{(\GVar{A} : \U)}$.
\end{notation}
To complete the $\Lexp$-structure on $\FinGat/\FamG$, we need to prove that the
morphism from $\GatFam{(\GVar{A}:\U, \GVar{a}: \El\, A)}$ to $\GatFam{(\GVar{A}:\U)}$ is
\kl{exponentiable}. For that matter, it will be useful to re-construct it using the
available operations in $\Lexp$-categories (pullbacks and pushforwards along the
\kl{exponentiable} morphism), since then we can rely on standard results about
\kl{exponentiable} morphisms.
  \begin{proposition}
    \label{prop:p2-construction}
    Let $ℂ$ be a $\Lexp$-category with \kl{exponentiable} morphism $p ∶ Y → X$.
    We define $\pp∶\YY →\XX $ by the pullback on the right below, where
    \begin{itemize}
      \item $\XX$ is defined as the pullback on the left below;
      \item $P$ is the right adjoint $P_p$;
      \item $ \epsilon_X ∶ \XX → X$ is the counit component at $X$ of the adjunction $\dom ∘ p^* ⊣ P$.
    \end{itemize}
    \begin{equation}
      \label{eq:p2-diagram}
    {
      % \begin{mathpar}
    % YADE DIAGRAM p2.yade
    % GENERATED LATEX
    \input{diagrams/p2.tex}
    % END OF GENERATED LATEX
    % \end{mathpar}
    }
    \end{equation}  
Taking $(ℂ,p) = (\FinGat, \pGAT)$, we recover:
\begin{itemize}
  \item $\GatFam{(\GVar{A}:\U)} → \GatFam{()}$ as the morphism $\XX → PX$ on the left;
  \item  $\GatFam{(\GVar{A}:\U, \GVar{a}: \El\, A)} → \GatFam{(\GVar{A}:\U)}$ as $\pp$.
\end{itemize}
  \end{proposition}
  \begin{proof}
    $\XX$ is defined as the pullback of $Y=(\GVar{A}:\GatSet, \GVar{a}:A)$ and 
    $PX = (\GVar{U}:\GatSet, \GVar{El}:\U\ra\GatSet)$ over $(\GVar{A}:\GatSet)$, so it is, up to isomorphism, $(\GVar{U}:\GatSet,\GVar{El}:\U\ra \GatSet,\GVar{A}:\U)$.
    The claimed result follows from the counit $\XX → (\GVar{A}:\GatSet)$ selecting the 
    term $\El\, A$ in the context $\XX$.
  \end{proof}
  With this characterisation, it is straightforward to check that the morphism $\pp : \YY \to \XX$ from \cref{prop:p2-construction} is \kl{exponentiable} in $ℂ/PX$.
  \begin{proposition}
    \label{prop:p2-exponentiable}
    In the setting of \cref{prop:p2-construction}, the morphism $\pp ∶ \YY → \XX$ is \kl{exponentiable} in $ℂ/ PX$, as a morphism 
    between $\XX → PX$ defined in \cref{eq:p2-diagram} and 
    $\YY \xrightarrow{\pp} \XX → PX$, making $ℂ/PX$ a $\Lexp$-category.
  \end{proposition}
  \begin{proof}[Proof sketch]
    \kl{Exponentiable} morphisms are stable under pullbacks, and if a morphism $f:Y → X$ in a category $ℂ$ with finite limits is \kl{exponentiable}, then 
    for any morphism $X → Z$, the morphism $f$ is also \kl{exponentiable} in the slice category $ℂ/Z$.
  \end{proof}
From this we get two-sortification by bi-initiality.
\begin{definition}
  \label{def:two-sortification}
  By \cref{prop:strictification,prop:p2-exponentiable}, bi-initiality of $\FinGat$ yields a \kl{strict} $\Lexp$-functor from $\FinGat$ to $\FinGat/\FamG$ mapping $\pGAT$ to $\GatFam{(\GVar{A}:\U, \GVar{a}: \El\, A)} → \GatFam{(\GVar{A}:\U)}$. We call it the \emph{two-sortification functor} $\intro*\TW$. We often conflate $\TW Γ$ with its domain, seen as a GAT over $\FamG$.
\end{definition}

% \todo{examples of P applied to some GAT over Fam}

  \subsection{Soundness}
  \label{sec:two-sortification-family-gats}
  In this section, we show that the image of any GAT by the \kl{two-sortification functor} is (up to isomorphism) a \emph{family GAT}.
  Intuitively, a family GAT is a GAT without sort equations, and whose only sorts are $\GVar{U}:\GatSet$ and $\GVar{El}:\U\ra\GatSet$.
    Family GATs are inductively generated as follows:
    $(\U:\GatSet, \El:\U\ra\GatSet)$ is a family GAT, 
   and any extension of a family GAT by a new operation 
$∏ a_1:A_1. … ∏ a_n : A_n. B$ 
or an equation $∏ a_1:A_1. … ∏ a_n : A_n. t =_B u$ with $B$ of type $\GatSet$ 
is a family GAT.
    It will prove more convenient to exploit the $\Lexp$-structure of $\FinGat$ described by Uemura to give the following characterisation\footnote{Note that this characterisation makes sense in any $\Lexp$-category.} of family GATs.
  \begin{definition}
    We define the subset of family GATs, implicitly equipped with their projection to $\FamG=P\XGAT$,
    by the following generating rules.
     \begin{mathpar}
  \inferrule{ }{ P \XGAT \familyGAT}
  \and
  \inferrule{
    Γ \familyGAT \\
%\vcenter
{%\hbox
{\begin{tikzcd}[ampersand replacement=\&, baseline=(b.base)]
Δ \arrow[r]   \arrow[d] 
\arrow[dr, phantom," " {pullback=black}, very near start, color=black]
\& P^n \YGAT \arrow[d, "P^n \pGAT"] \\
|[alias=b]| Γ \arrow[r] \& P^n \XGAT 
\end{tikzcd}}}
  }{Δ \familyGAT}
  \and
  \inferrule{
    Γ \familyGAT \\
    {\begin{tikzcd}[ampersand replacement=\&, baseline=(b.base)]
% equaliser diagram
 Γ \arrow[r, "σ"] \arrow[d, "δ"'] 
 \& P^n \YGAT \arrow[d, "P^n \pGAT"]
 \\
  P^n \YGAT \arrow[r, "P^n \pGAT"']
  \& P^n \XGAT
\end{tikzcd}}
\\
{\begin{tikzcd}[ampersand replacement=\&, baseline=(b.base)]
% equaliser diagram
 Δ \arrow[r] \& Γ \arrow[r, shift left=0.7ex, "σ"] \arrow[r, "δ"', shift right=0.7ex] 
 \& P^n \YGAT
\end{tikzcd}}
\text{ equaliser}
  }{Δ \familyGAT}
\end{mathpar}
  \end{definition}
  \begin{remark}
    The natural number $n$ in the last two rules corresponds to the number 
    of premises when extending a family GAT with an operation $∏ a_1:A_1. … ∏ a_n : A_n. B$,
    or an equation $∏ a_1:A_1. … ∏ a_n : A_n. t =_B u$. 
    In the last rule, the morphisms $ σ$ and $ δ$ correspond to 
    the terms $t$ and $u$ in the above equation. In this view, the commutative square ensures that 
    $t$ and $u$ are terms of the same type.
  \end{remark}
  The full subcategory of $\FinGat/\FamG$ spanned by family GATs inherits the $\Lexp$-structure of $\FinGat/\FamG$ given by \cref{prop:p2-exponentiable}
  in the following sense.
  \begin{restatable}{proposition}{familyGats}
    \label{prop:family-gats}
 $\GatFam{(\GVar{A}:\U, \GVar{a}: \El\, A)}$
    and $\GatFam{(\GVar{A}:\U)}$ are family GATs.
    Moreover, family GATs 
     are closed under finite limits and pushforwards along 
     $\GatFam{(\GVar{A}:\U, \GVar{a}: \El\, A)} → \GatFam{(\GVar{A}:\U)}$.
    % The category $\FamGat$ of family GATs has finite limits, includes 
    % the morphism $\GatFam{(\GVar{A}:\U, \GVar{a}: \El\, A)} → \GatFam{(\GVar{A}:\U)}$ which is also exponentiable in $\FamGat$. Moreover, the embedding $\FamGat ↪ \FinGat/\FamG$ preserves finite limits as well as pushforwards along this exponentiable morphism.
  \end{restatable}

  \begin{corollary}
    The image of any GAT by 
    the \kl{two-sortification functor} is, up to isomorphism, a family GAT.
  \end{corollary}
  \begin{proof}
    By uniqueness of the bi-initial $\Lexp$-functor, 
    the \kl{two-sortification functor} is isomorphic to 
    the composition of the bi-initial $\Lexp$-functor from $\FinGat$ to
    its full subcategory of family GATs,
    with the embedding of family GATs into $\FinGat/\FamG$. This composition 
    maps any GAT to a \kl{family GAT} by definition.
  \end{proof}
  % The strategy employed in this section can also be used to prove 
  % that the two-sortification of a GAT does not have any sort equations. 

  % Finally, note that if we define two-sortification
  % as the \kl{strict} bi-initial $\Lexp$-functor from $\FinGat$ to $\FamGat$,
  % then the two-sortification of any GAT is a \kl{family GAT} on the nose, rather than up to isomorphism.
  % Adopting this definition does not change anything in the rest of this paper since we only rely on the fact that two-sortification is a (\kl{strict}) bi-initial $\Lexp$-functor
  % to $\FinGat/\FamG$.

 \section{Semantic Translation}
 \label{sec:semantic-translation}
 In this section, we prove the main result of this paper, which characterises
 the category of models of the translated GAT.
More specifically, we define, for each theory $Γ$, 
its \emph{family functor} $\FamFunctor{Γ}$ mapping a model $M$ of $Γ$ 
to a family $El_M ∶ U_M → 𝐒𝐞𝐭$ of sets, such that
the category of models of $\TW Γ$ is isomorphic to the full subcategory $\cartCommaFam{\FamFunctor{Γ}}$
of the comma category $\FamFunctor{Γ}/\FamCat{𝐒𝐞𝐭}$ spanned by morphisms that are cartesian in the following sense:\footnote{
  Cartesian morphisms are those morphisms in the splitting cleavage of the family fibration $\FamCat{𝐒𝐞𝐭} \to \Set$.}
\begin{definition}
  A \AP \intro{cartesian} morphism in $\FamCat{𝐒𝐞𝐭}$ is a family morphism $(f,g) ∶ (U,El) → (U',El')$ such that for all $u ∈ U$, the sets $El(u)$ and $El'(f(u))$ are equal and the map $g_u ∶ El(u) → El'(f(u))$ is the identity.
\end{definition}

More explicitly, an object $(U_M, El_M) \xrightarrow{(f,\id)} (U,El)$ consists of a model $M$ of $Γ$, a family $(U,El)$ of sets, and a function $f ∶ U_M → U$ such that 
    $El_M(u) = El(f(u))$ for all $u ∈ U_M$.

    % Moreover, the assignment $Γ ↦ \FamFunctor{Γ}$ is functorial, and the isomorphism $\model{\TW Γ} ≅ \cartCommaFam{\FamFunctor{Γ}}$ 
    % is natural in $Γ$ and compatible with the projection to $\FamS$, in 
    % a sense that will be made precise later.

     We already illustrated this result with transitive graphs in the introduction (\cref{ex:transitive-graphs-characterisation-reduced}).
     We now give more elementary examples that will motivate part of the proof.
  \begin{example}
    \label{ex:main-thms}
    The category of models of the empty theory is the terminal category. We define the family functor as mapping the unique model to the empty family. In this way, we indeed get an isomorphism between cartesian morphisms $(\emptyset, \emptyset) → (U,El)$ and models of the two-sortification of the empty theory, which is the GAT of families.

    Models of $(\GVar{A} : \GatSet)$ are sets.
    Applying two-sortification yields
    $\GatFam{(\GVar{A} ∶ \U)}$, whose models are base-pointed families, i.e. families equipped with an element in the indexing set.
    By taking the family functor as mapping a set $A$ to the family $(U_A,El_A)$ defined by $U_A = \{*\}$ and $El_A(*) = A$, we indeed get an isomorphism with the category of cartesian morphisms 
    $(U_A,El_A) \xrightarrow{(f,\id)} (U,El)$. Indeed, it is clear that such a cartesian morphism determines a base-pointed family $(U,El)$ by taking $f(*)$ as element of $U$.
    Conversely, given a base-pointed family $(U,El)$ with $u ∈ U$,
     we define the set $A = El(u)$, and choose the cartesian family morphism $(U_A, El_A) →(U,El)$ as picking the element $u ∈ U$.

    Similarly, models of $(\GVar{A} : \GatSet, \GVar{a} : A)$ are pointed sets, and by defining the family functor as mapping a pointed set $(A,a)$ to the same family $(U_A,El_A)$ as above, we get an isomorphism between the category of cartesian arrows and the category of models of $\GatFam{(\GVar{A} ∶ \U, \GVar{a} : \El\, A)} ≅ \TW (\GVar{A} : \GatSet, \GVar{a} : A)$.
  \end{example}

    In the following, we define the family functor $\FamFunctor{Γ}$ of any theory $Γ$ (\cref{sec:family-functor-gat}), and
     we prove the claimed characterisation of the models of $\TW Γ$ in terms of $\FamFunctor{Γ}$
    (\cref{sec:explicit-description-models-gat}).

%   \end{theorem}
  % \begin{remark}
  %   \label{rem:natural-characterisation-models-two-sortification}
  %   The naturality condition mentioned at the end will be made explicit later.
  % \end{remark}
  
%   Taking this characterisation of models of reduced GATs for granted, we sometimes identify models of $\TW Γ$ with objects 
%   of this category.

%   We can immediately derive the 
%    justification of the reduction 
%  from the point of view of Initial Algebra Semantics.
% %  \begin{definition}
% %   \Cref{thm:characterisation-models-two-sortification} induces two functors from 
% %   $\model{\TW Γ}$ to $\model{Γ}$ and $\FamS$ respectively, by postcomposing the isomorphism of the theorem with the projections from the comma category $\FamFunctor{Γ}/\FamS$ to $\model{Γ}$ and $\FamS$.
% %   We call the first one the \AP\intro{source projection functor} and second one the 
% %   \intro{target projection functor}.
% %  \end{definition}
%   \begin{corollary}
%     The inner model of an initial model of $\model{\TW Γ}$
%     is initial.
%   \end{corollary}
%   \begin{proof}
%   \end{proof}

% We start with a preliminary section to define 
%  the functor mapping a model to a family, a key ingredient
% in the explicit description of the models of $\TW Γ$.
  \subsection{The Family Functor of a GAT}
  \label{sec:family-functor-gat}
  In this section, we define the family functor of a theory $Γ$
  % To define the family functor of an arbitrary GAT, we will 
  by relying on bi-initiality of $\FinGat$.

  \begin{notation}
  \label{not:FamSet}
  We sometimes denote $\FamCat{𝐒𝐞𝐭}$ by \AP $\intro*\FamS$: this is the category of
  families $(El(A))_{A ∈ U}$ of sets indexed by some set $U$.
\end{notation}

  % \Cref{ex:main-thms} provides 
% These examples provides the base cases for the definition of the family functor of a GAT by bi-initiality of $\FinGat$. 
The strategy goes  as follows: we define a suitable 
  $\Lexp$-category whose objects are functors to $\FamS$, so that
  the bi-initial $\Lexp$-functor $\FamFunctor{}$ to that category maps  
  a theory $Γ$ to its family functor $\FamFunctor{Γ}∶\model{Γ}→ \FamS$.
  The main subtlety lies in choosing the right notion of morphisms.
  It is tempting to consider commuting triangles: 
  a morphism between $ℂ \xrightarrow{F} \FamS$ and $ℂ' \xrightarrow{F'} \FamS$
  would be a functor $G:ℂ → ℂ'$ such that $F = F' \circ G$. This would yield the \kl{slice category} $\CAT/\FamS$.
  Unfortunately, the bi-initial $\Lexp$-functor to $\CAT/\FamS$ would not meet our expectations.
  In particular, because it preserves terminal objects, it would map 
  the empty theory to $\FamS \xrightarrow{\mathrm{id}}\FamS$, instead of  
  the family functor $1 → \FamS$ of the empty theory described in \cref{ex:main-thms}.

  If we relax the notion of morphism,
  by replacing the equality $F = F' ∘ G$ by a natural transformation $F' ∘ G ⇒ F$,
   we recover the expected limits and exponentiable morphism.
  % In this way, we recover the right terminal objects, and more generally the expected (finite) limits, and exponentiable morphism.
  However, there is still an issue. We expect the isomorphism $\model{\TW Γ} ≅ \cartCommaFam{\FamFunctor{Γ
  }}$ to be natural in $Γ$, in some way. But for naturality to even make sense, the right hand side must be functorial in $Γ$, and so we need to understand how a given GAT morphism $σ ∶ Γ → Δ$ induces a functor 
  $\cartCommaFam{\FamFunctor{Γ}} → \cartCommaFam{\FamFunctor{Δ}}$.
  There is an obvious candidate: the functor
  precomposing a cartesian morphism $\FamFunctor{Γ}(M) → (U,El)$ 
    with the $M$-component of the natural transformation $\FamFunctor{σ} ∶ \FamFunctor{Δ} ∘ \model{σ} ⇒ \FamFunctor{Γ}$ induced by $σ$, according to our relaxed notion of morphism.
Unfortunately, this composite is not necessarily cartesian, and so is not an object of $\cartCommaFam{\FamFunctor{Δ}}$. It is, of course, if $(\FamFunctor{σ})_M$ is cartesian. Therefore, 
in the notion of morphism, we require the natural transformation $F' ∘ G ⇒ F$ to be cartesian.
This leads us to consider the following target for our desired bi-initial $\Lexp$-functor.
\begin{definition}
  Let \AP$\intro*\colaxSliceCart$ denote the \emph{cartesian colax slice category} of locally small categories over $\FamS$: objects are functors to $\FamS$ and a morphism between $ℂ \xrightarrow{F} \FamS$ and $ℂ' \xrightarrow{F'} \FamS$ consists of a functor $G ∶ ℂ → ℂ'$, together with a (pointwise) cartesian natural transformation $α : F' ∘ G ⇒ F$, i.e., a natural transformation such that for every object $c$ of $ℂ$, the component $α_c$ is a \kl{cartesian} morphism in $\FamS$.
\end{definition}

To make it a $\Lexp$-category, we need to choose an exponentiable morphism in $\colaxSliceCart$.
This morphism will 
  be the image of $\pGAT∶ (\GVar{A}:\GatSet, \GVar{a}:A) \to (\GVar{A}:\GatSet)$ by 
  the bi-initial $\Lexp$-functor $\FamFunctor{} ∶ \FinGat → \colaxSliceCart$.
  More explicitly, this morphism must be a suitable natural transformation from the family functor of pointed sets to the family functor of sets precomposed with the projection $\PtdSet → 𝐒𝐞𝐭$.
% to the forgetful functor $\model{\pGAT}$ from pointed sets to sets, with a suitable natural transformation involving the family functors.
Based on \cref{ex:main-thms}, the identity natural transformation is an obvious candidate.
\begin{restatable}{theorem}{colaxExponentiable}
  \label{thm:colax-exponentiable}
    $\colaxSliceCart$ has finite limits and the projection functor $\colaxSliceCart → \CAT$ 
  preserves them. Moreover,
  the forgetful functor from pointed sets to sets, equipped with the identity natural transformation between their family functors from \cref{ex:main-thms}, is exponentiable in $\colaxSliceCart$, and the functor $\colaxSliceCart → \CAT$ 
  \kl{preserves pushforwards} along it.

  Therefore, %combining with \cref{prop:colax-slice-limits},
   $\colaxSliceCart$ is a $\Lexp$-category and the functor $\colaxSliceCart → \CAT$ is a $\Lexp$-functor.
\end{restatable}
% \begin{proof}
%   A detailed proof is given in \cref{sec:proof-semantic-translation}.
% \end{proof}
\begin{corollary}
  \label{cor:family-functor-gat}
  There is a \kl{strict} bi-initial $\Lexp$-functor $\FamFunctor{}$ from $\FinGat$ to $\colaxSliceCart$ that maps a theory $Γ$ to its category of models $\model{Γ}$ equipped with 
  a functor $\FamFunctor{Γ}$ to $\FamS$, which we call the \AP\intro{family functor} of $Γ$.
\end{corollary}
\begin{proof}
  By uniqueness of the initial morphism, the composition $\FinGat → \colaxSliceCart → \CAT$ is isomorphic to the \kl{model functor}.
  Moreover, because the projection $\colaxSliceCart → \CAT$ is an isofibration, 
  we can refine the $\Lexp$-functor $\FinGat → \colaxSliceCart$ so that, 
   when postcomposed with the projection to $\CAT$,
  it coincides on the nose with the \kl{model functor}.
\end{proof}

\subsection{Description of the Models of the Reduced GAT}
\label{sec:explicit-description-models-gat}
In this section, we give an explicit description of the models of $\TW Γ$ in terms of the models of $Γ$ and its \kl{family functor}, defined in \cref{sec:family-functor-gat}.
More specifically, this section is devoted to the proof of the following statement.

 \begin{theorem}
  \label{thm:characterisation-models-two-sortification}
  Given a theory $Γ$, 
  the category of models of its two-sortification is isomorphic to the full subcategory
  \AP $\intro*\cartCommaFam{\FamFunctor{Γ}}$ of the comma category $\FamFunctor{Γ}/\FamS$ spanned by \kl{cartesian} morphisms of families.
   Explicitly, denoting by 
    $El_M ∶ U_M → 𝐒𝐞𝐭$ the image of a model $M$ by $\FamFunctor{Γ}$,
  % $\model{Γ}$
  \begin{itemize}
    \item an object consists of a model $M$ of $Γ$, a family $(U,El)$ of sets, and a \kl{cartesian} morphism $(U_M, El_M) → (U,El)$;
    \item a morphism from $(U_M, El_M) → (U,El)$ to $(U_{M'}, El_{M'}) → (U',El')$ consists 
     of a morphism of models $m ∶ M → M'$ and a morphism $(U,El) → (U',El')$ in $\FamS$ making the obvious square commute
     in $\FamS$.
  \end{itemize}
  Moreover, this isomorphism of categories is natural in $Γ$, and the induced natural transformation is unique such that it is
   compatible with the projection to $\FamS$, in the sense that the following diagrams commute for every GAT morphism $σ ∶ Γ → Δ$.
%   \begin{equation}
%     \label{eq:naturality-characterisation}
\[
% YADE DIAGRAM main-thm-naturality.yade
% GENERATED LATEX
\input{diagrams/main-thm-naturality.tex}
% END OF GENERATED LATEX
\]
\end{theorem}

In the rest of this subsection, we tackle the proof of \cref{thm:characterisation-models-two-sortification}.
The core argument consists 
in upgrading the below mappings 
into suitable $\Lexp$-functors from $\FinGat$ to
$\CAT/\FamS$, where the latter is seen as a $\Lexp$-category by \cref{prop:p2-exponentiable}.
\begin{mathpar}
Γ ↦ (\model{\TW Γ} → \FamS)
\and
 Γ ↦ (\cartCommaFam{\FamFunctor{Γ}} → \FamS)
\end{mathpar}
Then, by bi-initiality, those two functors are isomorphic, yielding the desired unique natural isomorphism between $\model{\TW Γ}$ and $\cartCommaFam{\FamFunctor{Γ}}$.

The first mapping can be recovered as the composition 
of $\TW$ with a suitable $\Lexp$-functor from $\FinGat/\FamG$ to $\CAT/\FamS$, thanks to the following proposition.
\begin{proposition}
  \label{prop:lexp-functor-p2}
  The mapping $(ℂ, Y \xrightarrow{p} X) ↦ (ℂ/PX, \YY \xrightarrow{\pp} \XX)$
  described in \cref{prop:p2-exponentiable}
  extends to an endofunctor \AP$\intro*\MFunctor{-}$ on $\Lexp$.
  The image of a $\Lexp$-functor $F ∶(ℂ, p) → (ℂ', p')$
  is the $\Lexp$-functor $\MFunctor{F} ∶ ℂ / P_pX → ℂ'/P_{p'}X'$
  mapping $Z → P_pX$ to $FZ → FP_pX ≅ P_{p'}X'$.
\end{proposition}
\begin{corollary}
  \label{cor:models-two-sortification-lexp-functor}
  The mapping $Γ ↦ (\model{\TW Γ} → \FamS)$ extends to a $\Lexp$-functor from $\FinGat$ to $\CAT/\FamS$ as the composition of the two-sortification functor $\TW$ and the image of the \kl{model functor} by 
  $\MFunctor{-}$.
\end{corollary}
Similarly, we 
decompose $Γ ↦ (\cartCommaFam{\FamFunctor{Γ}} → \FamS)$
as the composition of the bi-initial $\Lexp$-functor $\FamFunctor{}$ to $\colaxSliceCart$, 
and a suitable $\Lexp$-functor from $\colaxSliceCart$ to $\CAT/\FamS$, mapping 
 $F ∶ ℂ → \FamS$ to $\cartCommaFam{F} → \FamS$.
\begin{restatable}{proposition}{semTwosortFunctor}
  \label{prop:semantic-twosort-functor}
  The mapping $ (ℂ \xrightarrow{F} \FamS) ↦ (\cartCommaFam{F} → \FamS)$ extends to a $\Lexp$-functor from $\colaxSliceCart$ to $\CAT/\FamS$.
  % The image of a morphism $α : F' ∘ G ⇒ F$ is the functor $\laxSliceCart[F] → \laxSliceCart[F']$ over $\FamS$ mapping a cartesian morphism $F(M) → (U,El)$ to the composite $F'(G(M)) \xrightarrow{α_M} F(M) → (U,El)$.
\end{restatable}
The following corollary concludes the proof of \cref{thm:characterisation-models-two-sortification}, as bi-initiality implies that
the $\Lexp$-functor considered here is isomorphic to the one considered in \cref{cor:models-two-sortification-lexp-functor}.
  \begin{corollary}
    \label{cor:models-two-sortification-lax-slice-cart}
  %  The mapping $Γ ↦ (\laxSliceCart[\FamFunctor{Γ}] → \FamS)$
  %  extends to
  There is a $\Lexp$-functor from $\FinGat$ to the slice category $\CAT/\FamS$
  that maps a theory $Γ$ to $\cartCommaFam{\FamFunctor{Γ}} → \FamS$.
  This $\Lexp$-functor is the composition of the bi-initial $\Lexp$-functor $\FamFunctor{}$ to $\colaxSliceCart$ defined in \cref{cor:family-functor-gat}, mapping 
$Γ$ to its \kl{family functor}, and the $\Lexp$-functor from $\colaxSliceCart$ to $\CAT/\FamS$, mapping
 $F ∶ ℂ → \FamS$ to $\cartCommaFam{F} → \FamS$.
\end{corollary}

\section{The Core Projection Functor}
\label{sec:core-proj-functor}
\Cref{thm:characterisation-models-two-sortification} 
indirectly induces a functor from $\model{\TW Γ}$ to $\model{Γ}$
given by postcomposing the isomorphism $\model{\TW Γ} ≅ \cartCommaFam{\FamFunctor{Γ}}$
with the projection mapping a cartesian morphism $\FamFunctor{Γ}(M) → (U,El)$ to $M$.
We call this composite the \AP\intro{core projection functor} of $Γ$.

%  We deduce the 
%    justification of the reduction 
%  from the point of view of Initial Algebra Semantics.
%  Finally, we show that the \emph{core projection functor} $\model{\TW Γ} ≅\cartCommaFam{\FamFunctor{Γ}} → \model{Γ}$ the image 
%  is of the \kl{core projection} morphism (\cref{sec:connection-core-projection-morphism}).

% In this section, we show that it preserves initial models.
% , then we show that it is the image by the \kl{model functor} of a GAT morphism from $\TW Γ$ to $Γ$, which we call the \AP\intro{core projection morphism} of $Γ$.

In this section, we prove that the \kl{core projection functor} preserves initial models (\cref{sec:ini-preser-core-proj}), and that it is induced by a GAT morphism (\cref{sec:core-proj-mor}).
\subsection{Initiality Preservation}
\label{sec:ini-preser-core-proj}
 From the perspective of initial algebra semantics, preservation of initiality by the \kl{core projection functor} justifies the translation: intuitively, it means that the inductive type described by $\TW Γ$ can be used to compute a model of $Γ$ that behaves like the inductive type that it describes, in the sense that it enjoys the same recursion principle.
\begin{proposition}
  \label{prop:core_projection_preserves_initial}
  The \kl{core projection functor} of a theory $Γ$ preserves initial models.
\end{proposition}
\begin{proof}
  It is enough to show that
  $\cartCommaFam{\FamFunctor{Γ}} → \model{Γ}$ preserves the initial object.
  This follows from the fact that this functor is the right adjoint of a coreflection.
  We provide below a more explicit argument.

  Consider an initial object 
  $(U_M,El_M) \xrightarrow{(f,\id)} (U,El)$
  in $\cartCommaFam{\FamFunctor{Γ}}$.
We want to show that given a model $N$ of $Γ$,
there exists a unique morphism $M → N$.
Note that the set of morphisms from $M$ to $N$ is bijective to the set of morphisms 
in $\cartCommaFam{\FamFunctor{Γ}}$ between the identity morphisms on $(U_M,El_M)$ and $(U_N,El_N)$ respectively. Therefore, it is enough to show that the identity morphism on $(U_M, El_M)$ is initial. This follows from it being isomorphic to $(U_M,El_M) \xrightarrow{(f,\id)} (U,El)$, by initiality of the latter.
\end{proof}

  \subsection{The Core Projection GAT Morphism}
  \label{sec:core-proj-mor}
%   In \cref{sec:semantic-translation}, we will provide an alternative description of the models of $\TW Γ$ as tuples whose first components are models of $Γ$. 
% We will show later that that first component is initial whenever the model of $\TW Γ$ is initial, 
% a key property from the perspective of Initial Algebra Semantics.
% % This means that the first projection functor plays a key role, in the perspective of Initial Algebra Semantics.
  In this section, for each theory $Γ$ we define a GAT morphism from $\TW Γ$ to $Γ$ whose image by the \kl{model functor} is the core projection functor. We call it the \emph{core projection morphism} of $Γ$.

  % for each theory $Γ$, a morphism from its two-sortification $\TW Γ$ to $Γ$, which we will later show corresponds to that first proj
  
  % We call it 
  % the \emph{core projection morphism} of $Γ$. The naming is motivated by the fact the characterisation of the models of the translated GAT that we give in \cref{sec:semantic-translation}, 
  % we show that the image of this morphism by the \kl{model functor} is the right adjoint of the strict coreflection between the categories of models of $Γ$ and $\TW Γ$.

  \begin{definition}
    \label{def:core-projection-morphism}
    By~\cref{thm:universal-property-natural-transformation}, 
     the (right) pullback square defining $\pp$ in Equation~\eqref{eq:p2-diagram} 
     uniquely extends to a natural transformation 
     as below:
    \[\begin{tikzcd}[ampersand replacement=\&]
    \& \FinGat/\FamG \\
    \FinGat \&\& \FinGat,
    \arrow["\dom", from=1-2, to=2-3]
    \arrow["{T}", from=2-1, to=1-2]
    \arrow[""{name=0, anchor=center, inner sep=0}, equals, from=2-1, to=2-3]
    \arrow["\coref{}", between={0.4}{0.7}, Rightarrow, from=1-2, to=0]
    \end{tikzcd}\]
    The component of this natural transformation at a theory $Γ$ is a GAT morphism \AP $\intro*\coref{Γ} ∶ \TW Γ → Γ$ which we call 
    the \intro{core projection morphism} of $Γ$.
  \end{definition}
  \begin{example}
    \label{ex:core-projection-simple-gats}
    By definition,
     the bottom and top horizontal morphisms in the right pullback 
    square \eqref{eq:p2-diagram} for $ℂ = \FinGat$
    are the \kl{core projection morphisms} 
    of $(\GVar{A}:\GatSet)$ and $(\GVar{A}:\GatSet, \GVar{a}:A)$ respectively.
    More explicitly, the \kl{core projection morphism} of $(\GVar{A}:\GatSet)$ is the substitution $\GatFam{(\GVar{A}:\U)} → (\GVar{A}:\GatSet)$
    induced by the term $\GatFam{(\GVar{A}:\U)} ⊢ \El\, A ∶ \GatSet$.
    % Similarly, the core projection morphism of $(\GVar{A}:\GatSet, \GVar{a}:A)$ is the substitution $\GatFam{(\GVar{A}:\U, \GVar{a}: \El\, A)} → (\GVar{A}:\GatSet, \GVar{a}:A)$ generated by the terms 
    % $\El\, A$ and $a$.
    Its image by the \kl{model functor}  
    maps $(U,El,A)$ to the set $El(A)$, which is indeed the action of the \kl{core projection functor}, as can be seen from the isomorphism between cartesian arrows and base-pointed families described in \cref{ex:main-thms}.
    %  As claimed in 
    % the introduction, it has a coreflective left adjoint mapping a set $X$ to the model defined by $U:=\{*\}$, $El\, * := X$, and $A := *$. 
    % If we apply this left adjoint to a set $X$ and then the right adjoint, we indeed get back $X$.

    % Similarly, the image by the \kl{model functor} of the core projection morphism of $(\GVar{A}:\GatSet, \GVar{a}:A)$ is a functor from $\model{\GatFam{(\GVar{A}:\U, \GVar{a}: \El\, A)}}$ to $\model{(\GVar{A}:\GatSet, \GVar{a}:A)}$
    % mapping a model $(U,El,A,a)$ to the pointed set $(El(A), a)$.
    % Again, it has a coreflective left adjoint mapping a pointed set $(X,x)$ to the model defined by $U:=\{*\}$, $El\, * := X$, $A := *$, and $a := x$.
  \end{example}

% \subsection{Relating the Core Projection Morphism and Functor}
% \label{sec:connection-core-projection-morphism}

The rest of this section is devoted to the proof that the \kl{core projection functor} is the image of the \kl{core projection GAT morphism}
by the \kl{model functor}.
It consists in exploiting \cref{thm:universal-property-natural-transformation} 
to conclude that the two natural transformations below, from $\model{\TW -}$ to $\model{-}$, are identical,
where $α_Γ$ is the image of the \kl{core projection morphism} $\TW Γ → Γ$ by the \kl{model functor},
and $β_Γ$ is the \kl{core projection functor} of $Γ$ described above.
\[
% YADE DIAGRAM coreflector-model-functor.yade
% GENERATED LATEX
\input{diagrams/coreflector-model-functor.tex}
% END OF GENERATED LATEX
  \]

\begin{proposition}
  The families of morphisms $(α_Γ)_Γ$ and $(β_Γ)_Γ$ are natural.
  Moreover, their naturality squares for $\pGAT$ 
  are both the right pullback
  of Equation~\eqref{eq:p2-diagram}, for $ℂ = \CAT$.

  % Moreover, they coincide on $Γ= (\GVar{A:\GatSet})$ and 
  % $Γ=(\GVar{A}:\GatSet,\GVar{a}:A)$.
\end{proposition}
% \begin{proof}

%   We first show  naturality. 
%   The natural transformation $α$ is recovered as the following whiskering.
%   \[
%   % YADE DIAGRAM def-alpha.yade
%   % GENERATED LATEX
%   \input{diagrams/def-alpha.tex}
%   % END OF GENERATED LATEX
%   \]
%   On the other hand, $β$ is recovered as the following natural transformation,
%   where the isomorphism on the left 
%   is given by bi-initiality of $\FinGat$, and $ε$ denotes the counit of 
%   the adjunction between 
%   $\laxSliceCart$ and $\MFunctor{\CAT}$, see \cref{rem:right-adjoint-lax-slice-cart}.
%   \[
%   % YADE DIAGRAM def-beta.yade
%   % GENERATED LATEX
%   \input{diagrams/def-beta.tex}
%   % END OF GENERATED LATEX
% \]
% More explicitly, $β_Γ$ is the following composition
% \[
% \model{\TW Γ} ≅ 
% \cartCommaFam{\FamFunctor{Γ}} → \model{Γ}
% \]
% It remains to check the naturality squares of $α$ and $β$ for $\pGAT$.
% \end{proof}
We finally get the equality of $α$ and $β$ by a direct application of \cref{thm:universal-property-natural-transformation}.
\begin{corollary}
  \label{cor:core-projection-morphism-functor}
  The natural transformations $α$ and $β$ are equal. 
  In particular, the image of the \kl{core projection morphism}
  by the \kl{model functor} is the 
  \kl{core projection functor}, and thus preserves initiality by \cref{prop:core_projection_preserves_initial}.
\end{corollary}

\section{Conclusion}
We introduced the two-sortification translation of generalised algebraic theories and
notably gave 
an explicit description of the models of the translation of a theory in terms 
of its models. From this description, we deduce easily a functor from the former models to the latter that preserves initiality. Moreover, we show that this functor is induced by a GAT morphism.

Our work puts forward an approach to the study of QIITs based on Uemura's bi-initial characterisation of GATs. 
 We hope to demonstrate that this perhaps unusual take on the semantics of QIITs has potential and is worth exploring further.

Several directions remain open. First, we treat finite GATs here and infinite GATs in \iflong Appendix~\ref{sec:infinite-gats}\else ~\cite[Appendix~D]{avrillon:hal-05485795}\fi; it would be natural to also cover GATs with infinitary operations~\cite{DBLP:conf/lics/KovacsK20} and second-order generalised algebraic theories~\cite{DBLP:journals/mscs/Uemura23}, for which one would need an Uemura-style bi-initial characterisation. 

Second, our development is carried out in a set-theoretic metatheory.
Can we derive from it a two-sortification algorithm with its correctness proof in a theorem prover based on type theory? 
Different variants of the metatheory can be considered, such as intensional or extensional, with or without univalence. A related question is whether we can adapt Uemura's characterisation to a higher-dimensional variant of GATs corresponding to \emph{higher-inductive types}~\cite[Section 6]{hottbook}. 

Third, we may wonder if Uemura's characterisation extends to a biequivalence between $\Lexp$-categories and categories with families (CwFs) equipped with suitable structure, analogous to the one between \kl{Lex-categories} and CwFs with $\Sigma$ and extensional identity types~\cite{DBLP:journals/mscs/ClairambaultD14}. 

Finally, each GAT can be canonically equipped with an internal category structure: the object of morphisms is the GAT specifying the morphisms
of the original GAT. Therefore, the category of GATs carries a 2-categorical structure; whether this 2-category enjoys a universal property of its own is another question we leave for future work.

% Our development is heavily based on the bi-initiality of the category of finite GATs~\cite{UEMURA2022106991}, as a cartesian category equipped with the exponentiable morphism $(\GVar{U} : \GatSet, u : \U) → (\GVar{U} : \GatSet)$. We extended 
% % this characterisation to account for two-sortification of infinite GATs.
% Other generalisations remain to be investigated, such as infinitary operations~\cite{DBLP:conf/lics/KovacsK20}, second-order generalised algebraic theories~\cite{DBLP:journals/mscs/Uemura23},
% as well as adapting this work 
% to a type-theoretic, possibly univalent, metatheory.

\bibliography{bibliography}    

\appendixOrNot{
\newpage
\appendix
\input{s-appendices_csl}
}{}

\end{document}

%% file: knowledge.txt
\knowledge{notion}
 | strict

\knowledge{notion}
 | is specified
 | specifies

\knowledge{notion}
 | model functor

\knowledge{notion}
 | exponentiable
 | Exponentiable

\knowledge{notion}
 | cartesian

\knowledge{notion}
 | preserves pushforwards
 | preserving pushforwards
 | preserve pushforwards

\knowledge{notion}
 | two-sortification functor

\knowledge{notion}
 | family GAT
 | family GATs

\knowledge{notion}
 | family functor

\knowledge{notion}
 | core projection morphism
 | core projection GAT morphism
 | core projection morphisms
 | core projection

\knowledge{notion}
 | core projection functor

\knowledge{notion}
 | Lex-category
 | Lex-categories

\knowledge{notion}
 | free filtered completion
 | free filtered completions

\knowledge{notion}
 | faithful (resp. fully faithful) at
 | faithful (resp. full) at
 | fully faithful at
 | full at
 | fullness at
 | faithful at
 
 \knowledge{notion}
 | slice category

%% file: diagrams/nat-universal.tex
\begin{tikzpicture}
\draw[white,-,curve={ratio=0.2}, line width=0.20833333333333334em] (14.973958333333334em,-8.854166666666668em) .. controls (18.369444444444444em,-10.208333333333332em) and (21.76354166666667em,-10.20486111111111em) .. (25.15625em,-8.843750000000002em);
\draw[black,->, curve={ratio=0.2}, ] (14.973958333333334em,-8.854166666666668em) .. controls (18.369444444444444em,-10.208333333333332em) and (21.76354166666667em,-10.20486111111111em) .. (25.15625em,-8.843750000000002em);
\draw[white,-,curve={ratio=-0.2}, line width=0.20833333333333334em] (14.973958333333334em,-8.333333333333334em) .. controls (18.369444444444444em,-6.979166666666667em) and (21.76354166666667em,-6.982638888888889em) .. (25.15625em,-8.34375em);
\draw[black,->, curve={ratio=-0.2}, ] (14.973958333333334em,-8.333333333333334em) .. controls (18.369444444444444em,-6.979166666666667em) and (21.76354166666667em,-6.982638888888889em) .. (25.15625em,-8.34375em);
\draw[white,double=white,identity,line width=0.20833333333333334em] (20.066145833333334em,-7.841145833333334em) .. controls (20.066145833333334em,-8.342881944444445em) and (20.066145833333334em,-8.844618055555555em) .. (20.066145833333334em,-9.346354166666666em);
\draw[black,double=white,->, cell=0.05, ] (20.066145833333334em,-7.841145833333334em) .. controls (20.066145833333334em,-8.342881944444445em) and (20.066145833333334em,-8.844618055555555em) .. (20.066145833333334em,-9.346354166666666em);
\node at (14.322916666666668em,-8.59375em) {$\FinGat$} ;
\node at (25.78125em,-8.59375em) {$ℂ$} ;
\node[scale=0.7] at (20.066682406361526em,-10.391687635057519em) {$G$} ;
\node[scale=0.7] at (20.066682406361526em,-6.795812364942483em) {$F$} ;
\node[scale=0.7] at (21.062099321399877em,-8.59375em) {$ α_{(x,y)}$} ;
\end{tikzpicture}

%% file: diagrams/p2.tex
\begin{tikzpicture}
\draw[white,-,line width=0.20833333333333334em] (10.416666666666668em,-2.9427083333333335em) .. controls (10.416666666666666em,-3.75em) and (10.416666666666666em,-4.557291666666667em) .. (10.416666666666668em,-5.364583333333334em);
\draw[black,->, ] (10.416666666666668em,-2.9427083333333335em) .. controls (10.416666666666666em,-3.75em) and (10.416666666666666em,-4.557291666666667em) .. (10.416666666666668em,-5.364583333333334em);
\draw[white,-,line width=0.20833333333333334em] (11.145833333333334em,-2.0833333333333335em) .. controls (12.065972222222221em,-2.0833333333333335em) and (12.98611111111111em,-2.0833333333333335em) .. (13.90625em,-2.0833333333333335em);
\draw[black,->, ] (11.145833333333334em,-2.0833333333333335em) .. controls (12.065972222222221em,-2.0833333333333335em) and (12.98611111111111em,-2.0833333333333335em) .. (13.90625em,-2.0833333333333335em);
\draw[white,-,line width=0.20833333333333334em] (14.583333333333334em,-2.96875em) .. controls (14.583333333333334em,-3.767361111111111em) and (14.583333333333334em,-4.565972222222222em) .. (14.583333333333334em,-5.364583333333334em);
\draw[black,->, ] (14.583333333333334em,-2.96875em) .. controls (14.583333333333334em,-3.767361111111111em) and (14.583333333333334em,-4.565972222222222em) .. (14.583333333333334em,-5.364583333333334em);
\draw[white,-,line width=0.20833333333333334em] (11.458333333333334em,-6.25em) .. controls (12.256944444444445em,-6.25em) and (13.055555555555555em,-6.25em) .. (13.854166666666668em,-6.25em);
\draw[black,->, ] (11.458333333333334em,-6.25em) .. controls (12.256944444444445em,-6.25em) and (13.055555555555555em,-6.25em) .. (13.854166666666668em,-6.25em);
\draw[white,-,line width=0.20833333333333334em] (18.75em,-2.9427083333333335em) .. controls (18.75em,-3.7586805555555554em) and (18.75em,-4.574652777777778em) .. (18.75em,-5.390625em);
\draw[black,->, dashed, ] (18.75em,-2.9427083333333335em) .. controls (18.75em,-3.7586805555555554em) and (18.75em,-4.574652777777778em) .. (18.75em,-5.390625em);
\draw[white,-,line width=0.20833333333333334em] (19.479166666666668em,-6.25em) .. controls (20.381944444444443em,-6.25em) and (21.28472222222222em,-6.25em) .. (22.1875em,-6.25em);
\draw[black,->, ] (19.479166666666668em,-6.25em) .. controls (20.381944444444443em,-6.25em) and (21.28472222222222em,-6.25em) .. (22.1875em,-6.25em);
\draw[white,-,line width=0.20833333333333334em] (22.916666666666668em,-2.96875em) .. controls (22.916666666666668em,-3.767361111111111em) and (22.916666666666668em,-4.565972222222222em) .. (22.916666666666668em,-5.364583333333334em);
\draw[black,->, ] (22.916666666666668em,-2.96875em) .. controls (22.916666666666668em,-3.767361111111111em) and (22.916666666666668em,-4.565972222222222em) .. (22.916666666666668em,-5.364583333333334em);
\draw[white,-,line width=0.20833333333333334em] (19.427083333333336em,-2.0833333333333335em) .. controls (20.364583333333332em,-2.0833333333333335em) and (21.302083333333332em,-2.0833333333333335em) .. (22.239583333333336em,-2.0833333333333335em);
\draw[black,->, ] (19.427083333333336em,-2.0833333333333335em) .. controls (20.364583333333332em,-2.0833333333333335em) and (21.302083333333332em,-2.0833333333333335em) .. (22.239583333333336em,-2.0833333333333335em);
\draw[white] (10.677083333333334em,-3.463541666666667em) -- (11.666666666666668em,-3.463541666666667em);
\draw[white] (11.666666666666668em,-3.463541666666667em) -- (11.666666666666668em,-2.34375em);
\draw[black] (10.677083333333334em,-3.463541666666667em) -- (11.666666666666668em,-3.463541666666667em);
\draw[black] (11.666666666666668em,-3.463541666666667em) -- (11.666666666666668em,-2.34375em);
\draw[white] (19.010416666666668em,-3.463541666666667em) -- (19.947916666666668em,-3.463541666666667em);
\draw[white] (19.947916666666668em,-3.463541666666667em) -- (19.947916666666668em,-2.34375em);
\draw[black] (19.010416666666668em,-3.463541666666667em) -- (19.947916666666668em,-3.463541666666667em);
\draw[black] (19.947916666666668em,-3.463541666666667em) -- (19.947916666666668em,-2.34375em);
\node at (10.416666666666668em,-2.0833333333333335em) {$\XX$} ;
\node at (10.416666666666668em,-6.25em) {$PX$} ;
\node at (14.583333333333334em,-2.0833333333333335em) {$Y$} ;
\node at (14.583333333333334em,-6.25em) {$X$} ;
\node at (18.75em,-2.0833333333333335em) {$\YY$} ;
\node at (18.75em,-6.25em) {$\XX$} ;
\node at (22.916666666666668em,-6.25em) {$X$} ;
\node at (22.916666666666668em,-2.0833333333333335em) {$Y$} ;
\node[scale=0.7] at (14.859137276810568em,-4.166666666666667em) {$p$} ;
\node[scale=0.7] at (18.474135405809314em,-4.166666666666667em) {$\pp$} ;
\node[scale=0.7] at (20.833333333333336em,-6.782744881390259em) {$\epsilon_X$} ;
\node[scale=0.7] at (23.192470610143904em,-4.166666666666667em) {$p$} ;
\end{tikzpicture}

%% file: diagrams/main-thm-naturality.tex
\begin{tikzpicture}
\draw[white,-,line width=0.20833333333333334em] (8.28125em,-9.166666666666668em) .. controls (8.28125em,-10.416666666666666em) and (8.28125em,-11.666666666666666em) .. (8.28125em,-12.916666666666668em);
\draw[black,->, ] (8.28125em,-9.166666666666668em) .. controls (8.28125em,-10.416666666666666em) and (8.28125em,-11.666666666666666em) .. (8.28125em,-12.916666666666668em);
\draw[white,-,line width=0.20833333333333334em] (10.703125em,-8.28125em) .. controls (13.133680555555555em,-8.28125em) and (15.56423611111111em,-8.28125em) .. (17.994791666666668em,-8.28125em);
\draw[black,->, ] (10.703125em,-8.28125em) .. controls (13.133680555555555em,-8.28125em) and (15.56423611111111em,-8.28125em) .. (17.994791666666668em,-8.28125em);
\draw[white,-,line width=0.20833333333333334em] (19.322916666666668em,-9.166666666666668em) .. controls (19.322916666666668em,-10.416666666666666em) and (19.322916666666668em,-11.666666666666666em) .. (19.322916666666668em,-12.916666666666668em);
\draw[black,->, ] (19.322916666666668em,-9.166666666666668em) .. controls (19.322916666666668em,-10.416666666666666em) and (19.322916666666668em,-11.666666666666666em) .. (19.322916666666668em,-12.916666666666668em);
\draw[white,-,line width=0.20833333333333334em] (10.78125em,-13.802083333333334em) .. controls (13.151041666666668em,-13.802083333333334em) and (15.520833333333334em,-13.802083333333334em) .. (17.890625em,-13.802083333333334em);
\draw[black,->, ] (10.78125em,-13.802083333333334em) .. controls (13.151041666666668em,-13.802083333333334em) and (15.520833333333334em,-13.802083333333334em) .. (17.890625em,-13.802083333333334em);
\draw[white,-,line width=0.20833333333333334em] (27.265625em,-8.28125em) .. controls (29.696180555555554em,-8.28125em) and (32.12673611111111em,-8.28125em) .. (34.55729166666667em,-8.28125em);
\draw[black,->, ] (27.265625em,-8.28125em) .. controls (29.696180555555554em,-8.28125em) and (32.12673611111111em,-8.28125em) .. (34.55729166666667em,-8.28125em);
\draw[white,-,line width=0.20833333333333334em] (25.729166666666668em,-9.166666666666668em) .. controls (26.979166666666668em,-10.416666666666666em) and (28.229166666666668em,-11.666666666666666em) .. (29.479166666666668em,-12.916666666666668em);
\draw[black,->, ] (25.729166666666668em,-9.166666666666668em) .. controls (26.979166666666668em,-10.416666666666666em) and (28.229166666666668em,-11.666666666666666em) .. (29.479166666666668em,-12.916666666666668em);
\draw[white,-,line width=0.20833333333333334em] (35em,-9.166666666666668em) .. controls (33.75em,-10.416666666666666em) and (32.5em,-11.666666666666666em) .. (31.25em,-12.916666666666668em);
\draw[black,->, ] (35em,-9.166666666666668em) .. controls (33.75em,-10.416666666666666em) and (32.5em,-11.666666666666666em) .. (31.25em,-12.916666666666668em);
\node at (8.28125em,-8.28125em) {$\cartCommaFam{\FamFunctor{Γ}}$} ;
\node at (8.28125em,-13.802083333333334em) {$\cartCommaFam{\FamFunctor{Δ}}$} ;
\node at (19.322916666666668em,-8.28125em) {$\model{\TW Γ}$} ;
\node at (19.322916666666668em,-13.802083333333334em) {$\model{\TW Δ}$} ;
\node at (24.84375em,-8.28125em) {$\cartCommaFam{\FamFunctor{Γ}}$} ;
\node at (35.88541666666667em,-8.28125em) {$\model{\TW Γ}$} ;
\node at (30.364583333333336em,-13.802083333333334em) {$\FamS$} ;
\node[scale=0.7] at (6.671273228154625em,-11.041666666666668em) {$\cartCommaFam{\FamFunctor{σ}}$} ;
\node[scale=0.7] at (14.348958333333334em,-7.768609068208343em) {$\cong $} ;
\node[scale=0.7] at (20.152989040252518em,-11.041666666666668em) {$\model{\TW σ}$} ;
\node[scale=0.7] at (14.3359375em,-14.314724265124992em) {$\cong $} ;
\node[scale=0.7] at (30.911458333333336em,-7.768609068208343em) {$\cong $} ;
\node[scale=0.7] at (27.05740316803629em,-11.588430165297043em) {$\cod $} ;
\end{tikzpicture}

%% file: diagrams/coreflector-model-functor.tex
\begin{tikzpicture}
\draw[white,line width=0.20833333333333334em] (9.244791666666668em,-8.59375em) .. controls (12.343749999999998em,-8.59375em) and (15.442708333333334em,-8.59375em) .. (18.541666666666668em,-8.59375em);
\draw[black,->, ] (9.244791666666668em,-8.59375em) .. controls (12.343749999999998em,-8.59375em) and (15.442708333333334em,-8.59375em) .. (18.541666666666668em,-8.59375em);
\draw[white,curve={ratio=-0.2}, line width=0.20833333333333334em] (17.526041666666668em,-5.550575461716435em) .. controls (18.465805641467423em,-6.019072359932446em) and (19.114396344497035em,-6.7470055393603em) .. (19.471813775755507em,-7.734375em);
\draw[black,->, curve={ratio=-0.2}, ] (17.526041666666668em,-5.550575461716435em) .. controls (18.465805641467423em,-6.019072359932446em) and (19.114396344497035em,-6.7470055393603em) .. (19.471813775755507em,-7.734375em);
\draw[white,curve={ratio=-0.2}, line width=0.20833333333333334em] (9.238354649295403em,-7.734375em) .. controls (9.762802822775123em,-6.622606316907087em) and (10.598281828565545em,-5.845029236130341em) .. (11.744791666666668em,-5.401643757669764em);
\draw[black,->, curve={ratio=-0.2}, ] (9.238354649295403em,-7.734375em) .. controls (9.762802822775123em,-6.622606316907087em) and (10.598281828565545em,-5.845029236130341em) .. (11.744791666666668em,-5.401643757669764em);
\draw[white,double=white,identity,curve={ratio=0.4}, line width=0.20833333333333334em] (13.608666807217626em,-5.7760413487752285em) .. controls (13.001611116357811em,-6.546021969868029em) and (13.007055510269048em,-7.311647075831842em) .. (13.624999988951334em,-8.072916666666668em);
\draw[black,double=white,->, cell=0.05, curve={ratio=0.4}, ] (13.608666807217626em,-5.7760413487752285em) .. controls (13.001611116357811em,-6.546021969868029em) and (13.007055510269048em,-7.311647075831842em) .. (13.624999988951334em,-8.072916666666668em);
\draw[white,double=white,identity,curve={ratio=-0.3}, line width=0.20833333333333334em] (14.940622384093041em,-5.7760413487752285em) .. controls (15.209256780821086em,-6.656110854849187em) and (15.018516113970842em,-7.421735960813em) .. (14.36840038354231em,-8.072916666666668em);
\draw[black,double=white,->, cell=0.05, curve={ratio=-0.3}, ] (14.940622384093041em,-5.7760413487752285em) .. controls (15.209256780821086em,-6.656110854849187em) and (15.018516113970842em,-7.421735960813em) .. (14.36840038354231em,-8.072916666666668em);
\node at (8.59375em,-8.59375em) {$\FinGat$} ;
\node at (20.052083333333336em,-8.59375em) {$\CAT$} ;
\node at (14.635416666666668em,-4.9166663487752285em) {$\CAT/\Fam$} ;
\node[scale=0.7] at (13.893229166666668em,-9.105546729354133em) {$\model{-}$} ;
\node[scale=0.7] at (19.474023209599014em,-5.773662149921837em) {$\dom$} ;
\node[scale=0.7] at (7.980909289810773em,-5.88071364588971em) {$\MFunctor{\model{-}} ∘ \TW$} ;
\node[scale=0.7] at (12.735516191630259em,-6.930746093263613em) {$\alpha $} ;
\node[scale=0.7] at (15.468282311475345em,-7.127214273578466em) {$\beta $} ;
\end{tikzpicture}

%% file: s-appendices_csl.tex
\
% \section{Proofs of Properties of \Lexp}
% \label{sec:proof-lexp-property}

\section{Proofs of the Main Results}

In this appendix, we provide some detailed proofs of the main results of the paper. 

\iflong
\subsection[Lexp-structure on FinGat/Fam]{$\Lexp$-structure on $\FinGat/\FamG$}
The right adjoint to the pullback functor along $\pp$ can be computed explicitly as follows.
% In the case of $\FinGat/\FamG$, it maps a theory $\GatFam{Γ}$ to 
% $\GatFam{(\GVar{A} : \U, \GVar{a} : \El\,\GVar{A}→ Γ)}$.
% where $Z'$ is defined by substituting every occurrence of $\GVar{U}$ by $\GVar{El}\,\GVar{A}$ and every occurrence of $\GVar{El}$ by the family $\λ x. Z[\GVar{U}↦\GVar{El}\,\GVar{A}, \GVar{El}↦λ a. Z[\GVar{U}↦\GVar{El}\,\GVar{A}]]$.
\begin{proposition}
  \label{prop:p2-description-polynomial}
  In the setting of \cref{prop:p2-exponentiable}, the right adjoint $\Pu$ of $\dom ∘ \pp^*$ maps a morphism $z ∶ Z → PX$ to the vertical composite on the left, where the dashed arrow is
  a morphism in the slice over $X$ between $\XX \xrightarrow{ε} X$ and $PPX → X$ obtained as 
   the transpose of the morphism on the right under the adjunction $\dom ∘ p^* ⊣ P$.
  \begin{mathpar}
  % YADE DIAGRAM p2-pushforward.yade
  % GENERATED LATEX
  \input{diagrams/p2-pushforward.tex}
  % END OF GENERATED LATEX
  \end{mathpar}
\end{proposition}

\subsection{Soundness}
We recall \cref{prop:family-gats}.
\familyGats*
\begin{proof}
  $\GatFam{(\GVar{A} : \U)}$ and
  $\GatFam{(\GVar{A} : \U, \GVar{a} : \El\,\GVar{A})}$ are family GATs
  by the second generating rule with $n = 0$ (see \cref{eq:p2-diagram}).

  We now give more details about the stability of family GATs under the polynomial functor $\Pu$,
  as described in \cref{prop:p2-description-polynomial}.
  The proofs of the other stability statements are similar.

  We show by induction on the inference rules defining family GATs that 
  if $Γ → PX$ is a family GAT, then so is $\Pu Γ → PX$.

For the first case, note that $\Pu (PX \xrightarrow{=} PX)$ is just $\XX → PX$, 
which is indeed a family GAT, as the pullback of $p ∶ Y → X$ along $PX → X$,
see \eqref{eq:p2-diagram}.

Next, consider a family GAT $Γ→ PX$ such that 
$\Pu Γ$ is a family GAT.

Assume given a pullback as below.
\begin{mathpar}
{\begin{tikzcd}[ampersand replacement=\&, baseline=(b.base)]
Δ \arrow[r]   \arrow[d] 
\arrow[dr, phantom," " {pullback=black}, very near start, color=black]
\& P^n \YGAT \arrow[d, "P^n \pGAT"] \\
|[alias=b]| Γ \arrow[r] \& P^n \XGAT 
\end{tikzcd}}
\end{mathpar}
We must show that $\Pu Δ$ is a family GAT. By pasting pullback squares, 
$\Pu Δ$ can be computed as below, where the right pullback
is justified by the fact $P$ is a right adjoint,
and thus preserves pullbacks.
\begin{mathpar}
% YADE DIAGRAM stability-family-GAT-P-pbk.yade
% GENERATED LATEX
\input{diagrams/stability-family-GAT-P-pbk.tex}
% END OF GENERATED LATEX
\end{mathpar}
The commutative square sitting on the dashed arrows is a pullback by the pasting 
law for pullbacks, and it presents $\Pu Δ$ as a pullback 
of a morphism $\Pu Γ → P^{n+1} X$ along $P^{n+1}p$. Therefore, 
since by induction hypothesis, $\Pu Γ$ is a family GAT, so is $\Pu Δ$.

Finally, consider an equaliser diagram as below, such 
that $P^n p ∘ σ = P^n p ∘ δ$.
\begin{mathpar}
\begin{tikzcd}[ampersand replacement=\&, baseline=(b.base)]
% equaliser diagram
 Δ \arrow[r] \& Γ \arrow[r, shift left=0.7ex, "σ"] \arrow[r, "δ"', shift right=0.7ex] 
 \& P^n \YGAT
\end{tikzcd}
\end{mathpar}
We must show that $\Pu Δ$ is a family GAT. 
By pasting pullback squares, 
$\Pu Δ$ can be computed as below.
\begin{mathpar}
% YADE DIAGRAM stability-family-GAT-P-eq.yade
% GENERATED LATEX
\input{diagrams/stability-family-GAT-P-eq.tex}
% END OF GENERATED LATEX
\end{mathpar}
It can be shown that $\Pu Δ → \Pu Γ$ is the equaliser 
of the dashed morphisms, thus justifying that $\Pu Δ$ is a family GAT, since $\Pu Γ$ is by induction hypothesis.
\end{proof}
\else
\subsection{Soundness}
\begin{proof}[Proof sketch of \cref{prop:family-gats}]
  $\GatFam{(\GVar{A} : \U)}$ and
  $\GatFam{(\GVar{A} : \U, \GVar{a} : \El\,\GVar{A})}$ are family GATs
  by the second generating rule with $n = 0$ (see \cref{eq:p2-diagram}).
Denoting $P_{\pp}$ by $\Pu$, stability under $\Pu$ is proved by induction on the inference rules
  defining family GATs using the explicit description of $\Pu$ given in~\cite[Appendix~A.1]{avrillon:hal-05485795}.
  For the first rule, $\Pu (PX \xrightarrow{=} PX)$ is $\XX → PX$, the
  pullback of $p ∶ Y → X$ along $PX → X$ in \eqref{eq:p2-diagram}.
  If $Δ$ is the pullback of $P^n p$ along $Γ → P^n X$, then, since $P$
  preserves pullbacks, the pasting law presents $\Pu Δ$ as the pullback of
  $P^{n+1} p$ along a morphism $\Pu Γ → P^{n+1} X$, and we conclude by
  the induction hypothesis on $\Pu Γ$.
  Similarly, if $Δ$ is the equaliser of $σ, δ ∶ Γ → P^n Y$, then
  $\Pu Δ → \Pu Γ$ is the equaliser of the two induced morphisms of the form
  $\Pu Γ → P^{n+1} Y$.
  Stability under finite limits is proved similarly.
  \end{proof}
  We refer the reader to ~\cite[Appendix~A.2]{avrillon:hal-05485795} for the full proof.
\fi

\subsection[Lexp-structure on F/C]{$\Lexp$-structure on $F/\CC$}
\label{sec:proof-comma-exp-structure}
\longonly{The following proposition is used in the proof of \cref{thm:universal-property-natural-transformation}.}
\begin{proposition}
    \label{prop:comma-lexp}
    Let $(\DD, p_\DD : Y \to X)$ and $(\CC, p_\CC : Y' \to X')$ be $\Lexp$-categories, and let $F : \DD \to \CC$ be a functor preserving pullbacks along $p_\DD$, 
    and let the following square be a pullback in $\CC$.
    \begin{mathpar}
    \begin{tikzcd}[ampersand replacement=\&]
	{F Y} \& {Y'} \\
	{F X} \& {X'}
	\arrow["{y}",from=1-1, to=1-2]
	\arrow["{F p_\DD}"', from=1-1, to=2-1]
	\arrow["\lrcorner"{anchor=center, pos=0.125}, draw=none, from=1-1, to=2-2]
	\arrow["{p_\CC}", from=1-2, to=2-2]
	\arrow["{x}"',from=2-1, to=2-2]
    \end{tikzcd}
  \end{mathpar} 
    The comma category $F/\CC$, whose objects are triples 
    $(D \in \DD, C \in \CC, f : FD \to C)$ and morphisms are commuting squares,
    can be equipped with a $\Lexp$-structure, where the chosen 
    exponentiable morphism is given by the pullback square $(Fp_\DD, p_\CC) : y \to x$, and 
    the functor
    $F/\CC \to \DD \times \CC$ is a $\Lexp$-functor.
\end{proposition}
\begin{proof}[Proof sketch]
Denoting $P_{p_𝔻}$ and $P_{p_\CC}$ by $P_𝔻$ and $P_\CC$ respectively, the composite
on the left below 
has a right adjoint $P : F/\CC \to (F/\CC)/x$ defined on the right below.
\[
\begin{tikzcd}[ampersand replacement=\&]
	{(F/\CC)/x} \&[+2em] {(F/\CC)/y} \& {F/\CC}
	\arrow["{(Fp_\DD, p_\CC)^*}", from=1-1, to=1-2]
	\arrow["\dom", from=1-2, to=1-3]
	\end{tikzcd}
\qquad
\begin{tikzcd}
	{F A} \\
	B
	\arrow["\alpha"',from=1-1, to=2-1]
\end{tikzcd}  
\mapsto 
\begin{tikzcd}
	{F P_\DD A} & {F X} \\
	{P_\CC B} & {X'}
	\arrow["{F P_\DD A}", from=1-1, to=1-2]
	\arrow["P \alpha"',from=1-1, to=2-1]
	\arrow["x", from=1-2, to=2-2]
	\arrow["{P_\CC B}"', from=2-1, to=2-2]
\end{tikzcd}
\]
Here, $P \alpha$ is obtained as the transposition of 
$F \PP \xrightarrow{F(\epsilon_A)} FA \xrightarrow{\alpha} B$ with respect to the adjunction 
$\dom ∘ p_\CC^* ⊣ P_\CC$,
where $\epsilon_A$ is the $A$-component of the counit of the adjunction $\dom ∘ p_\DD^* \dashv P_\DD$. Indeed, $F \PP$ is the composition of the pullback squares on the right.
\[
\begin{tikzcd}
	{\PP} & P_\DD A \\
	Y & X \\
	\arrow[from=1-1, to=1-2]
	\arrow[from=1-1, to=2-1]
	\arrow["\lrcorner"{anchor=center, pos=0.125}, draw=none, from=1-1, to=2-2]
	\arrow[from=1-2, to=2-2]
	\arrow["p_\DD"', from=2-1, to=2-2]
\end{tikzcd}
\qquad
\begin{tikzcd}
	{F\PP} & F P_\DD A \\
	FY & FX \\
	{Y'} & {X'}
	\arrow[from=1-1, to=1-2]
	\arrow[from=1-1, to=2-1]
	\arrow["\lrcorner"{anchor=center, pos=0.125}, draw=none, from=1-1, to=2-2]
	\arrow["", from=1-2, to=2-2]
	\arrow["Fp_\DD", from=2-1, to=2-2]
	\arrow["y"', from=2-1, to=3-1]
	\arrow["\lrcorner"{anchor=center, pos=0.125}, draw=none, from=2-1, to=3-2]
	\arrow["x", from=2-2, to=3-2]
	\arrow["{p_\CC}"', from=3-1, to=3-2]
\end{tikzcd}\]
The upper square is a pullback since $F$ preserves pullbacks along $p_\DD$. The lower square is the pullback from the assumption.

\iflong
To prove the adjunction, we use the following characterization of adjoint functors:
a functor $L : \mathbb{A} \to \mathbb{B}$ has a right adjoint $R$ if and only if for any object $b$ of $\mathbb{B}$, the comma category $L/b$ has a terminal object $(Rb : \mathbb{A}, \alpha : LRb \to b)$.
Hence, we need to show that for all $(D \in \DD, C \in \CC, g : FD \to C)$ in $F/\CC$, there exists a morphism $\theta: LPg \to g$ in $F/\CC$ such that $(Pg , \theta)$ is terminal in $L/g$.
	We take $\theta$ to be the morphism $\epsilon_g$ in $F/\CC$ whose 
  two components are given by the counits of the adjunctions $\dom \circ p_\CC^* \dashv P_\CC$ and $\dom \circ p_\DD^* \dashv P_\DD$.
	To show that $(P g, \epsilon_g)$ is terminal in $L/g$, we need to show that for all $\alpha \in (F/\CC)/x$
	and $\beta: L \alpha \to g$ in $F/\CC$,
	there exists a unique morphism $\iota : \alpha \to  P g$ in $(F/\CC)/x$, making the following diagram commute in $F/\CC$.
	\[\begin{tikzcd}[ampersand replacement=\&]
	{L\alpha} \&\& LPg \\
	\& g
	\arrow["L \iota ", from=1-1, to=1-3]
	\arrow["\beta"', from=1-1, to=2-2]
	\arrow["\epsilon_g", from=1-3, to=2-2]
	\end{tikzcd}\]
  This follows from the adjunctions $\dom \circ p_\CC^* \dashv P_\CC$ and $\dom \circ p_\DD^* \dashv P_\DD$.
	% We take $\iota$ to be the transposition of $\beta$.
	% The above diagram commutes by the definition of $\overline{\beta}$.
	% Such an $\iota$ is unique by the universal property of the adjunction.

	Now to show that the functor $F/\CC \to \DD \times \CC$ is a $\Lexp$-functor, we need to show that it 
	preserves finite limits, the exponentiable morphism, and the pushforwards along it.
	Preservation of finite limits is a general fact about comma categories $F/G$ when $G$ preserves them.
	Finally, preservation of the exponentiable morphism and pushforwards along it follow directly from the construction as the right adjoint $P$ is
	defined such that it preserves the action of its underlying polynomial functors $P_\CC$ and $P_\DD$. 
\else
% The adjunction is given by transposition: a morphism $\beta : L\alpha \to g$ in $F/\CC$
% corresponds to its transpose $\alpha \to Pg$ in $(F/\CC)/x$, with counit $\epsilon_g$.

The functor $F/\CC \to \DD \times \CC$ preserves finite limits, as does any projection
from a comma category $F/G$ with $G$ preserving them, and it preserves the
exponentiable morphism and pushforwards along it since $P$ is built from
$P_\CC$ and $P_\DD$.
\fi
\end{proof}

\iflong
We now prove \cref{thm:universal-property-natural-transformation}, recalled below.
\universalNaturalTransformation*
\fi
\begin{proof}[Proof of \cref{thm:universal-property-natural-transformation}]
% \Cref{prop:comma-lexp} entails the existence of the natural transformation.
Let us start with existence of the natural transformation.
The given pullback turns $F/ ℂ$ into a $\Lexp$-category, by 
\cref{prop:comma-lexp}. Thus, we get a bi-initial $\Lexp$-functor 
$\FinGat → F/ℂ$. 
% By universal property of the comma category~\cite[Proposition 1.6.3]{BorceuxI}, 
Let $I : \FinGat → \FinGat$ and $G' : \FinGat → ℂ$ be induced 
by the composition of $\Lexp$-functors $\FinGat → F/ℂ → \FinGat × ℂ$.
It follows that $I$ and $G'$ are  $\Lexp$-functors, and are thus isomorphic to the identity endofunctor 
and $G$ respectively. Because the projection $F/ℂ → \FinGat × ℂ$ is an isofibration, 
we get a $\Lexp$-functor $H ∶ \FinGat → F / ℂ$ for which $G'$ and $I$ defined as above are 
exactly $G$ and the identity functor respectively.
The component of the desired natural transformation at $Γ$ is then given by taking the image of 
$Γ$ by $H$.

For uniqueness, consider another natural transformation $β$ between $F$ and $G$, with the same pullback naturality square for $\pGAT$.
By universal property of the comma category~\cite[Proposition 1.6.3]{BorceuxI}, 
this induces a functor $\FinGat → F/ℂ$ such that post composition with either projection 
from $F/ℂ$ yields the identity functor or $G$. It can be checked that it is a $\Lexp$-functor,
because $G$ is. Therefore, it is isomorphic to $H$. This isomorphism is easily seen to be an identity,
because its whiskering with $F/ℂ → \FinGat × ℂ$ is, by uniqueness of the isomorphism between 
two bi-initial $\Lexp$-functors.
\end{proof}

\subsection{The Semantic Translation}
\label{sec:proof-semantic-translation}
We prove \cref{thm:colax-exponentiable} through a preliminary step:
we first establish the corresponding results for the colax slice category
$\colaxSlice$, which is like $\colaxSliceCart$ but without the cartesian
restriction on morphisms (\cref{prop:colax-slice-limits,prop:colax-exponentiable}),
and then show that they restrict to $\colaxSliceCart$ (\cref{prop:cartesian-colax-lexp}).
Let $\AP\intro*\colaxSlice$ denote the \emph{colax slice category}
of locally small categories over $\FamS$: objects are functors to $\FamS$ and a morphism between $ℂ \xrightarrow{F} \FamS$ and $ℂ' \xrightarrow{F'} \FamS$
is a functor $G:ℂ → ℂ'$ and a natural transformation from $F' \circ G$ to $F$. Composition is defined in the obvious way.
\begin{proposition}
    \label{prop:colax-slice-limits}
    $\colaxSlice$ has limits and the functor $\colaxSlice → \CAT$
  preserves them.
  \end{proposition}
  \begin{proof}
This is because every fiber has limits and the reindexing functors preserve them, as they are given by precomposition.
  \end{proof}

\begin{proposition}
  \label{prop:colax-exponentiable}
  The forgetful functor from pointed sets to sets, equipped with the identity natural transformation between the family functors,
  is exponentiable in $\colaxSlice$, and the functor $\colaxSlice → \CAT$ preserves pushforwards along it.
  Therefore, combining with \cref{prop:colax-slice-limits}, $\colaxSlice$ is a $\Lexp$-category and the functor $\colaxSlice → \CAT$ is a $\Lexp$-functor.
\end{proposition}
  \begin{proof}
    We start by showing that 
  \iflong
    the following morphism is exponentiable in $\colaxSlice$:
    \[\begin{tikzcd}[ampersand replacement=\&, column sep=small]
	{\PtdSet} \&\& \Set \\
	\& \Fam,
	\arrow["p_{\CAT}", from=1-1, to=1-3]
	\arrow[""{name=0, anchor=center, inner sep=0}, from=1-1, to=2-2]
	\arrow["\delta", from=1-3, to=2-2]
	\arrow["\id", between={0.5}{0.7}, Rightarrow, from=1-3, to=0]
	\end{tikzcd}\]
  where $\delta$ 
  \else 
  the morphism $\overline{p_{\CAT}} = (p_{\CAT} : \PtdSet \to \Set, \id )$ from $\delta \circ p_{\CAT}$ to  $\delta $ is exponentiable in $\colaxSlice$, where $\delta : \Set \to \Fam$ 
  \fi 
  maps a set $X$ to the family defined by $U = 1 = \{*\}$ and $El(*) = X$, and $p_{\CAT}$ is the forgetful functor from pointed sets to sets.
    By \cref{prop:exp-mor-P-def}, we need to show that the composite
    \[\begin{tikzcd}[ampersand replacement=\&]
      {(\colaxSlice)/\delta} \& {(\colaxSlice)/(\delta \circ p_{\CAT})} \& \colaxSlice
      \arrow["{\overline{p_{\CAT}}^*}", from=1-1, to=1-2]
      \arrow["\dom", from=1-2, to=1-3]
    \end{tikzcd}\]
    has a right adjoint.
    Similar to the case of $\CAT$, %explained in \cref{exa:CAT},
     this composite maps an object $(F : \CC \to \Fam, U : \CC \to \Set, \gamma : \delta \circ U \Rightarrow F)$
    in $(\colaxSlice)/\delta$ to the object $\int U \to \CC \xrightarrow{F} \Fam$ in $\colaxSlice$, where $\int U$ is the category of elements of $U$, and
    the morphism $\int U \to \CC$ is the associated discrete opfibration.

    Before defining the right adjoint, we note that a morphism 
    \iflong 
    of the form
    \[\begin{tikzcd}[ampersand replacement=\&, column sep=small]
	\CC \&\& \DD \\
	\& \Fam
	\arrow["H", from=1-1, to=1-3]
	\arrow[""{name=0, anchor=center, inner sep=0}, "F"', from=1-1, to=2-2]
	\arrow["G", from=1-3, to=2-2]
	\arrow["\alpha", between={0.4}{0.7}, Rightarrow, from=1-3, to=0]
    \end{tikzcd}\]
    \else 
    $(H, \alpha) : F \to G$
    \fi
    in $\colaxSlice$ can be equivalently described as the following commuting square, using the equivalence between $\Fam$ and $\Set^\to$:
    \[\begin{tikzcd}[ampersand replacement=\&]
	{G_1 H c} \& {F_1 c} \\
	{G_0 H c} \& {F_0 c}.
	\arrow["{\alpha_1 c }", from=1-1, to=1-2]
	\arrow[from=1-1, to=2-1]
	\arrow[from=1-2, to=2-2]
	\arrow["{\alpha_0 c}"', from=2-1, to=2-2]
    \end{tikzcd}\]
%     For example, the natural transformation $\gamma$ in the above object $(F, U, \gamma)$ in $(\colaxSlice)/\delta$
%     corresponds to the following commuting square:
%     \[\begin{tikzcd}[ampersand replacement=\&]
% 	{U_c} \& {F_1 c} \\
% 	1 \& {F_0 c}.
% 	\arrow["{\gamma_1 c}", from=1-1, to=1-2]
% 	\arrow[from=1-1, to=2-1]
% 	\arrow[from=1-2, to=2-2]
% 	\arrow["{\gamma_0 c}"', from=2-1, to=2-2]
% \end{tikzcd}\]

    The claimed right adjoint maps an object $\DD \xrightarrow{G} \Fam$ in $\colaxSlice$ to an object
    $(F' : P\DD \to \Fam, U' : P\DD \to \Set, \gamma': \delta \circ U' \Rightarrow F')$ in $(\colaxSlice)/\delta$,
    where $P$ is $P_{p_{\CAT}}$,
    % defined in \cref{exa:CAT},
    $U'$ is defined as the composite $P\DD \xrightarrow{P(G)} P\Fam \to \Set$,
    $F'$ maps an object $(X \in \Set, f : X \to \DD)$ to the family identified by
    \begin{tikzcd}[ampersand replacement=\&]
	{X + \sum_x G_1 f(x)} \\
	{1 + \sum_x G_0 f(x)}
	\arrow[from=1-1, to=2-1]
    \end{tikzcd},
    and $\gamma'$ is identified by the following commuting square:
    \[\begin{tikzcd}[ampersand replacement=\&]
	X \& {X + \sum_x G_1 f(x)} \\
	1 \& {1 + \sum_x G_0 f(x)}.
	\arrow["{\iota_l}", from=1-1, to=1-2]
	\arrow[from=1-1, to=2-1]
	\arrow[from=1-2, to=2-2]
	\arrow["{\iota_l}"', from=2-1, to=2-2]
    \end{tikzcd}\]

    For the adjunction, it suffices to construct the homset bijection and show naturality in $A \in (\colaxSlice)/\delta$.
    Naturality in $A$ implies a unique way to extend the map on objects to a right adjoint.

    For the homset bijection, let $(F : \CC \to \Fam, U : \CC \to \Set, \gamma : \delta \circ U \Rightarrow F)$ be an object in $(\colaxSlice)/\delta$ and
    $\DD \xrightarrow{G} \Fam$ be an object in $\colaxSlice$.
    We need to show a bijection between the homsets identified by the following two commuting squares:
    \[
    \begin{tikzcd}[ampersand replacement=\&]
	{G_1H(c,u)} \& {F_1 c} \\
	{G_0H(c,u)} \& {F_0 c}
	\arrow["{\alpha_1 (c,u)}", from=1-1, to=1-2]
	\arrow[from=1-1, to=2-1]
	\arrow[from=1-2, to=2-2]
	\arrow["{\alpha_0 (c,u)}"', from=2-1, to=2-2]
    \end{tikzcd},
    \begin{tikzcd}[ampersand replacement=\&]
	{U_c + \sum_uG_1H(c,u)} \& {F_1 c} \\
	{1 + \sum_uG_0 H(c,u)} \& {F_0 c}
	\arrow["{\beta_1 c}", from=1-1, to=1-2]
	\arrow[from=1-1, to=2-1]
	\arrow[from=1-2, to=2-2]
	\arrow["{\beta_0 c}"', from=2-1, to=2-2]
    \end{tikzcd}\]
\iflong
    such that the precomposition of $\beta$ by $\gamma'$ is equal to $\gamma$:
    \[\begin{tikzcd}[ampersand replacement=\&]
	{U_c} \& {U_c + \sum_uG_1H(c,u)} \& {F_1 c} \\
	1 \& {1 + \sum_uG_0 H(c,u)} \& {F_0 c}
	\arrow["{\iota_l}", from=1-1, to=1-2]
	\arrow[from=1-1, to=2-1]
	\arrow["{\beta_1 c}", from=1-2, to=1-3]
	\arrow[from=1-2, to=2-2]
	\arrow[from=1-3, to=2-3]
	\arrow["{\iota_l}"', from=2-1, to=2-2]
	\arrow["{\beta_0 c}"', from=2-2, to=2-3]
    \end{tikzcd}
    =
    \begin{tikzcd}[ampersand replacement=\&]
	{U_c} \& {F_1 c} \\
	1 \& {F_0 c}.
	\arrow["{\gamma_1 c}", from=1-1, to=1-2]
	\arrow[from=1-1, to=2-1]
	\arrow[from=1-2, to=2-2]
	\arrow["{\gamma_0 c}"', from=2-1, to=2-2]
    \end{tikzcd}
    \]
    Note that the equation $\beta \circ \gamma' = \gamma$ implies that
    $\beta_1 c \circ \iota_l = \gamma_1 c$ and $\beta_0 c \circ \iota_l = \gamma_0 c$.
    Therefore, the data of $\beta$ is equivalent to that of the following commuting square:
    \[\begin{tikzcd}[ampersand replacement=\&]
	{\sum_uG_1H(c,u)} \& {F_1 c} \\
	{\sum_uG_0 H(c,u)} \& {F_0 c}
	\arrow["{\beta_{r_1} }", from=1-1, to=1-2]
	\arrow[from=1-1, to=2-1]
	\arrow[from=1-2, to=2-2]
	\arrow["{\beta_{r_0} }"', from=2-1, to=2-2]
    \end{tikzcd}\]
\else
    such that $\beta \circ \gamma' = \gamma$, which amounts to
    $\beta_1 c \circ \iota_l = \gamma_1 c$ and $\beta_0 c \circ \iota_l = \gamma_0 c$.
    Therefore, the data of $\beta$ is equivalent to the commuting square given by its restriction 
    $\beta_r$ to $\sum_u G_i H(c,u)$.
\fi

    Finally, by the universal property of the coproduct $\sum_u$,
    specifying $\beta_{r_i} : \sum_uG_i H(c,u) \to F_i c$ (for $i \in \{0,1\}$) is equivalent to specifying for each $u \in U_c$, a morphism
    $G_i H(c,u) \to F_i c$.
    These families of morphisms are exactly the data of the natural transformation $\alpha$.

	The naturality of $\alpha$ follows from the naturality of $\beta$ and the construction above.
	In particular, $\beta$ is natural in $c$,
	% \[
	% \begin{tikzcd}[ampersand replacement=\&]
	% {U_c + \sum_{u\in U_c} G H(c,u)} \& {F_1(c)} \\
	% {U_{c'} + \sum_{u'\in U_{c'}} G H(c',u')} \& {F_1(c')}
	% \arrow["{\beta_{1,c}}", from=1-1, to=1-2]
	% \arrow["{U(h) + \sum GH(h)}"', from=1-1, to=2-1]
	% \arrow["{F_1(h)}", from=1-2, to=2-2]
	% \arrow["{\beta_{1,c'}}"', from=2-1, to=2-2]
	% \end{tikzcd} \]
	% \[
	% \begin{tikzcd}[ampersand replacement=\&]
	% {1 + \sum_{u\in U_c} G_0 H(c,u)} \& {F_0(c)} \\
	% {1 + \sum_{u'\in U_{c'}} G_0 H(c',u')} \& {F_0(c')}
	% \arrow["{\beta_{0,c}}", from=1-1, to=1-2]
	% \arrow["{\id_1 + \sum G_0H(h)}"', from=1-1, to=2-1]
	% \arrow["{F_0(h)}", from=1-2, to=2-2]
	% \arrow["{\beta_{0,c'}}"', from=2-1, to=2-2]
	% \end{tikzcd}
	% \]
	and by restricting this naturality square
	along the coproduct injections, we obtain the naturality of $\alpha$.
	% \[
	% \begin{tikzcd}[ampersand replacement=\&]
	% {G H(c,u)} \& {F_1(c)} \\
	% {G H(c',u')} \& {F_1(c')}
	% \arrow["{\alpha_{1,(c,u)}}", from=1-1, to=1-2]
	% \arrow["{G H(h,u)}"', from=1-1, to=2-1]
	% \arrow["{F_1(h)}", from=1-2, to=2-2]
	% \arrow["{\alpha_{1,(c',u')}}"', from=2-1, to=2-2]
	% \end{tikzcd}
	% \]
	% \[
	% \begin{tikzcd}[ampersand replacement=\&]
	% {G_0 H(c,u)} \& {F_0(c)} \\
	% {G_0 H(c',u')} \& {F_0(c')}
	% \arrow["{\alpha_{0,(c,u)}}", from=1-1, to=1-2]
	% \arrow["{G_0 H(h,u)}"', from=1-1, to=2-1]
	% \arrow["{F_0(h)}", from=1-2, to=2-2]
	% \arrow["{\alpha_{0,(c',u')}}"', from=2-1, to=2-2]
	% \end{tikzcd}
	% \]

\iflong
    It remains to show the naturality of this bijection in $A$, that is
    for every morphism
    \(f : A \to A'\) in \((\colaxSlice)/\delta\) and every object
    \(B \in \colaxSlice\), the following square commutes:
    \[\begin{tikzcd}[ampersand replacement=\&]
	{\Hom_{\colaxSlice}(L A',\, B)} \& {\Hom_{(\colaxSlice)/\delta}(A',\, RB)} \\
	{\Hom_{\colaxSlice}(L A,\, B)} \& {\Hom_{(\colaxSlice)/\delta}(A,\, RB),}
	\arrow["{\Psi_{A',B}}", from=1-1, to=1-2]
	\arrow["{{(-)\circ L(f)}}"', from=1-1, to=2-1]
	\arrow["{{(-)\circ f}}", from=1-2, to=2-2]
	\arrow["{\Psi_{A,B}}"', from=2-1, to=2-2]
    \end{tikzcd}\]
    where $L$ is the composition of the pullback along $p_{\CAT}$ and $\dom$, and $R$ is $P_{\colaxSlice}$, the right adjoint defined above.
    The bijection $\Psi$ acts by taking coproducts indexed by the elements of $U_c$ and copairing with $\gamma$, both of which commute with reindexing of the
    indexing sets $U_c$ induced by $f$.
    Hence, the above square commutes.
\else
    This bijection is natural in $A$, since it acts by taking coproducts
    indexed by the elements of $U_c$ and copairing with $\gamma$, both of which
    commute with reindexing of the indexing sets $U_c$.
\fi

    This concludes proving the first part of the proposition.
    Showing that the forgetful functor $\colaxSlice → \CAT$ preserves pushforwards along the exponentiable morphism is straightforward,
    as $P_{\colaxSlice}$ is $P_{\CAT}$ on the level of the underlying categories, by definition.
%     Now, we show that the forgetful functor $U : \colaxSlice → \CAT$ preserves pushforwards along this exponentiable morphism.
%     By \cref{lem:preservation_pushforward}, this is to show that for any object $\DD \xrightarrow{G} \Fam$ in $\colaxSlice$, the following canonical morphism is an isomorphism:
%     \[\begin{tikzcd}[ampersand replacement=\&]
% 	\colaxSlice \& {(\colaxSlice)/\delta} \\
% 	{\CC'} \& {\CC'/FX}
% 	\arrow["P_{\colaxSlice}", from=1-1, to=1-2]
% 	\arrow["U"', from=1-1, to=2-1]
% 	\arrow["\iso", between={0.3}{0.7}, Rightarrow, from=1-2, to=2-1]
% 	\arrow["{U}", from=1-2, to=2-2]
% 	\arrow["{P_{\CAT}}"', from=2-1, to=2-2]
%   \end{tikzcd}\]
%   This is straightforward as $P_{\colaxSlice}$ is $P_{\CAT}$ on the level of the underlying categories, by definition.
  \end{proof}

\iflong
  We now recall \cref{thm:colax-exponentiable}.
\colaxExponentiable*

This is an immediate corollary of \cref{prop:colax-slice-limits,prop:colax-exponentiable}
and the following result.
\else
\Cref{thm:colax-exponentiable} is an immediate corollary of
\cref{prop:colax-slice-limits,prop:colax-exponentiable} and the following result.
\fi
\begin{proposition}
  \label{prop:cartesian-colax-lexp}
  The category
  $\colaxSliceCart$ is a sub-$\Lexp$-category of $\colaxSlice$, in the sense that it has finite limits, the exponentiable morphism from \cref{prop:colax-exponentiable} is in $\colaxSliceCart$ and is again exponentiable, and finally, the functor $\colaxSliceCart → \colaxSlice$ preserves finite limits and pushforwards along this exponentiable morphism.
\end{proposition}
\begin{proof}
	It is straightforward to check that the terminal object in $\colaxSlice$ is also in
	$\colaxSliceCart$.
	As for pullbacks, they can be computed in $\colaxSlice$ as follows:
	first compute the pullback of the categories
	in $\CAT$. The functor from this pullback to $\Fam$ is given
	by the pushout of the whiskering of the given natural transformations by the pullback projections.
	If these natural transformations are cartesian, we
	can compute this pushout so that the pushout coprojections are also cartesian. This results
	in a pullback in $\colaxSliceCart$ which is also a pullback in $\colaxSlice$.
	It follows that $\colaxSliceCart$ has finite limits and that the inclusion
	$\colaxSliceCart → \colaxSlice$ preserves them.

	The proof of \cref{prop:colax-exponentiable} can directly be extended to $\colaxSliceCart$
	to show that the chosen exponentiable morphism of $\colaxSlice$ is also exponentiable in $\colaxSliceCart$.
	It is straightforward that the inclusion preserves pushforwards along this exponentiable morphism.
    % To show that the exponentiable morphism of $\colaxSlice$ is also exponentiable in  $\colaxSliceCart$,
    % we extend the proof of  to the case of \colaxSliceCart.
    % Note that a morphism of the form
    % \[\begin{tikzcd}[ampersand replacement=\&, column sep=small]
	% \CC \&\& \DD \\
	% \& \Fam
	% \arrow["H", from=1-1, to=1-3]
	% \arrow[""{name=0, anchor=center, inner sep=0}, "F"', from=1-1, to=2-2]
	% \arrow["G", from=1-3, to=2-2]
	% \arrow["\alpha", between={0.4}{0.7}, Rightarrow, from=1-3, to=0]
    % \end{tikzcd}\]
    % in $\colaxSlice$ is \kl{cartesian} if and only if the following commuting square, explained in the proof of \cref{thm:colax-exponentiable}, is a pullback:
    % \[\begin{tikzcd}[ampersand replacement=\&]
	% {G_1 H c} \& {F_1 c} \\
	% {G_0 H c} \& {F_0 c}.
	% \arrow["{{\alpha_1 c }}", from=1-1, to=1-2]
	% \arrow[from=1-1, to=2-1]
	% \arrow["\lrcorner"{anchor=center, pos=0.125}, draw=none, from=1-1, to=2-2]
	% \arrow[from=1-2, to=2-2]
	% \arrow["{{\alpha_0 c}}"', from=2-1, to=2-2]
    % \end{tikzcd}\]
    % The proof of \cref{thm:colax-exponentiable} can be extended to this case by
    % having all such commuting squares being pullbacks, as this does not
    % change the construction of the right adjoint.

\end{proof}
\iflong
We recall \cref{prop:semantic-twosort-functor}.
\semTwosortFunctor*
\fi
\begin{proof}[Proof of \cref{prop:semantic-twosort-functor}]
	The induced functor from $\colaxSliceCart$ to $\CAT/\Fam$ can be characterised as
	the right adjoint of the obvious embedding of $\CAT/\Fam$ into $\colaxSliceCart$.

	We abstract the situation as follows:
	we have an adjunction $L ⊣ R$, with $L ∶ 𝔻 → ℂ$ and $R ∶ ℂ → 𝔻$, between categories such that
	\begin{itemize}
		\item $(ℂ, Y_ℂ \xrightarrow{p_ℂ} X_ℂ)$ and $(𝔻, Y_𝔻 \xrightarrow{p_𝔻} X_𝔻)$ are $\Lexp$-categories;
		\item
		      $R  p_ℂ ≅ p_𝔻$;
		\item the left adjoint $L$ preserves pullbacks;
		\item the naturality square of the counit for the exponentiable morphism $p_ℂ$ is a pullback.
	\end{itemize}
	We can show that the right adjoint $R$ preserves pushforwards along $p_ℂ$
	and is thus a $\Lexp$-functor as desired.
	Consider indeed the square below commuting up to isomorphism, where the vertical functors are induced by
	$R$ and the isomorphism $R p_ℂ ≅ p_𝔻$.
	\[
	\begin{tikzcd}
	ℂ/X_ℂ \ar[r, "p_{ℂ}^*"] \ar[d] & ℂ/Y_ℂ \ar[d] \\
	𝔻/X_𝔻 \ar[r, "{p_{𝔻}^*}"'] & 𝔻/Y_𝔻
	\end{tikzcd}
	\]
	Preserving pushforwards along $p_ℂ$ requires the mate natural transformation
	to be an isomorphism.
	Note that we have another mate natural transformation, induced by the left adjoints of the vertical functors.
	It is easy to check that this one is an isomorphism.
	 The result follows from the fact that in this situation, one mate natural transformation is an isomorphism if and only if the other one is.
\end{proof}

\iflong

\section{Initial Models of GATs}
  \label{sec:initial-models-gats}
  In this section, we provide a direct proof 
  that any theory has an initial model, exploiting
  the bi-initiality property of $\FinGat$.

  % with an initial model $\initModel{Ω}$.
  The first step consists in constructing 
a suitable model of a theory $Ω$.
As in \cite{DBLP:journals/pacmpl/KaposiKA19}, we exploit
the Yoneda lemma~\cite[§III.2]{MacLane:cwm}: 
an object of $\model{Ω}$ is equivalently given by a natural 
transformation from the representable functor $\hom(Ω, -)$ to the 
model functor. The point is that we can use \cref{thm:universal-property-natural-transformation}
to construct such a natural transformation.
Before, we introduce some convenient type-theoretic notations to denote morphisms to $\XGAT$ and $\YGAT$ in $\FinGat$.

  In the rest of this section, we assume given a theory $Ω$, 
  the goal being to construct an initial model of it.
\begin{notation}
  \label{not:mor-terms-set}
  Given $A ∶ Ω → \XGAT$, we denote the set of morphisms $t ∶Ω → \YGAT$ such that 
  $\pGAT ∘ t = A$ by $\AP\intro*\TmSet{A}$.
  The notation is motivated by the fact that 
  a morphism $Ω → \XGAT$ corresponds to 
  a term $Ω ⊢ A ∶ \GatSet$, and a morphism $Ω → \YGAT$ above $\XGAT$ 
  corresponds to a term $Ω ⊢ t ∶ A$.
  % Given $Ω ⊢ A ∶ \GatSet$, we denote the corresponding morphism $Ω → \XGAT$ by 
  % $\AP\intro*A$.
  %  Similarly, given such $A$ and a term $Ω ⊢ t ∶ A$,
  %  we denote the corresponding morphism $Ω → \YGAT$ by 
  %  $\intro*t$. 
  %  Note that 
  %  $\pGAT ∘ t = A$. 

  Given a model $ω$ of $Ω$ and element $t ∈ \TmSet{A}$, the pointed set $\model{t}(ω)$ is
necessarily $\model{A}(ω)$, with 
a distinguished element denoted by
$\AP\intro*\modelTerm{t}{ω} $. 
\end{notation}
% Note that any morphism $Ω → \XGAT$ is $A$ for a unique $A$, 
% and similarly any morphism $Ω → \YGAT$ is $t$ for a unique $t$.
% \begin{remark}
%   We could define the judgements $Ω ⊢ A ∶ \GatSet$ and $Ω ⊢ t ∶ A$
%   directly as morphisms $Ω → \XGAT$ and $Ω → \YG
% \end{remark}
As a direct consequence of the functoriality of 
 $\model{t}$, we get naturality of $\modelTerm{t}{ω}$ in $ω$:
\begin{lemma}
  \label{lem:model-term-natural}
  
Given any morphism of models $f ∶ ω → ω'$, we have
$\model{A}(f)(\modelTerm{t}{ω}) = \modelTerm{t}{ω'}$.
\end{lemma}
We now apply 
\cref{thm:universal-property-natural-transformation}
to construct our candidate initial model, using the above notations.
% \begin{definition}
%   To each theory $Γ$, we associate the family 
%   $(Sort_Γ, Tm_Γ)$ defined by $U_Γ = \hom(Γ, A:𝐒𝐞𝐭)$ and
% \end{definition}
% where we consider a term $Ω ⊢ A : \GatSet$ as a morphism $Ω → \XGAT$,
% and similarly, given such morphism, we consider a term $Ω ⊢ t : A$ as a morphism $Ω → \YGAT$ over $\XGAT$.
\begin{proposition}
  \label{prop:initial-model-gat-existence}
  There is a unique model $0_Ω$ of $Ω$
  such that, 
  \begin{itemize}
    \item for any $A ∶ Ω → \XGAT$, 
    the set $\model{A}(0_Ω)$ is $\TmSet{A}$;
    \item for any $t ∈ \TmSet{A}$, 
    the element $\modelTerm{t}{0_Ω}$ of $\model{A}(0_Ω)$ is $t$ itself. 
  \end{itemize}
\end{proposition}

% \begin{remark}
%   Up to iso, we can reformulate the above conditions to avoid any reference 
%   to the type-theoretic syntax of $Ω$. The first one says that 
%    for any morphism $A ∶ Ω → \XGAT$, the set $\model{A}(0_Ω)$ consists of 
%    $Ω → \YGAT$ over $\XGAT$.
%    Similarly, the second condition says that
%    for any morphism $t ∶ Ω → \YGAT$, the pointed set
   
% \end{remark}
  \begin{proof}%[Proof of \cref{prop:initial-model-gat-existence}]
    % Consider the representable functor $\hom(Ω, -)$
    % composed with the embedding of sets into $\CAT$, as discrete categories.
% By the Yoneda lemma, models of $Ω$
% are in one-to-one correspondence with natural transformations from $\hom(Ω,-)$ to the \kl{model functor}:
Any model $ω$ induces a natural transformation from $\hom(Ω,-)$ to the \kl{model functor}
such that $\hom(Ω, Γ) → \model{Γ}$ maps a morphism $σ ∶ Ω → Γ$ to $\model{σ}(ω)$.
In this respect, the two conditions determines the following naturality square for $\pGAT$.
\begin{equation}
  \label{eq:naturality-square-initial-model}
% YADE DIAGRAM naturality-square-initial-model.yade
% GENERATED LATEX
\input{diagrams/naturality-square-initial-model.tex}
% END OF GENERATED LATEX
\end{equation}
    By \cref{thm:universal-property-natural-transformation}, 
    there is a unique natural transformation from 
    $\hom(Ω, -)$ to the model functor, and by the Yoneda lemma,
    it corresponds to a unique model $0_Ω$ of $Ω$, 
    recovered from the natural transformation by 
    evaluating it at $id_Ω ∈ \hom(Ω, Ω)$.
  \end{proof}

  Next, we assume given a model $ω$ of $Ω$.
  To prove that $0_Ω$ is initial, we need to construct a morphism from $0_Ω$ to $ω$ 
  and show that it is unique. Note that a morphism in $\model{Ω}$ is an object 
  of its category $\model{Ω}^→$ of arrows.
  This suggests using a similar strategy 
  to construct a natural transformation
  from $\hom(Ω, -)$ to $\model{-}^→$. To apply \cref{thm:universal-property-natural-transformation},
  we need $\FinGat \xrightarrow{\model{-}} \CAT \xrightarrow{(-)^→} \CAT$ to be a $\Lexp$-functor. 
  By definition, the \kl{model functor} is a $\Lexp$-functor, so we only need to ensure that $(-)^→$ is.
  But clearly, $(-)^→$ does not map the exponentiable functor
  $p_{\CAT} ∶\PtdSet → 𝐒𝐞𝐭$ to itself, so it cannot be a $\Lexp$-functor, unless 
  we consider the codomain $\CAT$ with a different $\Lexp$-structure, 
  choosing the exponentiable morphism to be precisely 
  the image of $p_{\CAT}$ by $(-)^→$.
  Unfortunately, $(-)^→ ∶ \CAT → \CAT$ is not 
  a $\Lexp$-functor between $(\CAT,p_{\CAT})$ and $(\CAT,p^→_{\CAT})$, because it does not 
  \kl{preserve pushforwards} along $\PtdSet → 𝐒𝐞𝐭$.
However, the following variant works.
% of $(-)^→ ∶ \CAT → \CAT$ works.
  \begin{proposition}
    \label{prop:arrow-functor-lexp}
    Let $\cod_- ∶ \CAT → \CAT^→$ denote the functor mapping a category $C$ to the codomain functor $\cod_C ∶ C^→ → C$.
    This induces a $\Lexp$-functor between
    the $\Lexp$-categories $(\CAT,p_{\CAT})$ and $(\CAT^→, \cod_-(p_{\CAT}))$.
  \end{proposition}
  \begin{proof}
    The functor $P_{\cod_-(p_{\CAT})} ∶ \CAT^→ → \CAT^→/\cod_-(p_{\CAT})$ maps a functor $F ∶ C → D$ to 
    the blue functor below left, equipped with the below right morphism to $𝐒𝐞𝐭^→ \xrightarrow{\cod} 𝐒𝐞𝐭$, where
    \begin{itemize}
      \item 
      $𝐒𝐞𝐭/\FamCat{D}$ denotes the comma category of $𝐒𝐞𝐭 \xrightarrow{id} 𝐒𝐞𝐭$ and 
  $\FamCat{D} → 𝐒𝐞𝐭$: objects are triples 
      $(X,(d_y)_{y ∈ Y}, X \xrightarrow{f} Y)$ of a set $X$, a family $(d_y)_{y ∈ Y}$ of objects  
    of $D$ and a function $f ∶ X → Y$;
    \item the functor $𝐒𝐞𝐭/\FamCat{D} → \FamCat{D}$ on the left 
    maps such an object to the family $(d_{f(x)})_{x ∈ X}$.
    \end{itemize}
    \[
    % YADE DIAGRAM famcat-arrow-functor.yade
    % GENERATED LATEX
    \input{diagrams/famcat-arrow-functor.tex}
    % END OF GENERATED LATEX
    \]
    It can be checked that this is indeed the pushforward of $F$ along $\cod_-(p_{\CAT})$,
    and moreover, 
    that $\cod_-$ \kl{preserves pushforwards} along $p_{\CAT}$.
    % $\dom_-(P_{p_{\CAT}}(F))$ is canonically isomorphic to $P_{\cod_-(p_{\CAT})}(\dom_-(F))$.
  \end{proof}
Now we are able to construct a unique morphism from $0_Ω$ to $ω$ of $Ω$, exploiting 
  \cref{thm:universal-property-natural-transformation} to specify 
  a suitable natural transformation 
  between the functor $\FinGat → \CAT^→$ 
  mapping $Γ$ to the function $\model{-}(ω)∶\hom(Ω, Γ) → \model{Γ}$,
  and $\FinGat \xrightarrow{\model{-}} \CAT \xrightarrow{\cod_-} \CAT^→$.
  \begin{proposition}
    \label{prop:initial-model-gat-morphism-existence}
    There exists a unique morphism $i$ from $0_Ω$ to $ω$ in $\model{Ω}$
     such that, 
     for any $A ∶Ω → \XGAT$, 
    the function $\model{A}(i) ∶ \model{A}(0_Ω) → \model{A}(ω)$ 
    maps $t ∈ \TmSet{A}$ to $\modelTerm{t}{ω}$.
  \end{proposition}
  \begin{proof}
    The proof is similar to that of \cref{prop:initial-model-gat-existence},
    but we additionally need to ensure that the morphism induced by the 
    natural transformation $\hom(Ω, -) → \model{-}^→$ by the Yoneda lemma is indeed a morphism 
    between $0_Ω$ and $ω$.
    For the domain, this follows from \cref{prop:initial-model-gat-existence}.
    For the codomain, this is by definition of a morphism in $\CAT^→$.
  \end{proof}
  Uniqueness of the morphism comes from the following simple observation.
  \begin{lemma}
    Any morphism from $0_Ω$ to $ω$ satisfies the condition of 
    \cref{prop:initial-model-gat-morphism-existence}.
  \end{lemma}
  \begin{proof}
    Consider a morphism $f ∶ 0_Ω → ω$.
    By \cref{lem:model-term-natural}, we have
    $\model{A}(f)(\modelTerm{t}{0_Ω}) = \modelTerm{t}{ω}$.
    By \cref{prop:initial-model-gat-existence}, we have $\modelTerm{t}{0_Ω} = t$, so the conclusion follows.
  \end{proof}
  \begin{corollary}
    \label{cor:initial-model-gat}
    The model $0_Ω$ is initial in $\model{Ω}$: there exists a unique morphism from $0_Ω$ to any model $ω$ of $Ω$.
  \end{corollary}

  \section{Two-sortification is Fully Faithful}
  \label{sec:two-sortification-fully-faithful}
  In this section, we show that the two-sortification functor $\TW$ from $\FinGat$ to $\FinGat /\FamG$ is fully faithful.
  We mostly rely on the following proposition.

\begin{proposition}
  \label{prop:faithful-at-XGAT-YGAT}
  Let $ℂ$ be a $\Lexp$-category
  and $I ∶ \FinGat → ℂ$ be a $\Lexp$-functor.
  % Denoting by $\TmSetF{I}{Ω}$
  Then, $I$ is faithful (resp. fully faithful) if and only if
  $I$ is \intro{faithful (resp. fully faithful) at} $\XGAT$ and $\YGAT$, i.e.,
  for any theory $Ω$, the functions $\hom(Ω, \XGAT) → \hom(I Ω, I \XGAT)$ and $\hom(Ω, \YGAT) → \hom(I Ω, I \YGAT)$
  are injective (resp. bijective).
  % , i.e.,
  % the following morphism of families
  % $(\TmSet{A})_{A ∈ \hom(Ω,\XGAT)}
  % → (\TmSetF{I}{A})_{∈ \hom(IΩ, I \XGAT)}$
  % is a monomorphism (resp. isomorphism), 
  % where \AP $\intro*\TmSetF{I}{A}$ denotes the set of morphisms
  % $t ∶ I Ω → I \YGAT$ such 
  % monic (resp. ):
  % $I$ is \kl{faithful (resp. fully faithful) at} $\XGAT$ and $\YGAT$.
\end{proposition}
\begin{proof}  
  We focus on the hard direction, which is the "if" part.
  %By \cref{lem:localy-faithful}, Clearly,
  We show that $\FinGat$ is equal to its full subcategory $S$ spanned by 
  theories $Γ$ such that $I$ is \kl{faithful (resp. fully faithful) at} $Γ$.
  Note that $S$ includes $\XGAT$ and $\YGAT$ by assumption.
  It is easy to check that $S$ is stable under limits preserved by $I$, and 
  that if $Γ$ is in $S$, then so is $P Γ$.
  Therefore, $S$ is a full $\Lexp$-subcategory 
  of $\FinGat$. By bi-initiality of $\FinGat$, we 
  get a functor from $\FinGat$ to $S$ such that 
  the inclusion composed with this functor is isomorphic to the identity.
  This entails that any GAT is isomorphic to a theory at which $I$ is faithful (resp. fully faithful).
  The claimed result easily follows.
\end{proof}
\Cref{prop:faithful-at-XGAT-YGAT} requires the morphism from $\hom(Ω,\XGAT) \allowbreak → \hom(Ω,\YGAT)$ to 
$\hom(IΩ,I\XGAT) → \hom(IΩ,I\YGAT)$ in $𝐒𝐞𝐭^→$ induced by the functorial action of $I$ to be a monomorphism or an isomorphism.
It will prove useful to rephrase this condition through the equivalence $𝐒𝐞𝐭^→ ≃ \FamS$.
\begin{notation}
   Let $I ∶ \FinGat → ℂ$ be a $\Lexp$-functor
   and $Ω$ be a theory.
  The image of $\hom(IΩ,I\XGAT) → \hom(IΩ,I\YGAT)$
    by the equivalence $𝐒𝐞𝐭^→ ≃ \FamS$ 
    is denoted by
   \AP$\intro*\TmFamF{I}{Ω}$, omitting $I$ when 
  it is the identity functor on $\FinGat$.
  Explicitly, $\TmFamF{I}{Ω}$ is the family 
  $(\TmSetF{I}{A})_{A ∈ \hom(IΩ, I \XGAT)}$,
  where $\intro*\TmSetF{I}{A}$ denotes the 
  set of morphisms $t ∶ IΩ → I\YGAT$ such that
  $I \pGAT ∘ t = A$.
\end{notation}
\begin{corollary}
  \label{cor:faithful-via-terms}
  Let $I ∶ \FinGat → ℂ$ be a $\Lexp$-functor.
  Then, $I$ is faithful (resp. fully faithful) if and only if,
  for any theory $Ω$,
  the family morphism 
  $\TmFamF{}{Ω}
  → \TmFamF{I}{Ω}$
  induced by the functorial action of $I$
  is a monomorphism (resp. isomorphism).
\end{corollary}

Using these results, in \cref{sec:two-sortification-faithful}, we show that $\TW$ is faithful,
and in \cref{sec:two-sortification-full}, we show that it is also full.
\subsection{Faithfulness}
\label{sec:two-sortification-faithful}
In this section, we show that $\model{\TW -}$ is faithful, 
and therefore, so is $\TW$.
We will rather work with the following isomorphic variant of 
$\model{\TW-}$
for which we have an exact description
of the images of $\XGAT$ and $\YGAT$.
\begin{notation}
   We denote 
   the $\Lexp$-functor $\FinGat → \CAT / \FamS$
   from \cref{cor:models-two-sortification-lax-slice-cart}
  mapping a theory $Γ$ to 
  $\cartCommaFam{\FamFunctor{Γ}}$
  by \AP$\intro*\TWb$.
\end{notation}
Recall that $\TWb$ is isomorphic to $\model{\TW -}$ by \cref{thm:uemura}.
Now, for any theory $Ω$, we construct an object of $\TWb Ω$ 
such that given a pair of morphisms $f,g$ from $Ω$ to $\XGAT$ (resp. $\YGAT$),
if its images by ${\TWb f}$ and ${\TWb g}$ are equal, 
then $f$ and $g$ are equal as well. This ensures 
that $\TWb$, and therefore $\model{\TW -}$, is \kl{faithful at} $\XGAT$ and $\YGAT$.
\begin{proposition}
  \label{prop:model-faithful}
  For any theory $Ω$, 
  there is an object \AP$\intro*\M{Ω}$ of $\TWb Ω$ such that 
  \begin{itemize}
    \item  $\bang (\M{Ω})$ is the family
     $\TmFamF{}{Ω}$, 
    where \AP $\intro*\bang ∶ \TWb Ω → \FamS$ is the canonical projection mapping 
    any object $\FamFunctor{Ω}(ω) → A$ to $A$;
    \item for any $A ∶ Ω → \XGAT$, the
    \kl{cartesian} family morphism $\TWb{A}(\M{Ω})$ is 
  the morphism
  % (resp $\tw{\YGAT}{\model{\TW A}}$)
   from $\FamFunctor{\XGAT}(\TmSet{A}) = (\TmSet{A})_{* ∈ \{ * \}}$ 
  to $\TmFamF{}{Ω}$ mapping $*$ to $A ∈ \hom(Ω,\XGAT)$;
  \item for any $(Ω \xrightarrow{t} \YGAT) ∈ \TmSet{A}$,
  the \kl{cartesian} family morphism $\TWb{t}(\M{Ω})$ is
  the same morphism  
  $\FamFunctor{\YGAT}(1 \xrightarrow{t} \TmSet{A}) = (\TmSet{A})_{* ∈ \{ * \}}$
  to $\TmFamF{}{Ω}$ as above,
  where $A$ denotes $\pGAT ∘ t$.
  \end{itemize}
\end{proposition}
\begin{proof}
  The proof is similar to \cref{prop:initial-model-gat-existence}:
  we apply \cref{thm:universal-property-natural-transformation} to 
  construct a natural transformation between 
  \begin{itemize}
    \item the functor from $\FinGat$ to $\CAT / \FamS$ that maps any theory
  $Γ$ to the constant functor $\hom(Ω, Γ) → \FamS$ mapping 
  any morphism to $\TmFamF{}{Ω}$, and
    \item the functor $\TWb$.
  \end{itemize} 
  The Yoneda lemma yields the desired object $\M{Ω}$ of $\TWb Ω$.
\end{proof}
\begin{corollary}
  \label{cor:two-sortification-faithful}
  $\model{\TW -}$ is faithful, and thus, so is $\TW$.
\end{corollary}
\begin{remark}
  \Cref{cor:two-sortification-faithful} raises the question whether the \kl{model functor} is faithful as well. We leave this question open.
\end{remark}
\subsection{Fullness}
\label{sec:two-sortification-full}
In this section, we show that $\TW$ is also full.
Since we showed in the previous section that $\TW$ is faithful,
by \cref{cor:faithful-via-terms},
we know that $\TmFamF{}{Ω} → \TmFamF{\TW}{Ω}$ is a monomorphism for any theory $Ω$. 
Thus, it is enough to construct a section.
We start by replacing $\TmFamF{\TW}{Ω}$ with 
an isomorphic family.
\begin{notation}
  Given a theory $Ω$, we denote 
  the composition $\TW Ω → P \XGAT → \XGAT$ by \AP$\intro*\UGATw{Ω}$.
  Given $(\TW Ω \xrightarrow{A} \YGAT) ∈\TmSet{\UGATw{Ω}}$, 
  we denote the composition 
  % where $! ∶ Ω → 1$ is the unique morphism to the terminal theory.
$\TW Ω → \TW\XGAT \xrightarrow{ε_{\XGAT}} \XGAT$  
by \AP$\intro*\ElMorTw{A}$, where
the first morphism is the one corresponding to $A$ and $\UGATw{Ω}$ by the universal property of the pullback \eqref{eq:p2-diagram} defining $\TW \XGAT$, and 
$ε_{\XGAT}$ is the counit at $\XGAT$
of the adjunction $\dom ∘ \pGAT^* ⊣ P$. 
\end{notation}
\begin{remark}
  The notations are motivated by the fact that 
  an element of $\TmSet{\UGATw{Ω}}$ 
  corresponds to a term $\TW Ω ⊢ A ∶ \U $,
  and an element of $\TmSet{\ElMorTw{A}}$
  corresponds to a term $\TW Ω ⊢ t ∶ \El\,A$,
\end{remark}
\begin{proposition}
  \label{prop:two-sortification-terms-isomorphic}
  Given a theory $Ω$, 
  let \AP $\intro*\TmElFam{Ω}$ denote the family 
  $(\TmSet{\ElMorTw{A}})_{A ∈ \TmSet{\UGATw{Ω}}}$.
  Then, the morphism from
  $\TmFamF{\TW}{Ω}$ to $\TmElFam{Ω}$ 
  mapping $A ∈ \hom(\TW Ω, \TW \XGAT)$
  to $\TW Ω \xrightarrow{A} \TW \XGAT → \YGAT$
  and $(\TW Ω \xrightarrow{t} \TW \YGAT) ∈ \TmSet{A}$
  to $ \TW Ω \xrightarrow{t} \TW \YGAT → \YGAT$
  is an isomorphism, where 
  $\TW \XGAT → \YGAT$ and $\TW \YGAT → \YGAT$ are the projections
  from the pullbacks \eqref{eq:p2-diagram} defining $\TW \XGAT$ and $\TW \YGAT$.
\end{proposition}
\begin{proof}
  This is a consequence of the definitions of $\TW \XGAT$ and $\TW \YGAT$ as suitable pullbacks.
\end{proof}
Thanks to this isomorphism, it is enough to construct a section of 
 $\TmFamF{}{Ω} → \TmFamF{\TW}{Ω} \xrightarrow{≅} \TmElFam{Ω}$.
We exploit initiality of the following object of $\TWb Ω$.
\begin{proposition}
  \label{prop:initial-object-two-sortification}
  For any theory $Ω$, there is an initial object \AP$\intro*\Z{Ω}$ of $\TWb Ω$
  such that:
  \begin{itemize}
    \item $\bang (\Z{Ω})$ is the family
     $\TmElFam{Ω}$;
    \item for any $A ∶ Ω → \XGAT$,
    the \kl{cartesian} family morphism $\TWb{A}(\Z{Ω})$ is 
  the morphism from  $(\TmSet{\ElMorTw{A'}})_{* ∈ \{ * \}}$ to $\TmElFam{Ω}$ mapping $*$ to $A'$, where 
  $A'$ denotes
  $(\TW Ω \xrightarrow{\TW A} \TW \XGAT → \YGAT)$;
  
  \item for any $( Ω \xrightarrow{t} \YGAT) ∈ \TmSet{A}$,
  the \kl{cartesian} family morphism $\TWb{t}(\Z{Ω})$ is
  the same as above, but where 
  $\TmSet{A'}$ is considered 
  pointed by $\TW Ω \xrightarrow{\TW t} \TW \YGAT → \YGAT$.
  \end{itemize}
\end{proposition}
\begin{proof}
  We take $\Z{Ω}$ to be the initial model $0_{\TW Ω}$ of $\TW Ω$
  from \cref{prop:initial-model-gat-existence}, through the isomorphism 
  $\model{\TW Ω} ≃ \TWb Ω$ from \cref{thm:uemura}.
  The result follows by reasoning about the natural transformation 
  $(\hom(\TW Ω, Γ) →  \model{Γ})_Γ$ mapping $σ$ to $\model{σ}{0_{\TW Ω}}$.
  More details are available in \cref{app:proof-two-sortification-fully-faithful}.
\end{proof}
We conclude by the following result.
\begin{proposition}
  \label{prop:section-terms-two-sortification}
  The image of the initial morphism 
  $\Z{Ω} → \M{Ω}$ by $\bang$
  yields a family morphism 
   which 
is a section of the morphism 
  $\TmFamF{}{Ω} → \TmFamF{\TW}{Ω} ≅ \TmElFam{Ω}$.
\end{proposition}
\begin{proof}
  In this proof, we write $\coref{Ω}$ for the \kl{core projection functor} $\TWb Ω → \model{Ω}$, and we conflate an object of $\TWb Γ$ with its underlying \kl{cartesian} family morphism.
  We deduce the proposition from the uniqueness result of \cref{thm:universal-property-natural-transformation}
  to show equality of two suitable natural transformations 
  between functors from $\FinGat$ to
  the $\Lexp$-category $\colaxSlice$ 
  from \cref{prop:colax-exponentiable}.
  The domain of these natural transformations
  is the  
  functor $\FinGat → \colaxSlice$
  mapping $Γ$ to
  the constant function $\hom(Ω, Γ) → \FamCat{𝐒𝐞𝐭}$
  that maps any morphism to $\TmElFam{Ω}$.
  The codomain is the strict initial functor 
  $\FinGat → \colaxSliceCart$ 
  from \cref{cor:family-functor-gat}
  composed with the embedding 
  $\colaxSliceCart → \colaxSlice$
  from \cref{prop:cartesian-colax-lexp},
  mapping $Γ$ to $\FamFunctor{Γ} ∶ \model{Γ} → \FamCat{𝐒𝐞𝐭}$.
  % Note that $\FamFunctor{Γ}$ is the composition $\bang ∘ L_Γ$, where 
  % $L_Γ ∶ \model{Γ} → \TWb{Γ}$ is the left adjoint of the coreflection, 
  % mapping a model $γ$ to the identity family morphism on $\Fam_Γ γ$.

  % Recall  $\Z{Ω}$ by $z 
  The first natural transformation consists
  of the function from $\hom(Ω, Γ)$ to $\model{Γ}$
  mapping any morphism $σ$ to $\model{σ}(\coref{Ω}\Z{Ω})$, 
  and the family morphism 
  $\FamFunctor{Γ}\model{σ}\coref{Ω}\Z{Ω} 
  \xrightarrow{\TWb σ(\Z{Ω})} \TmElFam{Ω}$.
  The second natural transformation is like the first one, but we postcompose 
  the family morphism with $\TmElFam{Ω} → \TmFamF{}{Ω} → \TmElFam{Ω}$.
  We now show that the two natural transformations are equal, 
  thus recovering the desired equality by taking $σ = \id_{Ω}$, since 
  $\TWb \id_Ω(\Z{Ω}) = \Z{Ω}$ is necessarily an isomorphism as $\Z{Ω}$ is initial.
  To that end, we apply \cref{thm:universal-property-natural-transformation}.
  The fact that the naturality squares at $\pGAT$ are pullbacks 
  reduces to the fact that the naturality square at $\pGAT$ after whiskering by 
  $\colaxSlice → \CAT$ are pullbacks.  This is deduced 
   from the fact that $\coref{Ω}\Z{Ω}$ is initial 
  and thus is isomorphic to the initial model of $Ω$ of \cref{prop:initial-model-gat-existence}, 
  which satisfies this pullback property by construction.
  
  It remains to show that the above family morphisms are equal for any 
  morphism $σ$ from $Ω$ to $\XGAT$ or $\YGAT$.
  We focus on the case of $\XGAT$, the case of $\YGAT$ being similar.
  Let $σ ∶ Ω → \XGAT$ be a morphism of GATs.
  By applying 
  $\TWb{σ}$ to the initial morphism $i ∶ \Z{Ω} → \M{Ω}$,
  we get that 
  $\FamFunctor{\XGAT}\model{σ}\coref{Ω}\Z{Ω} 
  \xrightarrow{\TWb σ(\Z{Ω})}  \TmElFam{Ω} 
  \xrightarrow{\bang i} \TmFamF{}{Ω}$
  is equal to the following composition.
  \[
  \FamFunctor{\XGAT}\model{σ}\coref{Ω}\Z{Ω} 
  \xrightarrow{\FamFunctor{\XGAT}\model{σ}\coref{Ω}i} 
  \FamFunctor{\XGAT}\model{σ}\coref{Ω}\M{Ω} 
  \xrightarrow{\TWb σ(\M{Ω})}  \TmFamF{}{Ω}
  \]
  It is easy to check that 
  the two family morphisms are equal on the base sets,
  using the characterisations of $\Z{Ω}$ and $\M{Ω}$.
  Now, for the unique fiber $\TmSet{σ}$ of the domain, 
  consider an element $(Ω \xrightarrow{t} \YGAT) ∈ \TmSet{σ}$.
  Then, by considering $\model{t}\coref{Ω} i$, we 
  get that the image of $t$ by 
  $\FamFunctor{\XGAT}\model{σ}\coref{Ω}i$ is 
  $(\TW Ω \xrightarrow{\TW t} \TW \YGAT → \YGAT )$.
  It follows
  that the two family morphisms are equal on that fiber as well.
\end{proof}
\begin{corollary}
  $\TW$ is fully faithful.
\end{corollary}
\begin{proof}
  By \cref{cor:faithful-via-terms}, it is enough to show that 
  for any theory $Ω$, the family morphism $\TmFamF{}{Ω} → \TmFamF{\TW}{Ω}$
  is an isomorphism. 
  We already showed that it is a monomorphism in \cref{cor:two-sortification-faithful},
  and \cref{prop:section-terms-two-sortification} provides a section.
  Therefore, it is an isomorphism and we get the desired result.
\end{proof}

\subsection{Proof of \cref{prop:initial-object-two-sortification}}
\label{app:proof-two-sortification-fully-faithful}
In this section, we prove a generalisation of
the key lemma stated in the proof of \cref{prop:initial-object-two-sortification}.

% \begin{notation}
%   Given a morphism $σ∶Ω → \Fam$, we denote 
%   the composition $Ω \xrightarrow{f}  \Fam → \XGAT$ by \AP$\intro*\UGAT{f} ∶ Ω → \XGAT$.
%   % where $! ∶ Ω → 1$ is the unique morphism to the terminal theory.
%   By \cref{not:mor-terms-set}, $\TmU{f}$ is the set of morphisms
%   $Ω → \YGAT$ over $\XGAT$.
%     Note that $\TmU{f}$ is bijective to 
%     $\hom_{\FinGat /\Fam}(f, (\TW\XGAT → \Fam))$
%     because $\TW \XGAT$ is the pullback of $\Fam → \XGAT ← \YGAT$.
%     We sometimes use this bijection implicitly.
%   %  $\TW Ω → \TW\XGAT$ over $\Fam$, and so we sometimes conflate the two sets.

% Given $A ∈ \TmU{Ω}$, we denote the composition 
% $Ω \xrightarrow{A} \TW\XGAT \xrightarrow{ε_{\XGAT}} \XGAT$
% by \AP$\intro*\ElMor{}{A}$, where $ε_{\XGAT}$ is the counit at $\XGAT$
% of the adjunction $\dom ∘ \pGAT^* ⊣ P$. 
% Therefore, $\TmEl{}{A}$ is 
% the set of morphisms $t ∶ \TW Ω → \YGAT$ such that $\pGAT ∘ t = \ElMor{}{A}$.
% Similarly, $\TmEl{}{A}$ is bijective to the set of morphisms
%   $Ω → \TW \YGAT$ over $\TW \XGAT$,
%  and we sometimes use this bijection implicitly.
% \end{notation}
\begin{notation}
  Given a morphism $f∶Ω → PX$, we denote 
  the composition $Ω \xrightarrow{f}  PX → \XGAT$ by \AP$\intro*\UGAT{f} ∶ Ω → \XGAT$.
  % where $! ∶ Ω → 1$ is the unique morphism to the terminal theory.
  % By \cref{not:mor-terms-set}, $\TmU{f}$ is the set of morphisms
  % $Ω → \YGAT$ over $\XGAT$.
  %   Note that $\TmU{f}$ is bijective to 
  %   $\hom_{\FinGat /\Fam}(f, (\TW\XGAT → \Fam))$
  %   because $\TW \XGAT$ is the pullback of $\Fam → \XGAT ← \YGAT$.
  %   We sometimes use this bijection implicitly.
  %  $\TW Ω → \TW\XGAT$ over $\Fam$, and so we sometimes conflate the two sets.

Given $f∶ Ω → PX$ and $A ∈ \TmU{f}$, we denote the composition 
$Ω → \TW\XGAT \xrightarrow{ε_{\XGAT}} \XGAT$
by \AP$\intro*\ElMor{f}{A}$, where 
the first morphism is the one corresponding to $A$ and ${f}$ by the universal property of the pullback defining $\TW \XGAT$,
and the second morphism 
$ε_{\XGAT}$ is the counit at $\XGAT$
of the adjunction $\dom ∘ \pGAT^* ⊣ P$. 
\end{notation}
\begin{remark}
  The notations are motivated by the fact that 
  a morphism $f ∶ Ω → \Fam$ corresponds to
  a pair of sorts $Ω ⊢ \U_f ∶ \GatSet$ and
  $Ω ⊢ \El_f ∶ \U_f → \GatSet$,
  an element of $\TmU{f}$ corresponds to
  a term $Ω ⊢ A ∶ \U_f$, and an element of $\TmEl{f}{A}$
  corresponds to a term $ Ω ⊢ t ∶ \El_f\,A$.
\end{remark}
From the following proposition, we easily recover
the key lemma stated in the proof of \cref{prop:initial-object-two-sortification},
by considering the initial model of $\TW Ω$ and $f$ to be the projection $\TW Ω → \Fam$.
\begin{proposition}
    \label{prop:initial-value-fam}
  Let $0_{Ω}$ denote the initial model of $Ω$ from \cref{prop:initial-model-gat-existence}.
  For any morphism $f ∶ Ω → \Fam$,
  the family $\model{f}
  (0_{Ω})$ is 
   $(\TmEl{f}{A})_{A ∈ \TmU{f}}$.
  % \model{\UGAT{Ω}}(0_{Ω})$
\end{proposition}

The proof relies on the following lemma.
\begin{lemma}
  \label{lem:nat-universal-pullback-P}
  Let $α $ be a natural transformation as induced by 
  \cref{thm:universal-property-natural-transformation}, i.e., any natural transformation
  between two functors
  $F,G ∶ \FinGat → \CC$, such that $G$ is a \kl{strict} $\Lexp$-functor,
  $F$ preserves pullbacks along $\pGAT$,
  and the naturality square of $α$ at $\pGAT$ is a pullback square in $\CC$.
  Then,
  for any theory $Γ$, the morphism $ F P Γ \xrightarrow{α_{P Γ}} G P Γ ≅ P_ℂ G Γ$
  % is the transpose of dashed composite morphism below, where 
  % $Γ'$ is the pullback of $P Γ → \XGAT$ along $\pGAT$, 
  % and $ε ∶ Γ' → Γ$ is the counit of the adjunction making $\pGAT$ exponentiable.
  % \[
  % % YADE DIAGRAM nat-universal-pullback-P.yade
  % % GENERATED LATEX
  % \input{diagrams/nat-universal-pullback-P.tex}
  % % END OF GENERATED LATEX
  % \]
  % More explicitly, $α_{P Γ}$ 
  is the composition below left,
  where $β_Γ ∶ FPΓ → P_ℂ F Γ$ is the transpose 
  of the dashed morphism below right with respect to 
  the adjunction $\dom ∘ p_ℂ^* ⊣ P_ℂ$, with $ε_Γ ∶ Γ' → Γ$
  denoting the counit component at $Γ$.
  \[
  % YADE DIAGRAM nat-universal-pullback-P-explicit.yade
  % GENERATED LATEX
  \input{diagrams/nat-universal-pullback-P-explicit.tex}
  % END OF GENERATED LATEX
  \]
\end{lemma}
  \begin{proof}
    This morphism and $α_{P Γ}$
    have the same transpose.
  \end{proof}
% \begin{proposition}
%   The initial model $0_Ω$ of $Ω$ satisfies that 
%   Let $σ ∶ Ω → \XGAT$ 
%   and $t ∶ Ω' → Γ$ be GAT morphisms, where $Ω'$ is the pullback of $σ$ along $\pGAT$,
%   inducing a morphism $σ' ∶ Ω → P Γ$.
%   Then, $\model{σ'}(0_Ω)$ is the model of $Γ$ satisfying
% \end{proposition}

\begin{proof}[Proof of \cref{prop:initial-value-fam}]
  We consider the natural transformation $α_Γ ∶ \hom(Ω,Γ)→ \model{Γ}$
  mapping $σ$ to $\model{σ}(0_Ω)$: the naturality square at $\pGAT$ is the pullback
  \eqref{eq:naturality-square-initial-model}.
  Note that $\model{f}(0_{Ω})$ is
  $α_{P\XGAT}(f)$.
  By
  \cref{lem:nat-universal-pullback-P}, 
  we conclude that
  %  $α_{\TW 1}∶ \hom(Ω, P \XGAT) → \FamCat{𝐒𝐞𝐭}$ 
  % is the transpose of
  % \[ 
  % \begin{array}{rcccl}
  % \hom(Ω, \TW \XGAT) 
  % & \xrightarrow{ε ∘ -} & \hom(Ω, \XGAT) & \xrightarrow{α_{\XGAT}} & \model{\XGAT} ≅ 𝐒𝐞𝐭.
  % \\
  % % A & \mapsto & \ElMor{A} & \mapsto & 
  % % \TmEl{A}
  % f & ↦ & ε ∘ f & ↦ &
  % \TmSet{ε ∘ f}
  % \end{array}
  % \]
  % Therefore,
   $α_{P\XGAT}$ is the following map 
  \begin{equation}
  \label{eq:compo-def-TmElA}
  % \begin{array}{ccccc}
  \hom(Ω, P\XGAT)
  \xrightarrow{β_{\XGAT}} 
  \FamCat{\hom(Ω, \XGAT)}
  \xrightarrow{\FamCat{α_{\XGAT}}} 
  \FamCat{𝐒𝐞𝐭}  
  % \\
  % f & \mapsto & 
  % \begin{tikzcd}
  % a \ar[d] \\ 𝐒𝐞𝐭  
  % \end{tikzcd}
  %  & \mapsto &
  % \end{array}
  \end{equation}
  Note that $α_{\XGAT}$ maps $A ∶ Ω → \XGAT$ to 
  $\model{A}0_Ω = \TmSet{A}$.
  It remains to compute $β_{\XGAT} ∶ \hom(Ω, P\XGAT)
  → \FamCat{\hom(Ω, \XGAT)}$.
  % Let us denote the canonical morphism $P \XGAT → \XGAT$ by $π$,
  % and the pullback of $π$ along $\pGAT$ by $π' ∶ \TW \XGAT → \YGAT$.
  As the transpose of $ε_{\XGAT} ∘- ∶\hom(Ω,\TW \XGAT) → \hom(Ω,\XGAT)$,
  the function $β_{\XGAT}$ maps $f ∶ Ω → P\XGAT$ 
  to the family 
  $(\ElMor{f}{A} )_{A ∈ \TmU{f}}$.
  It is now clear that the map \eqref{eq:compo-def-TmElA}
  yields the expected family %$(\TmEl{f}{A} )_{A ∈ \TmU{f}}$
  when applied to $f$.
  %  Taking $σ$  to be $f$, we recovered the claimed isomorphism.
  
  % maps $g ∶ Ω → P \XGAT$ to 

\end{proof}

  \section{Infinite GATs}
  \label{sec:infinite-gats}
  In this section, we sketch how to adapt our results
  to GATs that are not necessarily finite. 
  One could think of them as infinite contexts, following the type-theoretic 
  description of finite GATs, or more formally, as 
  Pro-objects~\cite[Section 2.2]{blom2023simplicial} of $\FinGat$, that is,
  as small cofiltered diagrams of finite GATs.
  % Such a GAT can always can be presented as a cofiltered diagram of finite GATs.
  For example, the infinite context $A₁ : \GatSet, A₂ : A_1 → \GatSet, ...$
  is modeled as a functor from the ordinal $ω^\op$, viewed as a category, to $\FinGat$,
  mapping $n$ to the finite theory $A_1 : \GatSet, A_2 : A_1 → \GatSet, ..., A_n : … → \GatSet$.

 The main result of this section is the following statement.
  \begin{theorem}
    \label{thm:infinite-GATs}
    The category \AP$\intro*\GatCat$ of GATs, equipped with $\pGAT$ viewed as a morphism in $\GatCat$, is bi-initial 
  in the 2-category which is 
  like $\Lexp$ except that objects are now categories with all limits and 
  the morphisms must preserve them as functors. Explicitly:
  \begin{itemize}
    \item an object is a category $ℂ$ with all limits 
    equipped with an \kl{exponentiable} morphism $p ∶ Y → X$ in $ℂ$;
    \item a morphism from $(ℂ, p)$ to $(ℂ', p')$ is a limit-preserving functor $F∶ ℂ \to ℂ'$ equipped with an isomorphism 
    $F(p) ≅ p'$, \kl{preserving pushforwards} along $p$;
    \item a 2-cell between two morphisms is a natural transformation between the underlying functors
     compatible with the isomorphisms to the target exponentiable morphism.
  \end{itemize}
  \end{theorem}
  The proofs and definitions of the previous sections can be easily adapted to this 
  bi-initiality property. 
  In particular, we still get a model functor and initial models, a fully faithful
  two-sortification functor from $\GatCat$ to the slice category $\GatCat/\FamG$,
  a core projection morphism whose image by the model functor 
  is the core projection functor 
  induced by the explicit description of the models of the reduced GAT.

The rest of this section is devoted to the proof of \cref{thm:infinite-GATs},
  mostly combining standard results about locally finitely presentable categories~\cite{Adamek}.
  The starting point is the following characterisation of 
  the category of GATs.
  \begin{proposition}
    \label{prop:carGATs}
    The embedding \AP$\intro*\iGAT ∶ \FinGat → \GatCat$ of finite GATs into GATs 
    makes $\GatCat$ the \AP\intro{free filtered completion} (or Pro-completion) of 
    $\FinGat$, in the sense that 
   $\GatCat$ has filtered limits and 
   for any category $𝔻$ with filtered limits, 
   the precomposition functor $- ∘ \iGAT ∶ [\GatCat, 𝔻]_{fl} → [\FinGat, 𝔻]$
  is an equivalence of categories, where $[\GatCat, 𝔻]_{fl}$ denotes the full subcategory of 
  $[\GatCat, 𝔻]$ consisting of functors preserving filtered limits.
  \end{proposition}
  This follows from making formal the above intuitive description of infinite GATs.
  In the same vein, Uemura~\cite{UEMURA2022106991} states that the category of GATs, with our convention 
  of the direction of morphisms\footnote{Recall that it is opposite 
  to Cartmell's original definition.}, is 
  equivalent to the opposite category of functors preserving finite limits from $\FinGat$ to $\Set$.
  The connection between these two characterisations of $\GatCat$ is provided by the following standard construction of \kl{free filtered completions}.
  \begin{proposition}
    \label{prop:def-lex}
    Let $ℂ$ be a small \kl{Lex-category}.
    Denoting the full subcategory 
    of functors $ℂ → 𝐒𝐞𝐭$ preserving finite limits by \AP$\intro*\Lex{ℂ}$, 
    the Yoneda embedding \AP$\intro*\yo ∶ ℂ → \Lex{ℂ}^\op$ 
    mapping $c$ to $\hom(c,-)$ is a \kl{free filtered completion}.
  \end{proposition}
  \begin{proof}
     By
   \cite[Proposition 1.2.4]{MakkaiPare} and 
  \cite[Proposition 5.39]{Kelly1982Basic}.
  \end{proof}
  % We deduce \cref{prop:carGATs} from the fact~\cite{UEMURA2022106991} 
  % that $\GatCat$ is equivalent to $\Lex{\FinGat}^\op$.
  \Cref{prop:carGATs} allows us to conclude that 
   $\GatCat$ is complete.
  \begin{proposition}
    \label{prop:free-filtered-preserves-lim}
    Given a small \kl{Lex-category} $ℂ$, 
    any \kl{free filtered completion} $i ∶ ℂ → \overline{ℂ}$
    preserves finite limits, and moreover, 
    $\overline{ℂ}$ is complete.
    Therefore, $\iGAT ∶ \FinGat → \GatCat$ preserves finite limits
    and $\GatCat$ is complete.
  \end{proposition}
  \begin{proof}
    It is easy to check that if this is true for one \kl{free filtered completion} of $ℂ$, then it is true for all such completions of $ℂ$. 
  This is true, in particular, for $ \yo ∶ ℂ → \Lex{ℂ}^{\op}$ from \cref{prop:def-lex}, by \cite[Proposition 1.45.(ii)]{Adamek}.
  \end{proof}
  
In the view to proving \cref{thm:infinite-GATs}, 
we need to lift the
$\Lexp$-structure of $\FinGat$ to $\GatCat$.
In particular, we need to show that $\iGAT(\pGAT)$ is exponentiable in $\GatCat$.
We notably rely on the following \kl{free filtered completion} of a slice category.
\begin{proposition}
  \label{prop:free-filtered-slice}
  Let $ℂ$ be a small \kl{Lex-category} and $j ∶ ℂ → \overline{ℂ}$ be 
  a \kl{free filtered completion}.
  Then, given any object $X$ of $ℂ$, the canonical functor $ℂ/X → \overline{ℂ}/jX$ mapping 
  $c → X$ to its image by $j$ is a \kl{free filtered completion}.
\end{proposition}
\begin{proof}
  Again, it is enough to show it for one \kl{free filtered completion} of $ℂ$.
  Thus, we consider 
  $ \yo ∶ ℂ → \Lex{ℂ}^{\op}$ from \cref{prop:def-lex}.
  By \cite[Proof of Proposition 2.6]{PittsAMreslfp},
  the slice category $\Lex{ℂ}^\op/\yo c$ is equivalent to $\Lex{ℂ/c}^\op$.
   Explicitly, the equivalence from $\Lex{ℂ/c}$ to $\yo c/\Lex{ℂ}$ maps 
   $F ∶ (ℂ/c) → 𝐒𝐞𝐭$ to the functor
   $F' ∶ ℂ → 𝐒𝐞𝐭$
   mapping $d$ to $F(d × c \xrightarrow{π₂} c)$, equipped with 
   the morphism $\yo c → F'$ given, by the Yoneda lemma, 
   by the element
    corresponding to the image of the single element 
    of the (terminal) set $ F(c → c)$
    by $F Δ ∶ F(c → c) → F'(c ×c \xrightarrow{π₂} c) $.
    From this description, it is clear 
    that $ℂ/c \xrightarrow{\yo} \Lex{ℂ/c}^\op ≃ \Lex{ℂ}^\op / \yo c$
    is, up to isomorphism, the claimed functor.
\end{proof}
To show exponentiability of $\iGAT(\pGAT)$ in $\GatCat$, we need 
to define a right adjoint to the pullback functor.
The following proposition gives a general construction of right adjoints,
 exploiting the universal property of \kl{free filtered completions}.
\begin{proposition}
  \label{prop:extension-adjoints-filtered-limits}
  Let $i ∶ ℂ → \overline{ℂ}$ and 
  $j ∶ 𝔻 → \overline{𝔻}$ be  
  \kl{free filtered completions} of small \kl{Lex-categories},
  with a left adjoint $L ∶ ℂ → 𝔻$ and 
  a functor  $\overline{L} ∶ \overline{ℂ} → \overline{𝔻}$
  preserving filtered limits, related by an isomorphism 
  $j ∘ L ≅ \overline{L} ∘ i$.
  Then, $\overline{L}$ has a right adjoint
  $\overline{R} ∶ \overline{𝔻} → \overline{ℂ}$
  such that the mate 
  $i ∘ R → \overline{R} ∘ j$
  of the above isomorphism is an isomorphism as well,
  where $R ∶ 𝔻 → ℂ$ denotes the right adjoint of $L$.
\end{proposition}

\begin{corollary}
  \label{cor:exponentiable-free-filtered}
  Given a \kl{free filtered completion} $i ∶ ℂ → \overline{ℂ}$ of a 
  \kl{Lex-category} $ℂ$ and an exponentiable morphism
  $p ∶ Y → X$ in $ℂ$, the morphism
  $i p ∶ i Y → i X$ is exponentiable in $\overline{ℂ}$
  and $i$ \kl{preserves pushforwards} along $p$.

\end{corollary}
\begin{proof}
We apply \cref{prop:extension-adjoints-filtered-limits}
by taking $i$ and $j$ to be the embeddings $ℂ/X → \overline{ℂ}/i X$
and $ℂ/Y → \overline{ℂ}/i Y$
by \cref{prop:free-filtered-slice},
 with
  $L$ and $\overline{L}$ being the functors pulling back
  along $p$ and $i p$ respectively.
 The isomorphism $j ∘ L ≅ \overline{L} ∘ i$ is 
 given by the fact that $i$ preserves pullbacks,
 by \cref{prop:free-filtered-preserves-lim}.
 Moreover, $\overline{L}$ preserves filtered limits 
 since it is right adjoint to 
 the functor $\overline{ℂ}/i Y → \overline{ℂ}/i X$ 
 mapping a morphism $c → iY$ 
 to $ c → iY \xrightarrow{ip} iX$. 
 We then get a right adjoint $\overline{R}$ 
 which ensures that $i p$ is exponentiable in $\overline{ℂ}$.
 Moreover, the isomorphism $i ∘ R ≅ \overline{R} ∘ j$,
  where $R$ is the pushforward along $p$,
  ensures that $i$ \kl{preserves pushforwards} along $p$.
\end{proof}
\begin{corollary}
  \label{cor:iGAT-lexp}
  $\GatCat$, equipped with 
  $\iGAT(\pGAT)$, 
   is a $\Lexp$-category, and
  $\iGAT ∶ \FinGat → \GatCat$ is a $\Lexp$-functor.  
\end{corollary}
\begin{proof}
  By \cref{prop:free-filtered-preserves-lim}, $\GatCat$ has finite limits, preserved by 
  $\iGAT$.
  By \cref{cor:exponentiable-free-filtered},
  $\iGAT (\pGAT)$ is exponentiable in $\GatCat$
  and $\iGAT$ \kl{preserves pushforwards} along $\pGAT$.
\end{proof}
The proof of \cref{thm:infinite-GATs} finally relies on 
the following observation.
\begin{proposition}
  \label{prop:lift-lexp-functor}
  Let $i ∶ ℂ → \overline{ℂ}$ be a \kl{free filtered completion}
  of a small \kl{Lex-category} $ℂ$
  with an exponentiable morphism $p_ℂ ∶ Y_ℂ → X_ℂ$.
  Consider a $\Lexp$-functor 
  $F ∶ ℂ → 𝔻$ with $𝔻$ complete.
  By universal property of the \kl{free filtered completion},
  we get a functor $\overline{F} ∶ \overline{ℂ} → 𝔻$
  with an isomorphism $\overline{F} ∘ i ≅ F$.

   Then, $\overline{F}$ preserves limits and pushforwards 
   along $i p_ℂ$.
\end{proposition}
\begin{proof}
  The fact that it preserves limits is a consequence 
  of \cite[Proposition 1.45.(ii)]{Adamek}.
  For preservation of pushforwards, we need to show that 
  the mate of the following natural isomorphism is also an isomorphism.
  \[
% YADE DIAGRAM mate-free-filtered.yade
% GENERATED LATEX
\input{diagrams/mate-free-filtered.tex}
% END OF GENERATED LATEX
\]
  By the universal property 
  of the \kl{free filtered completion} 
  $ℂ/Y_ℂ → \overline{ℂ}/iY_ℂ$ by \cref{prop:free-filtered-slice}, it is enough to show 
  that the whiskering of the mate by 
  $ℂ/Y_ℂ → \overline{ℂ}/iY_ℂ$ is an isomorphism.
  But this follows from the fact that $F$ itself 
  \kl{preserves pushforwards} along $p_ℂ$.
\end{proof}
We are now ready to prove the claimed bi-initial 
characterisation of $\GatCat$.

\begin{proof}[Proof of \cref{thm:infinite-GATs}]
  First, $(\GatCat, \iGAT(\pGAT))$ is a $\Lexp$-category by 
  \cref{cor:iGAT-lexp}. Moreover, it is complete by 
  \cref{prop:free-filtered-preserves-lim}.
Consider a $\Lexp$-category $ℂ$ with all limits.
By bi-initiality of $\FinGat$, we get 
a $\Lexp$-functor $F ∶ \FinGat → ℂ$, which 
extends to a $\Lexp$-functor $\overline{F} ∶ \GatCat → ℂ$
 preserving limits by \cref{prop:lift-lexp-functor}.

 Now, assume given another $\Lexp$-functor $G ∶ \GatCat → ℂ$
 preserving limits,
 and let us show that $G$ is isomorphic to $\overline{F}$, for a unique isomorphism.
 By the universal property of the \kl{free filtered completion}
 $\iGAT ∶ \FinGat → \GatCat$, it is enough to show 
 that $G ∘ \iGAT$ is isomorphic to $\overline{F} ∘ \iGAT$ for a unique morphism, which follows 
 from bi-initiality of $\FinGat$, since both are composition 
 of $\Lexp$-functors, by \cref{cor:iGAT-lexp}.
  
\end{proof}

\fi

%% file: diagrams/p2-pushforward.tex
\begin{tikzpicture}
\draw[white,-,line width=0.20833333333333334em] (9.114583333333334em,-2.7083333333333335em) .. controls (9.114583333333334em,-3.3333333333333335em) and (9.114583333333334em,-3.9583333333333335em) .. (9.114583333333334em,-4.583333333333334em);
\draw[black,->, ] (9.114583333333334em,-2.7083333333333335em) .. controls (9.114583333333334em,-3.3333333333333335em) and (9.114583333333334em,-3.9583333333333335em) .. (9.114583333333334em,-4.583333333333334em);
\draw[white,-,line width=0.20833333333333334em] (6.197916666666667em,-5.46875em) .. controls (6.71875em,-5.46875em) and (7.239583333333334em,-5.46875em) .. (7.760416666666667em,-5.46875em);
\draw[black,->, dashed, ] (6.197916666666667em,-5.46875em) .. controls (6.71875em,-5.46875em) and (7.239583333333334em,-5.46875em) .. (7.760416666666667em,-5.46875em);
\draw[white,-,line width=0.20833333333333334em] (6.510416666666667em,-1.8229166666666667em) .. controls (7.057291666666667em,-1.8229166666666667em) and (7.604166666666667em,-1.8229166666666667em) .. (8.151041666666668em,-1.8229166666666667em);
\draw[black,->, ] (6.510416666666667em,-1.8229166666666667em) .. controls (7.057291666666667em,-1.8229166666666667em) and (7.604166666666667em,-1.8229166666666667em) .. (8.151041666666668em,-1.8229166666666667em);
\draw[white,-,line width=0.20833333333333334em] (5.46875em,-2.7083333333333335em) .. controls (5.46875em,-3.3420138888888884em) and (5.46875em,-3.975694444444444em) .. (5.46875em,-4.609375em);
\draw[black,->, ] (5.46875em,-2.7083333333333335em) .. controls (5.46875em,-3.3420138888888884em) and (5.46875em,-3.975694444444444em) .. (5.46875em,-4.609375em);
\draw[white,-,line width=0.20833333333333334em] (5.46875em,-6.328125em) .. controls (5.46875em,-6.961805555555555em) and (5.46875em,-7.595486111111111em) .. (5.46875em,-8.229166666666668em);
\draw[black,->, ] (5.46875em,-6.328125em) .. controls (5.46875em,-6.961805555555555em) and (5.46875em,-7.595486111111111em) .. (5.46875em,-8.229166666666668em);
\draw[white,-,curve={ratio=0.3}, line width=0.20833333333333334em] (4.937500000000001em,-2.7083333333333335em) .. controls (3.833333333333334em,-4.548611111111111em) and (3.833333333333334em,-6.388888888888888em) .. (4.937500000000001em,-8.229166666666668em);
\draw[black,->, curve={ratio=0.3}, ] (4.937500000000001em,-2.7083333333333335em) .. controls (3.833333333333334em,-4.548611111111111em) and (3.833333333333334em,-6.388888888888888em) .. (4.937500000000001em,-8.229166666666668em);
\draw[white,-,line width=0.20833333333333334em] (12.760416666666666em,-2.682291666666667em) .. controls (12.760416666666666em,-3.3246527777777777em) and (12.760416666666666em,-3.9670138888888884em) .. (12.760416666666666em,-4.609375em);
\draw[black,->, ] (12.760416666666666em,-2.682291666666667em) .. controls (12.760416666666666em,-3.3246527777777777em) and (12.760416666666666em,-3.9670138888888884em) .. (12.760416666666666em,-4.609375em);
\draw[white,-,line width=0.20833333333333334em] (12.760416666666666em,-6.328125em) .. controls (12.760416666666666em,-6.961805555555555em) and (12.760416666666666em,-7.595486111111111em) .. (12.760416666666666em,-8.229166666666668em);
\draw[black,->, ] (12.760416666666666em,-6.328125em) .. controls (12.760416666666666em,-6.961805555555555em) and (12.760416666666666em,-7.595486111111111em) .. (12.760416666666666em,-8.229166666666668em);
\draw[white] (5.729166666666667em,-3.229166666666667em) -- (7.03125em,-3.229166666666667em);
\draw[white] (7.03125em,-3.229166666666667em) -- (7.03125em,-2.0833333333333335em);
\draw[black] (5.729166666666667em,-3.229166666666667em) -- (7.03125em,-3.229166666666667em);
\draw[black] (7.03125em,-3.229166666666667em) -- (7.03125em,-2.0833333333333335em);
\node at (9.114583333333334em,-1.8229166666666667em) {$PZ$} ;
\node at (9.114583333333334em,-5.46875em) {$PPX$} ;
\node at (5.46875em,-5.46875em) {$\XX$} ;
\node at (5.46875em,-1.8229166666666667em) {$\Pu Z$} ;
\node at (5.46875em,-9.114583333333334em) {$PX$} ;
\node at (12.760416666666666em,-1.8229166666666667em) {$\YY$} ;
\node at (12.760416666666666em,-5.46875em) {$\XX$} ;
\node at (12.760416666666666em,-9.114583333333334em) {$PX$} ;
\node[scale=0.7] at (9.671379309805483em,-3.6458333333333335em) {$Pz$} ;
\node[scale=0.7] at (13.036281260857352em,-3.6458333333333335em) {$\pp$} ;
\end{tikzpicture}

%% file: diagrams/stability-family-GAT-P-pbk.tex
\begin{tikzpicture}
\draw[white,-,line width=0.20833333333333334em] (12.786458333333334em,-7.03125em) .. controls (13.428819444444445em,-7.03125em) and (14.071180555555554em,-7.03125em) .. (14.713541666666668em,-7.03125em);
\draw[black,->, ] (12.786458333333334em,-7.03125em) .. controls (13.428819444444445em,-7.03125em) and (14.071180555555554em,-7.03125em) .. (14.713541666666668em,-7.03125em);
\draw[white,-,line width=0.20833333333333334em] (11.71875em,-7.916666666666667em) .. controls (11.71875em,-8.88888888888889em) and (11.71875em,-9.86111111111111em) .. (11.71875em,-10.833333333333334em);
\draw[black,->, ] (11.71875em,-7.916666666666667em) .. controls (11.71875em,-8.88888888888889em) and (11.71875em,-9.86111111111111em) .. (11.71875em,-10.833333333333334em);
\draw[white,-,line width=0.20833333333333334em] (16.40625em,-7.916666666666667em) .. controls (16.40625em,-8.88888888888889em) and (16.40625em,-9.86111111111111em) .. (16.40625em,-10.833333333333334em);
\draw[black,->, ] (16.40625em,-7.916666666666667em) .. controls (16.40625em,-8.88888888888889em) and (16.40625em,-9.86111111111111em) .. (16.40625em,-10.833333333333334em);
\draw[white,-,line width=0.20833333333333334em] (12.682291666666668em,-11.71875em) .. controls (13.342013888888888em,-11.71875em) and (14.00173611111111em,-11.71875em) .. (14.661458333333334em,-11.71875em);
\draw[black,->, ] (12.682291666666668em,-11.71875em) .. controls (13.342013888888888em,-11.71875em) and (14.00173611111111em,-11.71875em) .. (14.661458333333334em,-11.71875em);
\draw[white] (11.979166666666668em,-8.4375em) -- (13.307291666666668em,-8.4375em);
\draw[white] (13.307291666666668em,-8.4375em) -- (13.307291666666668em,-7.291666666666667em);
\draw[black] (11.979166666666668em,-8.4375em) -- (13.307291666666668em,-8.4375em);
\draw[black] (13.307291666666668em,-8.4375em) -- (13.307291666666668em,-7.291666666666667em);
\draw[white,-,line width=0.20833333333333334em] (11.71875em,-12.604166666666668em) .. controls (11.71875em,-13.576388888888888em) and (11.71875em,-14.54861111111111em) .. (11.71875em,-15.520833333333334em);
\draw[black,->, ] (11.71875em,-12.604166666666668em) .. controls (11.71875em,-13.576388888888888em) and (11.71875em,-14.54861111111111em) .. (11.71875em,-15.520833333333334em);
\draw[white,-,line width=0.20833333333333334em] (7.760416666666667em,-16.40625em) .. controls (8.628472222222221em,-16.40625em) and (9.496527777777777em,-16.40625em) .. (10.364583333333334em,-16.40625em);
\draw[black,->, ] (7.760416666666667em,-16.40625em) .. controls (8.628472222222221em,-16.40625em) and (9.496527777777777em,-16.40625em) .. (10.364583333333334em,-16.40625em);
\draw[white,-,line width=0.20833333333333334em] (7.994791666666667em,-11.71875em) .. controls (8.914930555555555em,-11.71875em) and (9.835069444444445em,-11.71875em) .. (10.755208333333334em,-11.71875em);
\draw[black,->, ] (7.994791666666667em,-11.71875em) .. controls (8.914930555555555em,-11.71875em) and (9.835069444444445em,-11.71875em) .. (10.755208333333334em,-11.71875em);
\draw[white,-,line width=0.20833333333333334em] (7.03125em,-12.578125em) .. controls (7.03125em,-13.567708333333334em) and (7.03125em,-14.557291666666668em) .. (7.03125em,-15.546875em);
\draw[black,->, ] (7.03125em,-12.578125em) .. controls (7.03125em,-13.567708333333334em) and (7.03125em,-14.557291666666668em) .. (7.03125em,-15.546875em);
\draw[white,-,line width=0.20833333333333334em] (8.098958333333334em,-7.03125em) .. controls (8.949652777777777em,-7.03125em) and (9.800347222222221em,-7.03125em) .. (10.651041666666668em,-7.03125em);
\draw[black,->, ] (8.098958333333334em,-7.03125em) .. controls (8.949652777777777em,-7.03125em) and (9.800347222222221em,-7.03125em) .. (10.651041666666668em,-7.03125em);
\draw[white,-,line width=0.20833333333333334em] (7.03125em,-7.890625em) .. controls (7.03125em,-8.880208333333334em) and (7.03125em,-9.869791666666668em) .. (7.03125em,-10.859375em);
\draw[black,->, ] (7.03125em,-7.890625em) .. controls (7.03125em,-8.880208333333334em) and (7.03125em,-9.869791666666668em) .. (7.03125em,-10.859375em);
\draw[white] (7.291666666666667em,-8.411458333333334em) -- (8.619791666666668em,-8.411458333333334em);
\draw[white] (8.619791666666668em,-8.411458333333334em) -- (8.619791666666668em,-7.291666666666667em);
\draw[black] (7.291666666666667em,-8.411458333333334em) -- (8.619791666666668em,-8.411458333333334em);
\draw[black] (8.619791666666668em,-8.411458333333334em) -- (8.619791666666668em,-7.291666666666667em);
\draw[white,-,curve={ratio=-0.3}, line width=0.20833333333333334em] (8.098958333333334em,-6.390625em) .. controls (10.327199074074075em,-4.942708333333332em) and (12.60439814814815em,-4.861111111111111em) .. (14.930555555555557em,-6.145833333333334em);
\draw[black,->, dashed, curve={ratio=-0.3}, ] (8.098958333333334em,-6.390625em) .. controls (10.327199074074075em,-4.942708333333332em) and (12.60439814814815em,-4.861111111111111em) .. (14.930555555555557em,-6.145833333333334em);
\draw[white] (7.291666666666667em,-13.098958333333334em) -- (8.515625em,-13.098958333333334em);
\draw[white] (8.515625em,-13.098958333333334em) -- (8.515625em,-11.979166666666668em);
\draw[black] (7.291666666666667em,-13.098958333333334em) -- (8.515625em,-13.098958333333334em);
\draw[black] (8.515625em,-13.098958333333334em) -- (8.515625em,-11.979166666666668em);
\draw[white,-,curve={ratio=0.2}, line width=0.20833333333333334em] (7.994791666666667em,-12.104166666666668em) .. controls (10.175347222222221em,-13.097222222222223em) and (12.397569444444443em,-13.20138888888889em) .. (14.661458333333334em,-12.416666666666668em);
\draw[black,->, dashed, curve={ratio=0.2}, ] (7.994791666666667em,-12.104166666666668em) .. controls (10.175347222222221em,-13.097222222222223em) and (12.397569444444443em,-13.20138888888889em) .. (14.661458333333334em,-12.416666666666668em);
\node at (7.03125em,-11.71875em) {$\Pu Γ$} ;
\node at (11.71875em,-7.03125em) {$P\Delta $} ;
\node at (16.40625em,-7.03125em) {$PPⁿ Y$} ;
\node at (11.71875em,-11.71875em) {$PΓ$} ;
\node at (16.40625em,-11.71875em) {$PPⁿX$} ;
\node at (11.71875em,-16.40625em) {$PPX$} ;
\node at (7.03125em,-16.40625em) {$\XX$} ;
\node at (7.03125em,-7.03125em) {$\Pu \Delta $} ;
\node[scale=0.7] at (17.390760546348467em,-9.375em) {$PP^n p$} ;
\end{tikzpicture}

%% file: diagrams/stability-family-GAT-P-eq.tex
\begin{tikzpicture}
\draw[white,-,line width=0.20833333333333334em] (16.40625em,-7.916666666666667em) .. controls (16.40625em,-8.88888888888889em) and (16.40625em,-9.86111111111111em) .. (16.40625em,-10.833333333333334em);
\draw[black,->, ] (16.40625em,-7.916666666666667em) .. controls (16.40625em,-8.88888888888889em) and (16.40625em,-9.86111111111111em) .. (16.40625em,-10.833333333333334em);
\draw[white,-,curve={ratio=-0.1}, line width=0.20833333333333334em] (17.369791666666668em,-11.526041666666668em) .. controls (19.599652777777777em,-11.02951388888889em) and (21.83923611111111em,-10.980902777777777em) .. (24.088541666666668em,-11.380208333333334em);
\draw[black,->, curve={ratio=-0.1}, ] (17.369791666666668em,-11.526041666666668em) .. controls (19.599652777777777em,-11.02951388888889em) and (21.83923611111111em,-10.980902777777777em) .. (24.088541666666668em,-11.380208333333334em);
\draw[white,-,curve={ratio=0.1}, line width=0.20833333333333334em] (17.369791666666668em,-11.911458333333334em) .. controls (19.599652777777777em,-12.40798611111111em) and (21.83923611111111em,-12.456597222222221em) .. (24.088541666666668em,-12.057291666666668em);
\draw[black,->, curve={ratio=0.1}, ] (17.369791666666668em,-11.911458333333334em) .. controls (19.599652777777777em,-12.40798611111111em) and (21.83923611111111em,-12.456597222222221em) .. (24.088541666666668em,-12.057291666666668em);
\draw[white,-,line width=0.20833333333333334em] (16.40625em,-12.604166666666668em) .. controls (16.40625em,-13.576388888888888em) and (16.40625em,-14.54861111111111em) .. (16.40625em,-15.520833333333334em);
\draw[black,->, ] (16.40625em,-12.604166666666668em) .. controls (16.40625em,-13.576388888888888em) and (16.40625em,-14.54861111111111em) .. (16.40625em,-15.520833333333334em);
\draw[white,-,line width=0.20833333333333334em] (12.447916666666668em,-16.40625em) .. controls (13.315972222222221em,-16.40625em) and (14.184027777777777em,-16.40625em) .. (15.052083333333334em,-16.40625em);
\draw[black,->, ] (12.447916666666668em,-16.40625em) .. controls (13.315972222222221em,-16.40625em) and (14.184027777777777em,-16.40625em) .. (15.052083333333334em,-16.40625em);
\draw[white,-,line width=0.20833333333333334em] (12.682291666666668em,-11.71875em) .. controls (13.602430555555554em,-11.71875em) and (14.522569444444445em,-11.71875em) .. (15.442708333333334em,-11.71875em);
\draw[black,->, ] (12.682291666666668em,-11.71875em) .. controls (13.602430555555554em,-11.71875em) and (14.522569444444445em,-11.71875em) .. (15.442708333333334em,-11.71875em);
\draw[white,-,line width=0.20833333333333334em] (11.71875em,-12.578125em) .. controls (11.71875em,-13.567708333333334em) and (11.71875em,-14.557291666666668em) .. (11.71875em,-15.546875em);
\draw[black,->, ] (11.71875em,-12.578125em) .. controls (11.71875em,-13.567708333333334em) and (11.71875em,-14.557291666666668em) .. (11.71875em,-15.546875em);
\draw[white] (11.979166666666668em,-13.098958333333334em) -- (13.203125em,-13.098958333333334em);
\draw[white] (13.203125em,-13.098958333333334em) -- (13.203125em,-11.979166666666668em);
\draw[black] (11.979166666666668em,-13.098958333333334em) -- (13.203125em,-13.098958333333334em);
\draw[black] (13.203125em,-13.098958333333334em) -- (13.203125em,-11.979166666666668em);
\draw[white,-,line width=0.20833333333333334em] (12.786458333333334em,-7.03125em) .. controls (13.637152777777777em,-7.03125em) and (14.487847222222221em,-7.03125em) .. (15.338541666666668em,-7.03125em);
\draw[black,->, ] (12.786458333333334em,-7.03125em) .. controls (13.637152777777777em,-7.03125em) and (14.487847222222221em,-7.03125em) .. (15.338541666666668em,-7.03125em);
\draw[white,-,line width=0.20833333333333334em] (11.71875em,-7.890625em) .. controls (11.71875em,-8.880208333333334em) and (11.71875em,-9.869791666666668em) .. (11.71875em,-10.859375em);
\draw[black,->, ] (11.71875em,-7.890625em) .. controls (11.71875em,-8.880208333333334em) and (11.71875em,-9.869791666666668em) .. (11.71875em,-10.859375em);
\draw[white] (11.979166666666668em,-8.411458333333334em) -- (13.307291666666668em,-8.411458333333334em);
\draw[white] (13.307291666666668em,-8.411458333333334em) -- (13.307291666666668em,-7.291666666666667em);
\draw[black] (11.979166666666668em,-8.411458333333334em) -- (13.307291666666668em,-8.411458333333334em);
\draw[black] (13.307291666666668em,-8.411458333333334em) -- (13.307291666666668em,-7.291666666666667em);
\draw[white,-,curve={ratio=-0.3}, line width=0.20833333333333334em] (12.682291666666668em,-11.140625em) .. controls (16.495254629629628em,-8.713541666666668em) and (20.369675925925925em,-8.61111111111111em) .. (24.305555555555554em,-10.833333333333334em);
\draw[black,->, dashed, curve={ratio=-0.3}, ] (12.682291666666668em,-11.140625em) .. controls (16.495254629629628em,-8.713541666666668em) and (20.369675925925925em,-8.61111111111111em) .. (24.305555555555554em,-10.833333333333334em);
\draw[white,-,curve={ratio=0.3}, line width=0.20833333333333334em] (12.682291666666668em,-12.296875em) .. controls (16.495254629629628em,-14.723958333333334em) and (20.369675925925925em,-14.826388888888888em) .. (24.305555555555554em,-12.604166666666668em);
\draw[black,->, dashed, curve={ratio=0.3}, ] (12.682291666666668em,-12.296875em) .. controls (16.495254629629628em,-14.723958333333334em) and (20.369675925925925em,-14.826388888888888em) .. (24.305555555555554em,-12.604166666666668em);
\node at (11.71875em,-11.71875em) {$\Pu Γ$} ;
\node at (11.71875em,-7.03125em) {$\Pu Δ$} ;
\node at (16.40625em,-7.03125em) {$PΔ$} ;
\node at (16.40625em,-11.71875em) {$PΓ$} ;
\node at (25.78125em,-11.71875em) {$PPⁿY$} ;
\node at (16.40625em,-16.40625em) {$PPX$} ;
\node at (11.71875em,-16.40625em) {$\XX$} ;
\node[scale=0.7] at (20.710189628859602em,-10.578825758174489em) {$P σ$} ;
\node[scale=0.7] at (20.710553001473976em,-12.841933146377718em) {$P δ$} ;
\end{tikzpicture}

%% file: diagrams/naturality-square-initial-model.tex
\begin{tikzpicture}
\draw[white,-,line width=0.20833333333333334em] (11.71875em,-7.916666666666667em) .. controls (11.71875em,-8.44791677263048em) and (11.71875em,-8.979166878594292em) .. (11.71875em,-9.510416984558105em);
\draw[black,->, ] (11.71875em,-7.916666666666667em) .. controls (11.71875em,-8.44791677263048em) and (11.71875em,-8.979166878594292em) .. (11.71875em,-9.510416984558105em);
\draw[white,-,line width=0.20833333333333334em] (14.453125em,-7.03125em) .. controls (15.182291666666668em,-7.03125em) and (15.911458333333334em,-7.03125em) .. (16.640625em,-7.03125em);
\draw[black,->, ] (14.453125em,-7.03125em) .. controls (15.182291666666668em,-7.03125em) and (15.911458333333334em,-7.03125em) .. (16.640625em,-7.03125em);
\draw[white,-,line width=0.20833333333333334em] (18.645833333333336em,-7.916666666666667em) .. controls (18.645833333333336em,-8.44791677263048em) and (18.645833333333336em,-8.979166878594292em) .. (18.645833333333336em,-9.510416984558105em);
\draw[black,->, ] (18.645833333333336em,-7.916666666666667em) .. controls (18.645833333333336em,-8.44791677263048em) and (18.645833333333336em,-8.979166878594292em) .. (18.645833333333336em,-9.510416984558105em);
\draw[white,-,line width=0.20833333333333334em] (14.583333333333334em,-10.395833651224773em) .. controls (15.581597222222221em,-10.395833651224772em) and (16.57986111111111em,-10.395833651224772em) .. (17.578125em,-10.395833651224773em);
\draw[black,->, ] (14.583333333333334em,-10.395833651224773em) .. controls (15.581597222222221em,-10.395833651224772em) and (16.57986111111111em,-10.395833651224772em) .. (17.578125em,-10.395833651224773em);
\draw[white] (11.979166666666668em,-8.4375em) -- (14.973958333333334em,-8.4375em);
\draw[white] (14.973958333333334em,-8.4375em) -- (14.973958333333334em,-7.291666666666667em);
\draw[black] (11.979166666666668em,-8.4375em) -- (14.973958333333334em,-8.4375em);
\draw[black] (14.973958333333334em,-8.4375em) -- (14.973958333333334em,-7.291666666666667em);
\node at (11.71875em,-7.03125em) {$\hom(Ω,\YGAT)$} ;
\node at (11.71875em,-10.395833651224773em) {$\hom(Ω,\XGAT)$} ;
\node at (18.645833333333336em,-7.03125em) {$\PtdSet$} ;
\node at (18.645833333333336em,-10.395833651224773em) {$𝐒𝐞𝐭$} ;
\end{tikzpicture}

%% file: diagrams/famcat-arrow-functor.tex
\begin{tikzpicture}
\draw[white,-,line width=0.20833333333333334em] (10.416666666666668em,-2.96875em) .. controls (10.416666666666666em,-3.7586805555555554em) and (10.416666666666666em,-4.548611111111111em) .. (10.416666666666668em,-5.338541666666667em);
\draw[black,->, ] (10.416666666666668em,-2.96875em) .. controls (10.416666666666666em,-3.7586805555555554em) and (10.416666666666666em,-4.548611111111111em) .. (10.416666666666668em,-5.338541666666667em);
\draw[white,-,line width=0.20833333333333334em] (11.40625em,-2.0833333333333335em) .. controls (12.504822530864198em,-2.0833333333333335em) and (13.603395061728394em,-2.0833333333333335em) .. (14.701967592592593em,-2.0833333333333335em);
\draw[black,->, ] (11.40625em,-2.0833333333333335em) .. controls (12.504822530864198em,-2.0833333333333335em) and (13.603395061728394em,-2.0833333333333335em) .. (14.701967592592593em,-2.0833333333333335em);
\draw[white,-,line width=0.20833333333333334em] (16.75925925925926em,-2.96875em) .. controls (16.75925925925926em,-3.767361111111111em) and (16.75925925925926em,-4.565972222222222em) .. (16.75925925925926em,-5.364583333333334em);
\draw[black,->, ] (16.75925925925926em,-2.96875em) .. controls (16.75925925925926em,-3.767361111111111em) and (16.75925925925926em,-4.565972222222222em) .. (16.75925925925926em,-5.364583333333334em);
\draw[white,-,line width=0.20833333333333334em] (13.567708333333334em,-6.25em) .. controls (13.937114197530866em,-6.25em) and (14.306520061728397em,-6.25em) .. (14.675925925925926em,-6.25em);
\draw[black,->, ] (13.567708333333334em,-6.25em) .. controls (13.937114197530866em,-6.25em) and (14.306520061728397em,-6.25em) .. (14.675925925925926em,-6.25em);
\draw[white,-,line width=0.20833333333333334em] (10.416666666666668em,-7.161458333333334em) .. controls (10.416666666666666em,-7.951388888888888em) and (10.416666666666666em,-8.741319444444445em) .. (10.416666666666668em,-9.53125em);
\draw[black,->, ] (10.416666666666668em,-7.161458333333334em) .. controls (10.416666666666666em,-7.951388888888888em) and (10.416666666666666em,-8.741319444444445em) .. (10.416666666666668em,-9.53125em);
\draw[white,-,curve={ratio=0.7}, line width=0.20833333333333334em] (9.427083333333334em,-2.7901785714285716em) .. controls (6.197916666666666em,-4.920535714285713em) and (6.114583333333332em,-7.167559523809523em) .. (9.177083333333334em,-9.53125em);
\draw[blue,->, curve={ratio=0.7}, ] (9.427083333333334em,-2.7901785714285716em) .. controls (6.197916666666666em,-4.920535714285713em) and (6.114583333333332em,-7.167559523809523em) .. (9.177083333333334em,-9.53125em);
\draw[white,-,line width=0.20833333333333334em] (21.510416666666668em,-2.0833333333333335em) .. controls (21.901041666666668em,-2.0833333333333335em) and (22.291666666666668em,-2.0833333333333335em) .. (22.682291666666668em,-2.0833333333333335em);
\draw[black,->, ] (21.510416666666668em,-2.0833333333333335em) .. controls (21.901041666666668em,-2.0833333333333335em) and (22.291666666666668em,-2.0833333333333335em) .. (22.682291666666668em,-2.0833333333333335em);
\draw[white,-,line width=0.20833333333333334em] (28.984375em,-2.0833333333333335em) .. controls (29.35763888888889em,-2.0833333333333335em) and (29.730902777777775em,-2.0833333333333335em) .. (30.104166666666668em,-2.0833333333333335em);
\draw[black,->, ] (28.984375em,-2.0833333333333335em) .. controls (29.35763888888889em,-2.0833333333333335em) and (29.730902777777775em,-2.0833333333333335em) .. (30.104166666666668em,-2.0833333333333335em);
\draw[white,-,line width=0.20833333333333334em] (31.5625em,-2.96875em) .. controls (31.5625em,-5.15625em) and (31.5625em,-7.34375em) .. (31.5625em,-9.53125em);
\draw[black,->, ] (31.5625em,-2.96875em) .. controls (31.5625em,-5.15625em) and (31.5625em,-7.34375em) .. (31.5625em,-9.53125em);
\draw[white,-,line width=0.20833333333333334em] (22.604166666666668em,-10.416666666666668em) .. controls (25.234375em,-10.416666666666666em) and (27.864583333333336em,-10.416666666666666em) .. (30.494791666666668em,-10.416666666666668em);
\draw[black,->, ] (22.604166666666668em,-10.416666666666668em) .. controls (25.234375em,-10.416666666666666em) and (27.864583333333336em,-10.416666666666666em) .. (30.494791666666668em,-10.416666666666668em);
\draw[white,-,line width=0.20833333333333334em] (20.520833333333336em,-2.96875em) .. controls (20.520833333333332em,-5.15625em) and (20.520833333333332em,-7.34375em) .. (20.520833333333336em,-9.53125em);
\draw[blue,->, ] (20.520833333333336em,-2.96875em) .. controls (20.520833333333332em,-5.15625em) and (20.520833333333332em,-7.34375em) .. (20.520833333333336em,-9.53125em);
\draw[white,-,line width=0.20833333333333334em] (25.252278645833336em,-2.994791666666667em) .. controls (23.86328125em,-5.173611111111111em) and (22.474283854166664em,-7.352430555555555em) .. (21.085286458333332em,-9.53125em);
\draw[black,->, ] (25.252278645833336em,-2.994791666666667em) .. controls (23.86328125em,-5.173611111111111em) and (22.474283854166664em,-7.352430555555555em) .. (21.085286458333332em,-9.53125em);
\draw[white] (10.677083333333334em,-3.4895833333333335em) -- (11.927083333333334em,-3.4895833333333335em);
\draw[white] (11.927083333333334em,-3.4895833333333335em) -- (11.927083333333334em,-2.34375em);
\draw[black] (10.677083333333334em,-3.4895833333333335em) -- (11.927083333333334em,-3.4895833333333335em);
\draw[black] (11.927083333333334em,-3.4895833333333335em) -- (11.927083333333334em,-2.34375em);
\draw[white] (25.23271753189503em,-3.43397249898819em) -- (28.644175865228366em,-3.43397249898819em);
\draw[white] (28.644175865228366em,-3.43397249898819em) -- (29.365219443030846em,-2.302923749494095em);
\draw[black] (25.23271753189503em,-3.43397249898819em) -- (28.644175865228366em,-3.43397249898819em);
\draw[black] (28.644175865228366em,-3.43397249898819em) -- (29.365219443030846em,-2.302923749494095em);
\node at (10.416666666666668em,-2.0833333333333335em) {$PF$} ;
\node at (10.416666666666668em,-6.25em) {$𝐒𝐞𝐭/\FamCat{D}$} ;
\node at (16.75925925925926em,-2.0833333333333335em) {$\FamCat{C}$} ;
\node at (16.75925925925926em,-6.25em) {$\FamCat{D}$} ;
\node at (10.416666666666668em,-10.416666666666668em) {$\FamCat{D}$} ;
\node at (20.520833333333336em,-2.0833333333333335em) {$PF$} ;
\node at (20.520833333333336em,-10.416666666666668em) {$\FamCat{D}$} ;
\node at (25.833333333333336em,-2.0833333333333335em) {$𝐒𝐞𝐭/\FamCat{D}$} ;
\node at (31.5625em,-2.0833333333333335em) {$𝐒𝐞𝐭^→$} ;
\node at (31.5625em,-10.416666666666668em) {$𝐒𝐞𝐭$} ;
\node[scale=0.7] at (18.108663118472048em,-4.166666666666667em) {$\FamCat{F}$} ;
\node[scale=0.7] at (9.977210649650214em,-8.346354166666668em) {$\pi_2 $} ;
\node[scale=0.7] at (32.182586468400125em,-6.25em) {$\cod$} ;
\node[scale=0.7] at (23.694205242536025em,-6.597977798496925em) {$ π₂$} ;
\end{tikzpicture}

%% file: diagrams/nat-universal-pullback-P-explicit.tex
\begin{tikzpicture}
\draw[white,-,line width=0.20833333333333334em] (10.9375em,-11.822916666666668em) .. controls (10.9375em,-12.274305555555555em) and (10.9375em,-12.725694444444445em) .. (10.9375em,-13.177083333333334em);
\draw[black,->, ] (10.9375em,-11.822916666666668em) .. controls (10.9375em,-12.274305555555555em) and (10.9375em,-12.725694444444445em) .. (10.9375em,-13.177083333333334em);
\draw[white,-,line width=0.20833333333333334em] (10.9375em,-14.947916666666668em) .. controls (10.9375em,-15.399305555555554em) and (10.9375em,-15.850694444444445em) .. (10.9375em,-16.302083333333336em);
\draw[black,->, ] (10.9375em,-14.947916666666668em) .. controls (10.9375em,-15.399305555555554em) and (10.9375em,-15.850694444444445em) .. (10.9375em,-16.302083333333336em);
\draw[white,-,line width=0.20833333333333334em] (10.9375em,-8.697916666666668em) .. controls (10.9375em,-9.149305555555555em) and (10.9375em,-9.600694444444445em) .. (10.9375em,-10.052083333333334em);
\draw[black,->, ] (10.9375em,-8.697916666666668em) .. controls (10.9375em,-9.149305555555555em) and (10.9375em,-9.600694444444445em) .. (10.9375em,-10.052083333333334em);
\draw[white,-,line width=0.20833333333333334em] (18.385416666666668em,-12.124999997516474em) .. controls (18.55034722222222em,-12.124999997516474em) and (18.71527777777778em,-12.124999997516474em) .. (18.880208333333336em,-12.124999997516474em);
\draw[black,->, ] (18.385416666666668em,-12.124999997516474em) .. controls (18.55034722222222em,-12.124999997516474em) and (18.71527777777778em,-12.124999997516474em) .. (18.880208333333336em,-12.124999997516474em);
\draw[white,-,line width=0.20833333333333334em] (20.208333333333336em,-9.88541666418314em) .. controls (20.208333333333332em,-10.336805553072029em) and (20.208333333333332em,-10.788194441960917em) .. (20.208333333333336em,-11.239583330849808em);
\draw[black,->, ] (20.208333333333336em,-9.88541666418314em) .. controls (20.208333333333332em,-10.336805553072029em) and (20.208333333333332em,-10.788194441960917em) .. (20.208333333333336em,-11.239583330849808em);
\draw[white,-,line width=0.20833333333333334em] (17.083333333333336em,-9.88541666418314em) .. controls (17.083333333333336em,-10.336805553072029em) and (17.083333333333336em,-10.788194441960917em) .. (17.083333333333336em,-11.239583330849808em);
\draw[black,->, ] (17.083333333333336em,-9.88541666418314em) .. controls (17.083333333333336em,-10.336805553072029em) and (17.083333333333336em,-10.788194441960917em) .. (17.083333333333336em,-11.239583330849808em);
\draw[white,-,line width=0.20833333333333334em] (18.177083333333336em,-8.999999997516474em) .. controls (18.45486111111111em,-8.999999997516474em) and (18.73263888888889em,-8.999999997516474em) .. (19.010416666666668em,-8.999999997516474em);
\draw[black,->, ] (18.177083333333336em,-8.999999997516474em) .. controls (18.45486111111111em,-8.999999997516474em) and (18.73263888888889em,-8.999999997516474em) .. (19.010416666666668em,-8.999999997516474em);
\draw[white,-,curve={ratio=0.4}, line width=0.20833333333333334em] (16.375em,-9.88541666418314em) .. controls (15.180555555555555em,-11.378472219738695em) and (15.180555555555555em,-12.87152777529425em) .. (16.375em,-14.364583330849808em);
\draw[black,->, dashed, curve={ratio=0.4}, ] (16.375em,-9.88541666418314em) .. controls (15.180555555555555em,-11.378472219738695em) and (15.180555555555555em,-12.87152777529425em) .. (16.375em,-14.364583330849808em);
\draw[white] (17.34375em,-10.145833330849808em) -- (18.4375em,-10.145833330849808em);
\draw[white] (18.4375em,-10.145833330849808em) -- (18.4375em,-9.26041666418314em);
\draw[black] (17.34375em,-10.145833330849808em) -- (18.4375em,-10.145833330849808em);
\draw[black] (18.4375em,-10.145833330849808em) -- (18.4375em,-9.26041666418314em);
\node at (17.083333333333336em,-8.999999997516474em) {$F\Gamma '$} ;
\node at (10.9375em,-7.8125em) {$F P Γ$} ;
\node at (10.9375em,-10.9375em) {$P_ℂ F Γ$} ;
\node at (10.9375em,-14.0625em) {$P_ℂ G Γ$} ;
\node at (10.9375em,-17.1875em) {$G P Γ$} ;
\node at (17.083333333333336em,-12.124999997516474em) {$FP Γ$} ;
\node at (20.208333333333336em,-12.124999997516474em) {$F\XGAT$} ;
\node at (20.208333333333336em,-8.999999997516474em) {$F\YGAT$} ;
\node at (17.083333333333336em,-15.249999997516474em) {$F Γ$} ;
\node[scale=0.7] at (11.862923512968147em,-12.5em) {$P_ℂ α_Γ$} ;
\node[scale=0.7] at (11.3036720941369em,-15.625em) {$≅$} ;
\node[scale=0.7] at (11.403010803611544em,-9.375em) {$ β_Γ$} ;
\node[scale=0.7] at (20.949752708452248em,-10.562499997516474em) {$F \pGAT$} ;
\node[scale=0.7] at (14.798718981997986em,-12.124999997516474em) {$F ε_Γ$} ;
\end{tikzpicture}

%% file: diagrams/mate-free-filtered.tex
\begin{tikzpicture}
\draw[white,-,line width=0.20833333333333334em] (15.885416666666668em,-2.8645833333333335em) .. controls (16.788194444444443em,-2.8645833333333335em) and (17.69097222222222em,-2.8645833333333335em) .. (18.59375em,-2.8645833333333335em);
\draw[black,->, ] (15.885416666666668em,-2.8645833333333335em) .. controls (16.788194444444443em,-2.8645833333333335em) and (17.69097222222222em,-2.8645833333333335em) .. (18.59375em,-2.8645833333333335em);
\draw[white,-,line width=0.20833333333333334em] (14.322916666666668em,-7.708333333333334em) .. controls (14.322916666666668em,-6.388888888888888em) and (14.322916666666668em,-5.069444444444445em) .. (14.322916666666668em,-3.75em);
\draw[black,->, ] (14.322916666666668em,-7.708333333333334em) .. controls (14.322916666666668em,-6.388888888888888em) and (14.322916666666668em,-5.069444444444445em) .. (14.322916666666668em,-3.75em);
\draw[white,-,line width=0.20833333333333334em] (15.833333333333334em,-8.59375em) .. controls (16.75347222222222em,-8.59375em) and (17.67361111111111em,-8.59375em) .. (18.59375em,-8.59375em);
\draw[black,->, ] (15.833333333333334em,-8.59375em) .. controls (16.75347222222222em,-8.59375em) and (17.67361111111111em,-8.59375em) .. (18.59375em,-8.59375em);
\draw[white,-,line width=0.20833333333333334em] (20.052083333333336em,-7.708333333333334em) .. controls (20.052083333333332em,-6.388888888888888em) and (20.052083333333332em,-5.069444444444445em) .. (20.052083333333336em,-3.75em);
\draw[black,->, ] (20.052083333333336em,-7.708333333333334em) .. controls (20.052083333333332em,-6.388888888888888em) and (20.052083333333332em,-5.069444444444445em) .. (20.052083333333336em,-3.75em);
\node at (20.052083333333336em,-8.59375em) {$ \overline{ℂ}/iY_ℂ$} ;
\node at (14.322916666666668em,-2.8645833333333335em) {$ 𝔻/X_𝔻$} ;
\node at (20.052083333333336em,-2.8645833333333335em) {$ 𝔻/Y_𝔻$} ;
\node at (14.322916666666668em,-8.59375em) {$ \overline{ℂ}/iX_ℂ$} ;
\node at (17.1875em,-5.8020830154418945em) {$\cong $} ;
\node[scale=0.7] at (17.239583333333336em,-2.335896063517366em) {$p_𝔻^*$} ;
\node[scale=0.7] at (13.956744572529768em,-5.729166666666667em) {$\overline{F}$} ;
\node[scale=0.7] at (17.213541666666668em,-9.121807014882435em) {$ip_ℂ^*$} ;
\node[scale=0.7] at (20.418255427470232em,-5.729166666666667em) {$\overline{F}$} ;
\end{tikzpicture}